\newif\ifmain
\newif\ifappendix
\newif\iffull
\newif\ifabstract

\maintrue
\abstractfalse \fulltrue \appendixtrue

\ifabstract
\documentclass[a4paper,USenglish]{lipics}
\usepackage{microtype}
\fi
\iffull
\documentclass[letter,11pt]{article}
\usepackage{fullpage}
\fi

\usepackage{url}
\usepackage[utf8]{inputenc}
\usepackage{amssymb,amsmath}
\iffull
\usepackage{amsthm}
\fi
\usepackage{tabularx}
\usepackage{graphicx}
\usepackage{tikz}
\usepackage{xspace}
\usepackage{comment}
\usepackage{enumitem}
\usepackage{graphics}
\usepackage{color}
\usepackage{hyperref}

\newcommand{\executeiffilenewer}[3]{%
\ifnum\pdfstrcmp{\pdffilemoddate{#1}}%
{\pdffilemoddate{#2}}>0%
{\immediate\write18{#3}}\fi%
} 

\newcommand{\svg}[2]{\def\svgwidth{#1}%
\executeiffilenewer{#2.svg}{#2.pdf}%
{inkscape -z -D --file=#2.svg %
--export-pdf=#2.pdf --export-latex}%
{\input{#2.pdf_tex}}}

\newcommand{\etal}{\emph{et al.}}

\newcounter{nlistcounter}

\newcommand{\probsc}[1]{{\sc #1}}
 
\newcommand{\probclique}{\probsc{Clique}}
\newcommand{\probnonrep}{\probsc{Constrained Closed Walk}} 
\newcommand{\probATSP}{\probsc{ATSP}} 
\newcommand{\probbalanced}{\probsc{Edge Balancing}}
\newcommand{\probmcb}{\probsc{Multicolored Biclique}}

\newcommand{\typeset}{\mathcal{T}}

\iffull
\newtheorem{theorem}{Theorem}[section]
\newtheorem{lemma}[theorem]{Lemma}

\newtheorem{definition}[theorem]{Definition}

\fi
\newtheorem{proposition}[theorem]{Proposition}

\newcommand{\OPT}{\mathsf{OPT}}
\newcommand{\cost}{\mathsf{cost}}
\newcommand{\nil}{\mathsf{nil}}
\newcommand{\indeg}{\mathsf{in\text{-}degree}}
\newcommand{\outdeg}{\mathsf{out\text{-}degree}}
\newcommand{\myleft}{\mathsf{left}}
\newcommand{\myright}{\mathsf{right}}
\newcommand{\myin}{\mathsf{in}}
\newcommand{\myout}{\mathsf{out}}
\newcommand{\eg}{\mathsf{eg}}
\newcommand{\eps}{\varepsilon}
\newcommand{\vG}{\vec{G}}
\newcommand{\vW}{\vec{W}}
\newcommand{\vH}{\vec{H}}
\newcommand{\vZ}{\vec{Z}}
\newcommand{\vR}{\vec{R}}

\newcommand{\vF}{\vec{F}}
\newcommand{\vS}{\vec{S}}
\newcommand{\mya}{\textbf{a}}
\newcommand{\cT}{{\cal T}}
\newcommand{\cP}{{\cal P}}

\title{Constant-factor approximations for asymmetric TSP
on nearly-embeddable graphs}
\ifabstract
\titlerunning{Constant-factor approximations for ATSP
on nearly-embeddable graphs}
\fi

\ifabstract
\author[1]{D\'aniel Marx}
\author[2]{Ario Salmasi}
\author[3]{Anastasios Sidiropoulos}
\affil[1]{Institute of Computer Science and Control,
 Hungarian Academy of Sciences (MTA SZTAKI), Budapest, Hungary.\\
  \texttt{dmarx@cs.bme.hu}}
\affil[2]{Dept.~of Computer Science and Engineering, The Ohio State University\\
Columbus OH, USA\\
\texttt{salmasi.1@osu.edu}}
\affil[3]{Dept.~of Computer Science and Engineering and Dept.~of Mathematics, The Ohio State University\\
  Columbus OH, USA\\
  \texttt{sidiropoulos.1@osu.edu}}
\authorrunning{D. Marx, A. Salmasi and A. Sidiropoulos} 

\Copyright{D\'aniel Marx, Ario Salmasi and Anastasios Sidiropoulos}


\serieslogo{}
\volumeinfo
  {Billy Editor and Bill Editors}
  {2}
  {Conference title on which this volume is based on}
  {1}
  {1}
  {1}
\EventShortName{}
\DOI{10.4230/LIPIcs.xxx.yyy.p}

\fi

\iffull
\author{
D\'{a}niel Marx
\footnote{Institute of Computer Science and Control,
 Hungarian Academy of Sciences (MTA SZTAKI), Budapest, Hungary; \texttt{dmarx@cs.bme.hu}.
 Research supported by the European Research Council (ERC)  grant ``PARAMTIGHT: Parameterized complexity and the search for tight complexity results,'' reference 280152 and OTKA grant NK105645.}
\and
Ario Salmasi
\footnote{Dept.~of Computer Science and Engineering, The Ohio State University, Columbus OH, USA; \texttt{salmasi.1@osu.edu}.
Supported by NSF grant CCF 1423230 and CAREER 1453472.}
\and
Anastasios Sidiropoulos
\footnote{Dept.~of Computer Science and Engineering and Dept.~of Mathematics, The Ohio State University, Columbus OH, USA, \texttt{sidiropoulos.1@osu.edu}.
Supported by NSF grant CCF 1423230 and CAREER 1453472.}
}
\fi

\begin{document}

\maketitle

\begin{abstract}
In the Asymmetric Traveling Salesperson Problem (ATSP) the goal is to find a closed walk of minimum cost in a directed graph visiting every vertex.
We consider the approximability of ATSP on topologically restricted graphs.
It has been shown by Oveis Gharan and Saberi \cite{gharan2011asymmetric} that there exists polynomial-time constant-factor approximations on planar graphs and more generally graphs of constant orientable genus.
This result was extended to non-orientable genus by Erickson and Sidiropoulos \cite{erickson2014near}.

We show that for any class of \emph{nearly-embeddable} graphs, ATSP admits a polynomial-time constant-factor approximation.
More precisely, we show that for any fixed $k\geq 0$, there exist $\alpha, \beta>0$, such that ATSP on $n$-vertex $k$-nearly-embeddable graphs admits a $\alpha$-approximation in time $O(n^\beta)$.
The class of $k$-nearly-embeddable graphs contains graphs with at most $k$ apices, $k$ vortices of width at most $k$, and an underlying surface of either orientable or non-orientable genus at most $k$.
Prior to our work, even the case of graphs with a single apex was open.
Our algorithm combines tools from rounding the Held-Karp LP via thin trees with dynamic programming.

We complement our upper bounds by showing that solving ATSP exactly on graphs of pathwidth $k$ (and hence on $k$-nearly embeddable graphs) requires time $n^{\Omega(k)}$, assuming the Exponential-Time Hypothesis (ETH).
This is surprising in light of the fact that both TSP on undirected graphs and Minimum Cost Hamiltonian Cycle on directed graphs are FPT parameterized by treewidth.
\end{abstract}

\iffull
\newpage
\fi

\section{Introduction}

An instance of the Asymmetric Traveling Salesman Problem (ATSP) consists of a directed graph $\vG$ and a (not necessarily symmetric) cost function $c:E(\vG) \to \mathbb{R}^{+}$.
The goal is to find a spanning closed walk of $\vG$ with minimum total cost.
This is one of the most well-studied NP-hard problems.

Asadpour \etal~\cite{asadpour2010log} obtained a polynomial-time $O(\log n /\log\log n)$-approximation algorithm for ATSP, which was  the first asymptotic improvement in almost 30 years \cite{DBLP:journals/networks/FriezeGM82,blaser2008new,feige2007improved,kaplan2005approximation}.
Building on their techniques, Oveis Gharan and Saberi \cite{gharan2011asymmetric} described a polynomial-time $O(\sqrt{g} \log g)$-approximation algorithm when the input includes an embedding of the input graph into an orientable surface of genus~$g$.
Erickson and Sidiropoulos \cite{erickson2014near} improved the dependence on the genus by obtaining a $O(\log g/\log\log g)$-approximation.

Anari and Oveis Gharan \cite{gharan_spectrally} have recently shown that the integrality gap of the natural linear programming relaxation of ATSP proposed by Held and Karp \cite{held1970traveling} is $\log\log^{O(1)} n$. 
This implies a polynomial-time $\log\log^{O(1)}n$-approximation algorithm for the \emph{value} of ATSP.
We remark that the best known lower bound on the integrality gap of the Held-Karp LP is $2$ \cite{charikar2004integrality}.
Obtaining a polynomial-time constant-factor approximation algorithm for ATSP is a central open problem in optimization.

\subsection{Our contribution}
We study the approximability of ATSP on topologically restricted graphs.
Prior to our work, a constant-factor approximation algorithm was known only for graphs of bounded genus.
We significantly extend this result by showing that there exist a polynomial-time constant-factor approximation algorithm for ATSP on nearly embeddable graphs.
These graphs include graphs with bounded genus, with a bounded number of apices and a bounded number of vortices of bounded pathwidth.
For any $a,g,k,p\geq 0$, we say that a graph is $(a,g,k,p)$-nearly embeddable if it is obtained from a graph of Euler genus $g$ by adding $a$ apices and $k$ vortices of pathwidth $p$ (see \cite{lovasz2006graph,kawarabayashi2007some,diestel2000graph} for more precise definitions).
The following summarizes our result.

\begin{theorem}\label{thm:main}
Let $a,g,k\geq 0$, $p\geq 1$.
There exists a $O(a (g+k+1) + p^2)$-approximation algorithm for ATSP on $(a,g,k,p)$-nearly embeddable digraphs, with running time $n^{O((a+p)(g+k+1)^4)}$.
\end{theorem}

The above algorithm is obtained via a new technique that combines the Held-Karp LP with a dynamic program that solves the problem on  vortices.
We remark that it is not known whether the integrality gap of the LP is constant for graphs of constant pathwidth.


We complement this result by showing that solving ATSP exactly on graphs of pathwidth $p$ (and hence on $p$-nearly embeddable graphs) requires time $n^{\Omega(p)}$, assuming the Exponential-Time Hypothesis (ETH).
This is surprising in light of the fact that both TSP on undirected graphs and Minimum Cost Hamiltonian Cycle on directed graphs are FPT parameterized by treewidth.
The following summarizes our lower bound.

\begin{theorem}\label{thm:lower}
Assuming ETH, there is no $f(p)n^{o(p)}$ time algorithm for ATSP on graphs of pathwidth at most $p$ for any computable function $f$.
\end{theorem}

\subsection{Overview of the algorithm}

We now give a high level overview of the main steps of the algorithm and highlight some of the main challenges.

\begin{description}
\item{\textbf{Step 1: Reducing the number of vortices.}}
We first reduce the problem to the case of nearly embeddable graphs with a single vortex.
This is done by iteratively merging pairs of vortices.
We can merge two vortices by adding a new handle on the underlying surface-embedded graph.
For the remainder we will focus on the case of graphs with a single vortex.

\item{\textbf{Step 2: Traversing a vortex.}}
We obtain an exact polynomial-time algorithm for computing a closed walk that visits all the vertices in the vortex.
We remark that this subsumes as a special case the problem of visiting all the vertices in a single face of a planar graph, which was open prior to our work.

Let us first consider the case of a vortex in a planar graph.
Let $\vW$ be an optimal walk that visits all the vertices in the vortex.
Let $F$ be the face on which the vortex is attached.
We give a dynamic program that maintains a set of partial solutions for each subpath of $F$.
We prove correctness of the algorithm by establishing structural properties of $\vW$.
The main technical difficulty is that $\vW$ might be self-crossing.
We first decompose $\vW$ into a collection ${\cal W}$ of non-crossing walks.
We form a conflict graph ${\cal I}$ of ${\cal W}$ and consider a spanning forest ${\cal F}$ of ${\cal I}$.
This allows us to prove correctness via induction on the trees of ${\cal F}$.

The above algorithm can be extended to graphs of bounded genus.
The main difference is that the dynamic program now computes a set of partial solutions for each bounded collection of subpaths of $F$.

Finally, the algorithm is extended to the case of nearly-embeddable graphs by adding the apices to the vortex without changing the cost of the optimum walk.

\item{\textbf{Step 3: Finding a thin forest in the absence of  vortices.}}
The constant-factor approximation for graphs of bounded genus was obtained by showing by constructing thin forests with a bounded number of components in these graphs \cite{gharan2011asymmetric,erickson2014near}.
We extend this result by constructing thin forests with a bounded number of components in graphs of bounded genus and with a bounded number of additional apices.
Prior to our work even the case of planar graphs with a single apex was open; in fact, no constant-factor approximation algorithm was known for these graphs.


\item{\textbf{Step 4: Combining the Held-Karp LP with the dynamic program.}}
We next combine the dynamic program with the thin forest construction.
We first compute an optimal walk $\vW$ visiting all the vertices in the vortex, and we contract the vortex into a single vertex.
A natural approach would be to compute a thin forest in the contracted graph.
Unfortunately this fails because such a forest might not be thin in the original graph.
In order to overcome this obstacle we change the feasible solution of the Held-Karp LP by taking into account $\vW$, and we modify the forest construction so that it outputs a subgraph that is thin with respect to this new feasible solution.

\item{\textbf{Step 5: Rounding the forest into a walk.}}
Once we have a thin spanning subgraph of $G$ we can compute a solution to ATSP via circulations, as in previous work.

\end{description}

\subsection{Organization}
The rest of the paper is organized as follows.
Section \ref{sec:notation} introduces some basic notation.
Section \ref{sec:LP} defines the Held-Karp LP for ATSP.
Section \ref{sec:approx} presents the main algorithm, using the dynamic program and the thin forest construction as a black box.
Section \ref{sec:constant-constant-thin} presents the technique for combining the dynamic program with the Held-Karp LP.
Section \ref{sec:thin_1-apex} gives the algorithm for computing a thin tree in a $1$-apex graph.
This algorithm is generalized to graphs with a bounded number of apices in Section \ref{sec:thin_a-apex}, and to graphs of bounded genus and with a bounded number of apices in Section \ref{sec:thin_forests_higher_genus}.
In Section \ref{sec:main_Lemma_(a,g,1,p)-case} we show how to modify the thin forest construction so that we can compute a spanning thin subgraph in a nearly-embeddable graph, using the solution of the dynamic program.

The dynamic program is given in Sections \ref{sec:normalization}, \ref{sec:uncrossing}, \ref{sec:vortex_planar}, \ref{sec:vortex_genus}, and \ref{sec:vortex_nearly}.
More precisely, Section \ref{sec:normalization} introduces a certain preprocessing step.
Section \ref{sec:uncrossing} establishes a structural property of the optimal solution.
Section \ref{sec:vortex_planar} presents the dynamic program for a vortex in a planar graph.
Sections \ref{sec:vortex_genus} and \ref{sec:vortex_nearly} generalize this dynamic program to graphs of bounded genus and with a bounded number of apices respectively.

Finally, Section \ref{sec:lower} presents the lower bound.

\section{Notation}
\label{sec:notation}


In this section we introduce some basic notation that will be used throughout the paper.

\textbf{Graphs.}
Unless otherwise specified, we will assume that for every pair of vertices in a graph there exists a unique shortest path; this property can always achieved by breaking ties between different shortest paths in a consistent manner (e.g.~lexicographically).
Moreover for every edge of a graph (either directed or undirected) we will assume that its length is equal to the shortest path distance between its endpoints.
Let $\vG$ be some digraph.
Let $G$ be the undirected graph obtained from $\vG$ by ignoring the directions of the edges,
that is 
$V(G)=V(\vG)$ and $E(G)=\{\{u,v\}:(u,v) \in E(\vG) \text{ or } (v,u)\in E(\vG)\}$.
We say that $G$ is the \emph{symmetrization} of $\vG$.
For some $x:E(\vG)\to \mathbb{R}$ we define
$\cost_{\vG}(x) = \sum_{(u,v)\in E(\vG)} x((u,v))\cdot d_{\vG}(u,v)$.
For a subgraph $S\subseteq G$ we define
$\cost_G(S) = \sum_{e\in E(S)} c(e)$.
Let $z$ be a weight function on the edges of $G$. For any $A, B \subseteq V(G)$ we define $z(A, B) = \sum_{a \in A, b \in B} z(\{a,b\})$.

\textbf{Asymmetric TSP.}
Let $\vG$ be a directed graph with non-negative arc costs.
For each arc $(u,v)\in E(\vG)$ we denote the cost of $(u,v)$ by $c(u,v)$.
A \emph{tour} in $\vG$ is a closed walk in $\vG$.
The \emph{cost} of a tour $\tau=v_1,v_2,\ldots,v_k,v_1$ is defined to be 
$\cost_{\vG}(\tau) = d_{\vG}(v_k,v_1) + \sum_{i=1}^{k-1} d_{\vG}(v_i,v_{i+1})$.
Similarly the cost of an open walk $W=v_1,\ldots,v_k$ is defined to be $\cost_{\vG}(W) = \sum_{i=1}^{k-1} d_{\vG}(v_i, v_{i+1})$.
The cost of a collection ${\cal W}$ of walks is defined to be $\cost_{\vG}({\cal W}) = \sum_{W\in {\cal W}} \cost_{\vG}(W)$.
We denote by $\OPT_{\vG}$ the minimum cost of a tour traversing all vertices in $\vG$.
For some $U\subseteq V(\vG)$ we denote by $\OPT_{\vG}(U)$ the minimum cost of a tour in $\vG$ that visits all vertices in $U$.

\section{The Held-Karp LP}\label{sec:LP}

We recall the Held-Karp LP for ATSP \cite{hk-tspmst-70}.  Fix a directed graph $\vG$ and a cost function $c: E(\vG) \to \mathbb{R}^+$.  For any subset $U\subseteq V$, we define
\begin{align*}
	\delta_{\vG}^+(U)
		&:= \{(u,v) \in E(\vG) : u\in U \text{~and~} v\notin U\}
	\\
	\text{and}\quad
	\delta_{\vG}^-(U)
		&:= \delta_{\vG}^+(V\setminus U).
\end{align*}
We omit the subscript $\vG$ when the underlying graph is clear from context.
We also write $\delta^+(v) = \delta^+(\{v\})$ and $\delta^-(v) = \delta^-(\{v\})$ for any single vertex $v$. 

Let $G$ be the symmetrization of $\vG$.
For any $U\subseteq V(G)$, we define
\[
	\delta_{G}(U) := \{\{u,v\} \in E(G) : u\in U \text{~and~}  v\notin U\}.
\]
Again, we omit the subscript $G$ when the underlying graph is clear from  context.  We also extend the cost function $c$ to undirected edges by defining 
\[
	c(\{u,v\}) := \min\{ c((u,v)), c((v,u)) \}.
\]
For any function $x\colon E(\vG)\to \mathbb{R}$ and any subset $W\subseteq E(\vG)$, we write $x(W) = \sum_{a\in W} x(a)$.  With this notation, the Held-Karp LP relaxation is defined as follows.
\begin{center}
	\fbox{$
	\begin{array}{r@{\,}r@{\quad}l}
	\text{minimize} &
		\sum_{a\in A} c(a) \cdot x(a)&
\\[2ex]
	\text{subject to} &
		x(\delta^+(U)) \geq 1
			& \text{for all nonempty~} U\subsetneq V(\vG)
\\[0.5ex]&
		x(\delta^+(v)) = x(\delta^-(v))
			& \text{for all~} v\in V(\vG)
\\[0.5ex]&
		x(a) \geq 0
			& \text{for all~} a\in E(\vG)
	\end{array}
	$}
\end{center}
We define the \emph{symmetrization} of~$x$ as the function $z\colon E(G) \to \mathbb{R}$ where
\[
	z(\{u,v\}) := x((u,v)) + x((v,u))
\]
for every edge $\{u,v\} \in E(G)$.
For any subset $W\subseteq E(G)$ of edges, we write $z(W) := \sum_{e\in W} z(e)$.
Let $\vW \subseteq \vG$.
Let $\alpha,s>0$.
We say that $\vW$ is \emph{$\alpha$-thin} (w.r.t.~$z$) if for all $U\subseteq V$ we have
\[
|E(\vW)\cap \delta(U)| \leq \alpha \cdot z(\delta(U)).
\]
We also say that $\vW$ is \emph{$(\alpha,s)$-thin} (w.r.t.~$x$) if $\vW$ is $\alpha$-thin (w.r.t.~$z$) and 
\[
c(E(\vW)) \leq s\cdot \sum_{e\in E(\vG)} c(e) \cdot x(e).
\]
We say that $z$ is \emph{$\vec{W}$-dense} if for all $(u,v)\in E(\vec{W})$ we have $z(\{u,v\})\geq 1$.
We say that $z$ is \emph{$\eps$-thick} if for all $U\subsetneq V(G)$ with $U\neq \emptyset$ we have $z(\delta(U)) \geq \eps$.


\section{An approximation algorithm for nearly-embeddable graphs}
\label{sec:approx}

The following Lemma is implicit in the work of Erickson and Sidiropoulos \cite{erickson2014near} (see also \cite{asadpour2010approximation}).
\begin{lemma}\label{lem:from_thin_trees_to_walks}
Let $\vec{G}$ be a digraph and let $x$ be a feasible solution for the Held-Karp LP for $\vec{G}$. Let $\alpha, s > 0$, and let $S$ be a $(\alpha, s)$-thin spanning subgraph of $G$ (w.r.t.~$x$), with at most $k$ connected components. Then, there exists a polynomial-time algorithm which computes a collection of closed walks $C_1, \ldots, C_{k'}$, for some $k' \leq k$, such that their union visits all the vertices in $V(\vec{G})$, and such that $\sum_{i=1}^{k} \cost_{\vec{G}}(C_i) \leq (2\alpha + s )\sum_{e \in E(\vec{G})} c(e)\cdot x(e)$.
\end{lemma}

The following is the main technical Lemma that combines a solution to the Held-Karp LP with a walk traversing the vortex that is computed via the dynamic program.
The proof of Lemma \ref{lem:constant-constant-thin} is deferred to Section \ref{sec:constant-constant-thin}.

\begin{lemma}\label{lem:constant-constant-thin}
Let $a,g,p>0$, let $\vec{G}$ be a $(a,g,1,p)$-nearly embeddable graph, and let $G$ be its symmetrization.
There exists an algorithm with running time $n^{O((a+p)g^4)}$ which computes a feasible solution $x$ for the Held-Karp LP for $\vec{G}$ with cost $O(\OPT_{\vec{G}})$ and a spanning subgraph $S$ of $G$ with at most $O(a+g)$ connected components, such that $S$ is $(O(a\cdot g + p^2), O(1))$-thin w.r.t.~$x$.
\end{lemma}

Using Lemma \ref{lem:constant-constant-thin} we are now ready to obtain an approximation algorithm for nearly-embeddable graphs with a single vortex.

\begin{theorem}\label{thm:one_vortex}
Let $a,g\geq 0$, $p\geq 1$.
There exists a $O(a\cdot g + p^2)$-approximation algorithm for ATSP on $(a,g,1,p)$-nearly embeddable digraphs, with running time $n^{O((a+p)(g+1)^4)}$.
\end{theorem}

\begin{proof}
We follow a similar approach to \cite{erickson2014near}.
The only difference is that in \cite{erickson2014near} the algorithm uses an optimal solution to the Held-Karp LP.
In contrast, here we use a feasible solution that is obtained by Lemma \ref{lem:constant-constant-thin}, together with an appropriate thin subgraph.

Let $\vec{G}$ be $(a,g,1,p)$-nearly embeddable digraph. By using Lemma \ref{lem:constant-constant-thin}, we find in time $n^{O((a+p)g^4)}$ a feasible solution $x$ for the Held-Karp LP for $\vec{G}$ with cost $O(\OPT_{\vec{G}})$ and a spanning subgraph $S$ of $G$ with at most $O(a+g)$ connected components, such that $S$ is $(O(a\cdot g + p^2), O(1))$-thin w.r.t.~$x$. Now we compute in polynomial time a collection of closed walks $C_1, \ldots, C_{k'}$, for some $k' \in O(a+g)$, that visit all the vertices in $V(\vec{G})$, and such that the total cost of all walks is at most $O((a\cdot g + p^2)\cdot OPT_{\vec{G}})$, using Lemma \ref{lem:from_thin_trees_to_walks}. For every $i \in \{1, \ldots, k'\}$, let $v_i \in V(\vec{G})$ be an arbitrary vertex visited by $C_i$. We construct a new instance $(\vec{G}' , c')$ of ATSP as follows. Let $V(\vec{G}') = \{v_1, \ldots, v_{k'}\}$. For any $u,v \in V(\vec{G}')$, we have an edge $(u,v)$ in $E(\vec{G}')$, with $c'(u,v)$ being the shortest-path distance between $u$ and $v$ in $G$ with edge weights given by $c$.
By construction we have $\OPT_{\vec{G}'} \leq \OPT_{\vec{G}}$. We find a closed tour $C$ in $\vec{G}'$ with $\cost_{\vec{G}'} (C) = \OPT_{\vec{G}'}$ in time $2^{O(|V(\vec{G}')|)} \cdot n^{O(1)} = 2^{O(a + g)} \cdot n^{O(1)}$. By composing $C$ with the $k'$ closed walks $C_1, \ldots, C_{k'}$, and shortcutting as in \cite{frieze1982worst}, we obtain a solution for the original instance, of total cost $O(a\cdot g + p^2)\cdot \OPT_{\vec{G}}$.
\end{proof}

We are now ready to prove the main algorithmic result of this paper.

\begin{proof}[Proof of Theorem \ref{thm:main}]
We may assume $k\geq 2$ since otherwise the assertion follows by Theorem \ref{thm:one_vortex}.
We may also assume w.l.o.g.~that $p\geq 2$.
Let $\vG$ be a $(a,g,k,p)$-nearly embeddable digraph.
It suffices to show that there exists a polynomial time computable $(a,g+k-1,1,2p)$-nearly embeddable digraph $\vG'$ with $V(\vG')=V(\vG)$ such that for all $u,v\in V(\vG)$ we have $d_{\vG}(u,v)=d_{\vG'}(u,v)$.
We compute $\vG'$ as follows.
Let $\vH_1,\ldots,\vH_k$ be the vortices of $\vG$ and let $\vF_1,\ldots,\vF_{k}$ be the faces on which they are attached.
For each $i\in \{1,\ldots,k\}$ pick distinct $e_i,f_i\in E(\vF_i)$, with $e_i=\{w_i,w_i'\}$, $f_i=\{z_i,z_i'\}$.
There exists a path decomposition $B_{i,1},\ldots,B_{i,\ell_i}$ of $\vH_i$, of width at most $2p$, and such that $B_{i,1}=e_i$, and $B_{i,\ell_i}=f_i$.
For each $i\in \{1,\ldots,k-1\}$, we add edges $(w_i,z_i)$, $(z_i,w_i)$, $(w_i',z_i')$, and $(z_i',w_i')$ to $\vG'$, and we set their length to be equal to the shortest path distance between their endpoints in $\vG$.
We also add a handle connecting punctures in the disks bounded by $\vF_i$ and $\vF_{i+1}$ respectively, and we route the four new edges along this handle.
Since we add $k-1$ handles in total the Euler genus of the underlying surface increases by at most $k-1$.
We let $\vH$ be the single vortex in $\vG'$ with $V(\vH)=\bigcup_{i=1}^k V(\vH_i)$ and
$E(\vH) = \left(\bigcup_{i=1}^k E(\vH_i)\right) \cup \left(\bigcup_{i=1}^{k-1} \{(w_i,z_i), (z_i,w_i), (w_i',z_i'), (z_i',w_i')\}\right)$.
It is immediate that 
\[
B_{1,1},\ldots,B_{1,\ell_1},\{f_1,e_2\},B_{2,1},\ldots,B_{2,\ell_2},\{f_2,e_3\},\ldots,B_{k,1}, \ldots ,B_{k,\ell_k}
\]
is a path decomposition of $\vH$ of width at most $2p$.
Thus $\vG'$ is $(a,g+k-1,1,2p)$-nearly embeddable, which concludes the proof.
\end{proof}

\section{Combining the Held-Karp LP with the dynamic program}
\label{sec:constant-constant-thin}

In this Section we show how to combine the dynamic program that finds an optimal closed walk traversing all the vertices in a vortex, with the Held-Karp LP.
The following summarizes our exact algorithm for traversing the vortex in a nearly-embeddable graph.
The proof of Theorem \ref{thm:vortex_nearly} is deferred to Section \ref{sec:vortex_nearly}.

\begin{theorem}\label{thm:vortex_nearly}
Let $\vG$ be an $n$-vertex $(a, g, 1, p)$-nearly embeddable graph and let $\vH$ be the single vortex of $\vG$.
Then there exists an algorithm which computes a walk $\vW$ visiting all vertices in $V(\vH)$ of total length at most $\OPT_{\vG}(V(\vH))$ in time $n^{O((a+p)g^4)}$.
\end{theorem}

\begin{definition}[$\vec{W}$-augmentation]
Let $\vec{G}$ be a directed graph.
Let $x:E(\vec{G})\to \mathbb{R}$ and let $\vec{W}\subseteq \vec{G}$.
We define the \emph{$\vec{W}$-augmentation} of $x$ to be the function $x':E(\vec{G})\to \mathbb{R}$ such that for all $e\in E(\vec{G})$ we have
\[
x'(e) = \left\{\begin{array}{ll}
x(e) + 1 & \text{ if } e\in E(\vec{W})\\
x(e) & \text{ otherwise}\end{array}\right.
\]
\end{definition}

The following summarizes the main technical result for computing a thin spanning subgraph in a nearly embeddable graph.
The proof of Lemma \ref{lem:main_Lemma_(a,g,1,p)-case} is deferred to Section \ref{sec:main_Lemma_(a,g,1,p)-case}.

\begin{lemma}\label{lem:main_Lemma_(a,g,1,p)-case}
Let $\vec{G}$ be a $(a,g,1,p)$-nearly embeddable digraph, let $\vec{H}$ be its vortex, and let $\vec{W}$ be a walk in $\vec{G}$ visiting all vertices in $V(\vec{H})$.
Let $G$, $H$, and $W$ be the symmetrizations of $\vec{G}$, $\vec{H}$, and $\vec{W}$ respectively.
Let $z:E(G)\to \mathbb{R}_{\geq 0}$ be $\alpha$-thick for some $\alpha\geq 2$, and $\vec{W}$-dense.
Then there exists a polynomial time algorithm which given $\vec{G}$, $\vec{H}$, $A$, $\vec{W}$, $z$, and an embedding of $\vec{G}\setminus (A \cup \vec{H})$ into a surface of genus $g$, outputs a subgraph $S\subseteq G\setminus H$, satisfying the following conditions:
\begin{description}
\item{(1)}
$W\cup S$ is a spanning subgraph of $G$ and has $O(a + g)$ connected components.

\item{(2)}
$W\cup S$ is $O(a\cdot g + p^2)$-thin w.r.t.~$z$.
\end{description}
\end{lemma}

We are now ready to prove the main result of this section.

\begin{proof}[Proof of Lemma \ref{lem:constant-constant-thin}]
Let $\vec{H}$ be the single vortex of $\vec{G}$. 
We compute an optimal solution $y:E(\vec{G}) \to \mathbb{R}$ for the Held-Karp LP for $\vec{G}$.
We find a tour $\vec{W}$ in $\vec{G}$ visiting all vertices in $V(\vec{H})$, with $\cost_{\vec{G}}(\vec{W}) = O(\OPT_{\vec{G}})$ using Theorem \ref{thm:vortex_nearly}.
Let $x:E(\vec{G})\to \mathbb{R}$ be the $\vec{W}$-augmentation of $y$.
Since for all $e\in E(\vec{G})$ we have $x(e)\geq y(e)$, it follows that $x$ is a feasible solution for the Held-Karp LP.
Moreover since $\cost_{\vec{G}}(\vec{W}) = O(\OPT_{\vec{G}})$, we obtain that $\cost_{\vec{G}}(x) = \cost_{\vec{G}}(y) + \cost_{G}(\vec{W}) = O(\OPT_{\vec{G}})$.
Let $z$ be the symmetrization of $x$.

Note that $z$ is $2$-thick and $\vec{W}$-dense. Therefore, by Lemma \ref{lem:main_Lemma_(a,g,1,p)-case} we can find a subgraph $S \subseteq G \setminus H$ such that $T = W \cup S$ is a $O(a \cdot g + p^2)$-thin spanning subgraph of $G$ (w.r.t.~$z$), with at most $O(g+a)$ connected components. Therefore, there exists a constant $\alpha$ such that for every $U \subseteq V(G)$ we have $|T \cap \delta(U)| \leq \alpha \cdot (a\cdot g + p^2) \cdot z(\delta(U))$. We can assume that $\alpha \geq 1$. Now we follow a similar approach to \cite{erickson2014near}.

Let $m = {\left \lfloor n^2 / \alpha \right \rfloor}$. We define a sequence of functions $z_0, \ldots, z_m$, and a sequence of spanning forests $T_1, \ldots, T_m$ satisfying the following conditions.

\begin{description}
\item{(1)}
For any $i \in \{0, \ldots, m\}$, $z_i$ is non-negative, $2$-thick and $\vec{W}$-dense.

\item{(2)}
For any $i \in \{1, \ldots, m\}$, $T_i$ has at most $O(a+g)$ connected components. 

\item{(3)}
For every $U \subseteq V(G)$ we have $|T_{i+1} \cap \delta(U)| \leq \alpha \cdot (a\cdot g + p^2) \cdot z_i(\delta(U))$.
\end{description}

We set $z_0 = 3 {\left \lfloor z n^2 \right \rfloor} / n^2$. Now suppose for $i \in \{0, \ldots, m-1\}$ we have defined $z_i$. We define $z_{i+1}$ and $T_{i+1}$ as follows. We apply Lemma \ref{lem:main_Lemma_(a,g,1,p)-case} and we obtain a subgraph $T_{i+1}$ of $G$ with at most $O(a + g)$ connected components such that for every $U \subseteq V(G)$ we have $|T_{i+1} \cap \delta(U)| \leq \alpha \cdot (a\cdot g + p^2) \cdot z_i(\delta(U))$. Also, for every $e \in E(G)$ we set $z_{i+1} (e) = z_i (e)$ if $e \not\in T_{i+1}$, and $z_{i+1} (e) = z_i (e) - 1/n^2$ if $e \in T_{i+1}$. Now by using the same argument as in \cite{erickson2014near}, we obtain that $z_{i+1}$ is non-negative and $2$-thick.
By the construction, we know that $z$ is $\vec{W}$-dense and thus for all $(u,v)\in \vec{W}$ we have $z_0(\{u,v\})\geq 3$. 
Note that for all $e\in E(G)$ we have $z_{i+1}(e)\geq z_i(e)-1/n^2$.
Thus for all $e\in E(G)$ and for all $i\in \{0,\ldots,m\}$ we have $z_i(e)\geq 2$.
Therefore for all $i\in \{0,\ldots,m\}$ we have that $z_i$ is $\vec{W}$-dense.

Now, similar to \cite{erickson2014near} we set the desired $S$ to be the subgraph $T_i$ that minimizes $\cost_G(T_i)$, which implies that $S$ is a $(O(a\cdot g + p^2), O(1))$-thin spanning subgraph with at most $O(a + g)$ connected components.
\end{proof}

\ifmain

\section{Thin trees in $1$-apex graphs}
\label{sec:thin_1-apex}

The following is implicit in the work of Oveis Gharan and Saberi \cite{gharan2011asymmetric}.

\begin{theorem}\label{thm:thin_tree_planar}
If $G$ is a planar graph and $z$ is an $\alpha$-thick weight function on the edges of $G$ for some $\alpha>0$, then there exists a $10/\alpha$-thin spanning tree in $G$ w.r.t.~$z$.
\end{theorem}

For the remainder of this section, let $G$ be an $1$-apex graph with planar part $\Gamma$ and apex $\mya$. Let $z$ be a $2$-thick weight function on the edges of $G$. We will find a $O(1)$-thin spanning tree in $G$ (w.r.t.~$z$). We describe an algorithm for finding such a tree in polynomial time. The algorithm proceeds in five phases.

\textbf{Phase 1.}
We say that a cut $U$ is \emph{tiny}~(w.r.t.~$z$) if $z(\delta(U)) < 0.1$.
We start with $\Gamma$ and we proceed to partition it via tiny cuts. Each time we find a tiny cut $U$, we partition the remaining graph by deleting all edges crossing $U$. This process will stop in at most $n$ steps. Let $\Gamma'$ be the resulting subgraph of $\Gamma$ where $V(\Gamma') = V(\Gamma)$ and $E(\Gamma') \subset E(\Gamma)$.

\textbf{Phase 2.}
By the construction, we know that there is no tiny cuts in each connected component of $\Gamma'$. Therefore, following \cite{gharan2011asymmetric}, in each connected component $C$ of $\Gamma'$, we can find a $O(1)$-thin spanning tree $T_C$ (w.r.t.~$z$). More specifically, we will find a $100$-thin spanning tree in each of them.

\textbf{Phase 3.}
We define a graph $F$ with $V(F)$ being the set of connected components of $\Gamma'$ and $\{C,C'\} \in E(F)$ iff there exists an edge between some vertex in $C$ and some vertex in $C'$ in $\Gamma$. We set the weight of $\{C,C'\}$ to be $z(C, C')$. We call $F$ the \emph{graph of components}.
\begin{center}
\scalebox{0.9}{\includegraphics{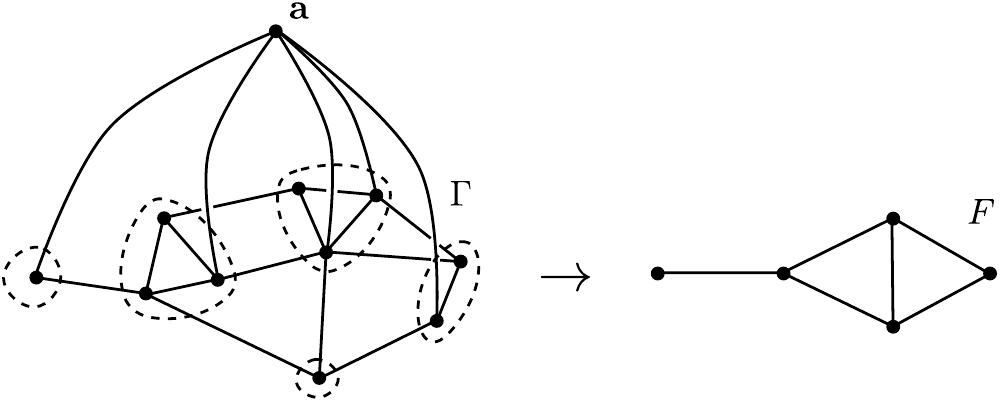}}
\end{center}
We define a graph $G'$ obtained from $G$ by contracting every connected component of $\Gamma'$ into a single vertex.
We remark that we may get parallel edges in $G'$.

\textbf{Phase 4.}
In this phase, we construct a tree $T'$ in $G'$.
We say that a vertex in $F$ is \emph{originally heavy}, if it has degree of at most $15$ in $F$.
Since $F$ is planar, the minimum degree of $F$ is at most $5$. We contract all vertices in $V(G')\setminus \{\mya\}$ into the apex sequentially. In each step, we find a vertex in $F$ with degree at most $5$, we contract it to the apex in $G'$, and we delete it from $F$. Since the remaining graph $F$ is always planar, there is always a vertex of degree at most $5$ in it, and thus we can continue this process until all vertices of $V(G')\setminus \{\mya\}$ are contracted to the apex.

Initially we consider all vertices of $F$ having no parent. In each step when we contract a vertex $C$ with degree at most $5$ in $F$ to the apex in $G'$, for each neighbor $C'$ of $C$ in $F$, we make $C$ the $\emph{parent}$ of $C'$ if $C'$ does not have any parents so far. Note that each vertex in $F$ can be the parent of at most $5$ other vertices, and can have at most one parent.

Every time we contract a vertex $C$ to the apex, we add an edge $e$ to $T'$. If $C$ is originally heavy, we add an arbitrary edge $e$ from $C$ to the apex; we will show in Lemma \ref{lem:originally_heavy_weight} that $z(C, \{a\}) \geq 0.5$, which  implies that such an edge always exists in $G$.
Otherwise, we add an arbitrary edge from $C$ to its parent (which is a neighbor vertex, therefore such an edge exists in $G$). We will show in the next section that each vertex in $F$ is originally heavy or it has a parent (or both). Therefore, $T'$ is a tree on $G'$.

\textbf{Phase 5.}
In this last phase we compute a tree $T$ in $G$.
We set $E(T)=E(T') \cup \bigcup_{C\in V(F)} E(T_C)$.
We prove in the next subsection that $T$ is a $O(1)$-thin spanning tree in $G$.

\fi

\ifappendix
\iffull
\subsection{Analysis}
\fi
We next show that $T$ is a $O(1)$-thin spanning tree in $G$.

\begin{lemma}\label{lem:weight_of_edges_of_F}
The weight of every edge in $F$ is less than $0.1$.
\end{lemma}

\begin{proof}
Let $\{C,C'\}\in E(F)$.
By construction, each component of $\Gamma'$ is formed by finding a tiny cut in some other component. Suppose $C$ was formed either simultaneously with $C'$ or later than $C'$ by finding a tiny cut in some $C''$. If $C' \subseteq C''$ then $z(C,C') < z(\delta(C)) < 0.1$.
Otherwise, the total weight of edges from $C$ to $C'$ is a part of a tiny cut which means that $z(C,C') < 0.1$.
\end{proof}

\begin{lemma}\label{lem:originally_heavy_weight}
Let $C$ be an originally heavy vertex in $F$. Then $z(C, \{\mya\}) \geq 0.5$.
\end{lemma}

\begin{proof}
For every neighbor $C'$ of $C$ in $F$, by Lemma \ref{lem:weight_of_edges_of_F} we have that the weight of $\{C,C'\}$ is less than $0.1$.
By the assumption on $z$, we have that $z(\delta(C))\geq 2$.
Now since $C$ has degree of at most $15$, we have that $z(C, \{\mya\}) \geq 0.5$, as desired.
\end{proof}

\begin{lemma}\label{lem:heavy_or_parent}
Each vertex $C \in V(F)$ is originally heavy or it has a parent (both cases might happen for some vertices).
\end{lemma}

\begin{proof}
Let $C \in V(F)$. If it is originally heavy, we are done. Otherwise, it has degree of at least $15$. We know that all vertices in $F$ are going to be contracted to $\mya$ at some point, and we only contract vertices with degree at most $5$ in each iteration. This means that at least $10$ other neighbors of $C$ were contracted to the apex before we decided to contract $C$ to the apex. Therefore one of them is the parent of $C$.
\end{proof}

\begin{lemma}\label{lem:T_is_spanning}
$T$ is a spanning tree in $G$.
\end{lemma}

\begin{proof}
Suppose $G$ has $n$ vertices and $F$ has $m$ vertices. After the second phase of our algorithm, we obtain a spanning forest on $\Gamma$ with $m$ components and $n-m-1$ edges. Each time we contract a vertex of $F$ to the apex, we add a single edge to $T$. Therefore, $T$ has $n-1$ edges. It is now sufficient to show that $T$ is connected.\par

We will show that for every vertex $u$ in $\Gamma$, there is path between $u$ and $\mya$ in $T$. Let $u$ be a vertex of $\Gamma$. Suppose $u$ is in some component $C_u$ which is a vertex of $F$. If $C_u$ is originally heavy, then there is an edge $e$ in $T$ between a vertex $v \in C_u$ and the apex. Since we have a spanning tree in $C_u$, there is a path between $u$ and $v$ in $T$. Therefore, there is path between $u$ and the apex $a$ in $T$.\par

Otherwise, $C_u$ must have some parent $C_{u_1}$ and there is an edge between these two components. Therefore, there is a path between $u$ and each vertex of these two components. Now, the same argument applies for $C_{u_1}$. Either it is originally heavy or it has a parent $C_{u_2}$. If it is originally heavy, we are done. Otherwise, we use the same argument for $C_{u_2}$. Note that by construction and the definition of a parent, we do not reach the same component in this sequence. Therefore, at some point, we reach a component $C_{u_k}$ which is originally heavy and we are done.
\end{proof}

Now we are ready to show that $T$ is a $O(1)$-thin tree in $G$ (w.r.t.~$z$). We have to show that there exists some constant $\alpha$ such that for every cut $U$, $\left\vert{E(T) \cap \delta(U)}\right\vert \leq \alpha \cdot z(\delta(U))$. Let $U$ be a cut in $G$. We can assume w.l.o.g.~that $\mya\notin U$, since otherwise we can consider the cut $V(G) \setminus U$.
We partition $E(T) \cap \delta(U)$ into three subsets:

\begin{description}
\item{(1)}
$T_1 = \{\{\mya,v\} \in E(T)\cap \delta(U) : v\in V(\Gamma)\}$.

\item{(2)}
$T_2 = \{\{u,v\}\in E(T)\cap \delta(U) : \text{$u$ and $v$ are in the same component of $\Gamma'$}\}$.

\item{(3)}
$T_3 = \{\{u,v\}\in E(T)\cap \delta(U) : \text{$u$ and $v$ are in different components of $\Gamma'$}\}$.
\end{description}


\begin{lemma}\label{lem:T_1_is_thin}
There exists some constant $\alpha_1$ such that $\left\vert{T_1}\right\vert \leq \alpha_1 \cdot z(\delta(U))$.
\end{lemma}

\begin{proof}
Let $e = \{\mya,v\} \in T_1$ where $v \in V(\Gamma)$.
Let $C_v \in V(F)$ such that $v \in C_v$. By the construction of $T$, $C_v$ is originally heavy. If $C_v \subseteq U$, we can charge $e$ to $z(C_v, \{\mya\})$, which we know is at least $0.5$.
Otherwise suppose $C_v$ is not a subset of $U$.
By the assumption we have $\mya\notin U$ and thus $v\in U$ which implies that $U\cap C_v \neq \emptyset$.
By the construction, we know that there is no tiny cuts in $C_v$.
Therefore, $z(\delta(U)\cap E(G[C_v]))\geq 0.1$.
Thus we can charge $e$ to the total weight of the edges in $\delta(U)\cap E(G[C_v])$.
Note that for each $C_v\in V(F)$, there is at most one edge in $T_1$ between $\mya$ and $C_v$.
Therefore we have that $\left\vert{T_1}\right\vert \leq 10 \cdot z(\delta(U))$.
\end{proof}

\begin{lemma}\label{lem:T_2_is_thin}
There exists some constant $\alpha_2$ such that $\left\vert{T_2}\right\vert \leq \alpha_2 \cdot z(\delta(U))$.
\end{lemma}

\begin{proof}
We have
\begin{align*}
\left\vert{T_2}\right\vert &= \left\vert{\bigcup_{C\in V(F)} (E(C) \cap T_2)}\right\vert\\
 &= \left\vert{\bigcup_{C\in V(F)} (E(C) \cap E(T) \cap \delta(U))}\right\vert\\
 &= \left\vert{\bigcup_{C\in V(F)} (E(T_C) \cap \delta(U))}\right\vert\\
 &\leq \sum_{C\in V(F)} 100\cdot z(\delta(U) \cap E(C))\\
 &\leq 100\cdot z(\delta(U)),
\end{align*}
concluding the proof.
\end{proof}

\begin{lemma}\label{lem:T_3_is_thin}
There exists some constant $\alpha_3$ such that $\left\vert{T_3}\right\vert \leq \alpha_3 \cdot z(\delta(U))$.
\end{lemma}

\begin{proof}
We partition $T_3$ into three subsets:

\begin{description}
\item{(1)}
$T_{31} = \{\{u,v\} \in T_3 : u \in U, v \notin U, \exists C_u, C_v \in V(F) \text{ such that } u \in C_u, v \in C_v, U \cap C_v \neq \emptyset \}$.
\item{(2)}
$T_{32} = \{\{u,v\} \in T_3 : u \in U, v \notin U, \exists C_u, C_v \in V(F) \text{ such that } u \in C_u, v \in C_v, U \cap C_u \neq \emptyset \}$.
\item{(3)}
$T_{33} = \{\{u,v\} \in T_3 : u \in U, v \notin U, \exists C_u, C_v \in V(F) \text{ such that } u \in C_u, v \in C_v, U \cap C_v = \emptyset, C_u \subseteq U \}$.
\end{description}

First for each $e = \{u,v\} \in T_{31}$ where $v \in C_v$ for some $C_v \in V(F)$, we have that $U_v = U \cap C_v$ is a cut in $C_v$ which is not tiny. By the construction, $C_v$ can be the parent of at most five other vertices in $F$ and it can have at most one parent. Therefore, there are at most six edges in $T_3$ with a vertex in $C_v$. So we can charge $e$ and at most five other edges to $z(\delta(U_v))$. Since $z(\delta(U_v)) \geq 0.1$ we get $\left\vert{T_{31}}\right\vert \leq 60 \cdot z(\delta(U))$.

Second for each $e = \{u,v\} \in T_{32}$ where $u \in C_u$ for some $C_u \in V(F)$, we have that $U_u = U \cap C_u$ is a cut in $C_u$ which is not tiny. Therefore, the same argument for $C_v$ as in the first case, applies here for $C_u$ and we get $\left\vert{T_{32}}\right\vert \leq 60 \cdot z(\delta(U))$.

Finally, for $T_{33}$ we need to find a constant $\alpha_{33}$ such that $\left\vert{T_{33}}\right\vert \leq \alpha_{33} \cdot z(\delta(U))$. First, we define a new cut $U_1$ as follows. For every $C \in V(F)$ with $C \cap U \neq \emptyset$, if $C \cap U \neq C$, we add all the other vertices of $C$ to $U$ and we say that $C$ is \emph{important}. This process leads to a new cut $U_1$ such that for every $C \in F$, either $C \cap U_1 = \emptyset$ or $C \subseteq U_1$. Let $U_2 = \{C \in V(F) : C \subseteq U_1\}$. Let $X = \{C \in V(F) : C \notin U_2\}$ and $Y = X \cup \{\mya\}$.

Let 
\[
T_{331} = \{\{u,v\}\in T_{33} : u\in U, u\in C \text{ for some } C \in V(F) \text{ with } \deg_{F[U_2]}(C) < 19\}.
\]
Let also $T_{332} = T_{33} \setminus T_{331}$.

For each edge $e = \{u,v\} \in T_{331}$ where $u \in U$ and $u \in C_u$ for some $C_u \in V(F)$, we have $\deg_{F[U_2]}(C_u) < 19$. By Lemma \ref{lem:weight_of_edges_of_F}, we know that for any $C, C' \in V(F)$, $z(C, C') \leq 0.1$. Therefore, we get $z(C_u, Y) \geq 0.2$. Note that there are at most six edges in $T_{331}$ with a vertex in $C_u$. So we can charge $e$ and at most five other edges to $z(C_u, Y)$. Therefore, $\left\vert{T_{331}}\right\vert \leq 30 \cdot z(\delta(U))$.

Let $V_1 = \{ C \in U_2 : \deg_{F[U_2]}(C) \geq 19 \}$ and $V_2 = \{ C \in U_2 : \deg_{F[U_2]}(C) \leq 5 \}$. By Euler's formula, we know that the average degree of a planar graph is at most $6$. Since $F[U_2]$ is planar, we get $\left\vert V_1 \right\vert \leq \left\vert V_2 \right\vert$. For any $C \in V_2$, if $C$ is important, then $C \cap U$ is a cut for $C$ and we have $z(C \cap U, Y) \geq 0.1$. If $C$ is not important, then we have $z(C , Y) \geq 1.5$. Note that for any $C' \in V_1$, there are at most six edges in $T_{332}$ with a vertex in $C'$. Therefore, we have $\left\vert{T_{332}}\right\vert \leq 60 \cdot z(\delta(U))$. 

Now since $T_3 = T_{31} \cup T_{32} \cup T_{331} \cup T_{332}$, we have $\left\vert{T_3}\right\vert \leq 210 \cdot z(\delta(U))$ completing the proof.
\end{proof}

\begin{lemma}\label{lem:T_is_thin_1_apex}
$T$ is a $O(1)$-thin spanning tree in $G$.
\end{lemma}

\begin{proof}
By Lemma \ref{lem:T_is_spanning} we know that $T$ is a spanning tree. For any $U \subseteq V(G)$, by Lemmas \ref{lem:T_1_is_thin}, \ref{lem:T_2_is_thin} and \ref{lem:T_3_is_thin} we get $\left\vert{T}\right\vert \leq 320 \cdot z(\delta(U))$. This completes the proof.
\end{proof}

We are now ready to prove the main result of this Section.

\iffull
\begin{theorem}\label{thm:thin_1-apex}
Let $G$ be a $1$-apex graph and let $z:E(G)\to \mathbb{R}_{\geq 0}$ be $\beta$-thick for some $\beta>0$.
Then there exists a polynomial time algorithm which given $G$ and $z$ outputs a $O(1/\beta)$-thin spanning tree in $G$ (w.r.t.~$z$).
\end{theorem}
\fi

\ifabstract
\begin{proof}[Proof of Theorem \ref{thm:thin_1-apex}]
\fi
\iffull
\begin{proof}
\fi
For $\beta \geq 2$, by Lemma \ref{lem:T_is_thin_1_apex} we know that we can find a $320$-thin spanning tree in $G$. For any $\beta$ with $0<\beta<2$, the assertion follows by scaling $z$ by a factor of $2/\beta$.
\end{proof}

\fi

\ifmain

\section{Thin forests in graphs with many apices}\label{sec:k-apices}
\label{sec:thin_a-apex}

Let $a\geq 1$.
In this section, we describe an algorithm for finding thin-forests in an $a$-apex graph. The high level approach is analogous to the case of $1$-apex graphs. We construct a similar graph $F$ of components and contract each vertex of $F$ to some apex.

Let $e_0 = \{u_0, v_0\} \in E(G)$.
Let $G'$ be obtained from $G$ by contracting $e_0$.
We define a new weight function $z'$ on the edges of $G'$ as follows.
For any $\{u , v_0\}\in E(G)$, we set $z'(\{u, v_0\}) = z(\{u, v_0\}) + z(\{u, u_0\})$.
For any other edge $e$ we set $z'(e) = z(e)$.
We say that $z'$ is \emph{induced} by $z$.
Similarly, when $G'$ is obtained by contracting a subset $X$ of edges in $G$, we define $z'$ by inductively contracting the edges in $X$ in some arbitrary order.

For the remainder of this section let $G$ be a $a$-apex graph with the set of apices $A = \{\mya_1, \mya_2, \cdots, \mya_a\}$. Let $\Gamma$ be the planar part of $G$. Let $z$ be a $2$-thick weight function on the edges of $G$. The algorithm proceeds in $5$ phases.

\textbf{Phase 1.}
We say that a cut $U$ is \emph{tiny} (w.r.t.~$z$) if $z(\delta(U)) < 1/(100\cdot a)$. Similar to the case of $1$-apex graphs, we start with $\Gamma$ and repeatedly partition it via tiny cuts until there are no more such cuts and we let $\Gamma'$ be the resulting graph.

\textbf{Phase 2.}
For each connected component $C$ of $\Gamma'$ we find a $O(a)$-thin tree $T_C$ using Theorem \ref{thm:thin_tree_planar}.

\textbf{Phase 3.}
We define $F$ and $G'$ exactly the same way as in the case of $1$-apex graphs.


\begin{lemma}\label{lem:weight_of_edges_of_F_k-apex_case}
For every $\{C, C'\} \in E(F)$, we have $z(C, C') \leq 1/(100 \cdot a)$.
\end{lemma}

\iffull
\begin{proof}
The same argument as in Lemma \ref{lem:weight_of_edges_of_F} applies here. The only difference here is that a cut $U$ is tiny if $z(\delta(U)) < 1/(100\cdot a)$.
\end{proof}
\fi

\textbf{Phase 4.}
We construct a forest $T'$ on $G'$. Let $m=\left\vert{V(F)}\right\vert$. We define a sequence of planar graphs $F_0, F_1, \cdots, F_m$, a sequence of graphs $G'_0, G'_1, \cdots, G'_m$ and a sequence of weight functions $z_0, z_1, \cdots, z_m$ as follows. Let $F_0 = F$, $G'_0 = G'$ and $z_0 = z$.
We also define a sequence of forests $P_0,\ldots,P_m$ where each $P_j$ contains a tree rooted at each $\mya_i\in A$.
We set $P_0$ to be the forest that contains a tree for each $\mya_i\in A$ and with no other vertices.

Let $C \in V(F_j)$ for some $j$. For any $\mya_i \in A$, we say that $C$ is \emph{$\mya_i$-heavy} in $F_j$ if $z_j(C, \{\mya_i\}) \geq 1/a$. Let $C' \in V(F)$. For any $\mya_i \in A$, we say that $C'$ is \emph{originally $\mya_i$-heavy} if $C'$ is $\mya_i$-heavy in $F$.

We maintain the following inductive invariant:
\begin{description}
\item{(I1)}
For any $j\in \{0,\ldots,m-1\}$, let $v\in V(F_j)$ be a vertex of minimum degree.
Then either there exists some $\mya_i\in A$ such that $v$ is originally $\mya_i$-heavy or $v\in V(P_j)$.
\end{description}

Consider some $i\in \{0,\ldots,m-1\}$.
Let $v_i \in V(F_i)$ be a vertex with minimum degree.
If $v_i$ is originally $\mya_j$-heavy for some $\mya_j\in A$, then we contract $v_i$ to $\mya_j$.
Otherwise, by the inductive invariant (I1), we have that $v_i\in V(P_i)$.
Thus there exists a tree in $P_i$ containing $v_i$ that is rooted in some $\mya_j\in A$; we contract $v_i$ to $\mya_j$.
In either case, by contracting $v_i$ to $\mya_j$ we obtained $G'_{i+1}$ from $G'_i$.
We also delete $v_i$ from $F_i$ to obtain $F_{i+1}$.
We let $z_{i+1}$ be the weight function on $G'_{i+1}$ induced by $z_i$.

Finally, we need to define $P_{i+1}$.
If $v_i$ was originally $\mya_j$-heavy then we add $v_i$ to $P_i$ via an edge $\{v_i,\mya_j\}$.
For each $u \in V(F_i)$ that is a neighbor of $v_i$, and is not in $V(P_i)$, we add $u$ to $P_{i+1}$ by adding the edge $\{v_i,u\}$ iff the following conditions hold:
\begin{description}
\item{(i)}
For all $\mya_r\in A$, we have that $u$ is not $\mya_r$-heavy in $F_i$.

\item{(ii)}
$u$ is $\mya_j$-heavy in $F_{i+1}$.
\end{description}
In this case we say that $v$ is the $\emph{parent}$ of $u$. This completes the description of the process that contracts each vertex in $V(G')$ into some apex.

\begin{lemma}\label{lem:being_a_i_heavy}
Let $j \in \{0, 1, \ldots, m-1\}$. Let $C \in V(F_j)$ be a vertex with minimum degree. Then there exists $\mya_i \in A$ such that $C$ is $\mya_i$-heavy in $F_j$.
\end{lemma}

\begin{proof}
Since $C$ has at most $5$ neighbors in $F_j$, by Lemma \ref{lem:weight_of_edges_of_F_k-apex_case} we have $z_j(C, A) \geq 2 - 5/(100 \cdot a) \geq 1$. Therefore by averaging, there exists an apex $\mya_i$ such that $z_j(C, \{\mya_i\}) \geq 1/a$.
\end{proof}

\begin{lemma}\label{lem:degree_in_P_j}
For any $j \in \{0, \ldots, m\}$ and for any $v \in V(\Gamma) \cap V(P_j)$, we have $\deg_{P_j}(v) \leq 6$.
\end{lemma}

\begin{proof}
By the construction, in each step we pick a vertex of minimum degree and contract it into some apex. Since $F_i$ is planar, its minimum degree is at most $5$. This means that for any $v \in V(\Gamma) \cap V(P_j)$, $v$ can be the parent of at most five other vertices and can have at most one parent. This completes the proof.
\end{proof}

\begin{lemma}\label{lem:Originally_a-i-heavy_or_parent}
The inductive invariant (I1) is maintained.
\end{lemma}

\begin{proof}
For any $j \in \{0, \ldots, m-1\}$, let $v \in V(F_j)$ be a vertex of minimum degree. If there exists some $\mya_i \in A$ such that $v$ is originally $\mya_i$-heavy, then we are done. Suppose for all $\mya_i \in A$, $v$ is not originally $\mya_i$-heavy. By Lemma \ref{lem:being_a_i_heavy} we know that there exists some $\mya_l \in A$ such that $v$ is $\mya_l$-heavy in $F_j$. 
Let $j^*\in \{1,\ldots,j\}$ be minimum such that $v$ is not $\mya_t$-heavy in $F_{j^*-1}$ for all $\mya_t\in A$, and $v$ is $\mya_{t'}$-heavy in $F_{j^*}$ for some $\mya_{t'}\in A$.
Let $u\in V(F_{j^*-1})$ be vertex that is contracted to some apex in step $j^*$.
It follows by construction that $u$ is the parent of $v$ in $P_{j^*}$.
Since $j\geq j^*$ it follows that $v\in V(P_j)$, concluding the proof.
\end{proof}

Now we are ready to describe how to construct $T'$ in $G'$. For any $l \in \{0, \ldots, m-1\}$, let $C \in V(F_l)$ be a vertex of minimum degree. If $C$ is originally $\mya_i$-heavy for some $\mya_i \in A$ and we contract $C$ to $\mya_i$, we pick an arbitrary edge $e$ between $C$ and $\mya_i$ and we add it to $T'$. Otherwise, by Lemma \ref{lem:Originally_a-i-heavy_or_parent} we have $C \in V(P_l)$. This means that $C$ has a parent $C'$. In this case, we pick an arbitrary edge $e$ between $C$ and $C'$ and we add it to $T'$.

\textbf{Phase 5.}
We construct a forest $T$ in $G$ the same way as in the $1$-apex case. We set $E(T)=E(T') \cup \bigcup_{C\in V(F)} E(T_C)$.

This completes the description of the algorithm.

\fi

\ifappendix

\iffull
\subsection{Analysis}
\fi

By the construction, $T$ has $a$ connected components. We will show that $T$ is a $O(a)$-thin spanning forest. Let $U$ be a cut in $G$. Similar to the $1$-apex case, we partition $E(T) \cap \delta(U)$ into three subsets:

\begin{description}
\item{(1)}
$T_1 = \{\{\mya_i,v\} \in E(T)\cap \delta(U) : \mya_i \in A, v\in V(\Gamma)\}$.

\item{(2)}
$T_2 = \{\{u,v\}\in E(T)\cap \delta(U) : \text{$u$ and $v$ are in the same component of $\Gamma'$}\}$.

\item{(3)}
$T_3 = \{\{u,v\}\in E(T)\cap \delta(U) : \text{$u$ and $v$ are in different components of $\Gamma'$}\}$.
\end{description}

\begin{lemma}\label{lem:T_1_is_thin_k-apex}
There exists a constant $\alpha_1$ such that $\left\vert{T_1}\right\vert \leq \alpha_1 \cdot a \cdot z(\delta(U))$.
\end{lemma}

\begin{proof}

A similar argument as in the case of $1$-apex graph applies here with two differences. First, a cut $U$ is tiny if $z(\delta(U)) < 1/(100 \cdot a)$. Second, for any $\mya_i \in A$ and $C \in V(F)$ where $C$ is originally $\mya_i$-heavy, we have that $z(C, \{\mya_i\}) \geq 1/a$. Therefore, we get $\left\vert{T_1}\right\vert \leq 100 \cdot a \cdot z(\delta(U))$.
\end{proof}

\begin{lemma}\label{lem:T_2_is_thin_k-apex}
There exists a constant $\alpha_2$ such that $\left\vert{T_2}\right\vert \leq \alpha_2 \cdot a \cdot z(\delta(U))$.
\end{lemma}

\begin{proof}
Again, a similar argument as in the case of $1$-apex graphs applies here. The only difference here is the definition of tiny cut. Therefore, we get $|T_2| \leq 100 \cdot a \cdot z(\delta(U))$.
\end{proof}

\begin{lemma}\label{lem:T_3_is_thin_k-apex}
There exists a constant $\alpha_3$ such that $\left\vert{T_3}\right\vert \leq \alpha_3 \cdot a \cdot z(\delta(U))$.
\end{lemma}

\begin{proof}

Similar to the $1$-apex case, we partition $T_3$ into three subsets:

\begin{description}

\item{(1)}
$T_{31} = \{\{u,v\} \in T_3 : u \in U, v \notin U, \exists C_u, C_v \in V(F) \text{ such that } u \in C_u, v \in C_v, U \cap C_v \neq \emptyset \}$.

\item{(2)}
$T_{32} = \{\{u,v\} \in T_3 : u \in U, v \notin U, \exists C_u, C_v \in V(F) \text{ such that } u \in C_u, v \in C_v, U \cap C_u \neq \emptyset \}$.

\item{(3)}
$T_{33} = \{\{u,v\} \in T_3 : u \in U, v \notin U, \exists C_u, C_v \in V(F) \text{ such that } u \in C_u, v \in C_v, U \cap C_v = \emptyset, C_u \subseteq U \}$.

\end{description}

The arguments for $T_{31}$ and $T_{32}$ are the same as in $1$-apex graphs. The only difference here is that a cut $U$ is tiny if $z(\delta(U)) < 1/(100\cdot a)$. Therefore, we have $\left\vert{T_{31}}\right\vert \leq 600 \cdot a \cdot z(\delta(U))$ and $\left\vert{T_{32}}\right\vert \leq 600 \cdot a \cdot z(\delta(U))$.

Now for $T_{33}$ we want to find a constant $\alpha_{33}$ such that $\left\vert{T_{33}}\right\vert \leq \alpha_{33} \cdot a \cdot z(\delta(U))$. We define two new cuts $U_1$ and $U_2$ as follows. For every $C \in V(F)$ with $C \cap U \neq \emptyset$, if $C \cap U \neq C$, we add all other vertices of $C$ to $U$ (delete all other vertices of $C$ from $U$) to obtain $U_1$ ($U_2$) and we say that $C$ is $U$-\emph{important}. Let $U'_1 = \{C \in V(F) : C \subseteq U_1\}$ and $U'_2 = \{C \in V(F) : C \subseteq U_2\}$.

For any $e =\{u, v\} \in T_{33}$ where $u \in U$, $v \notin U$, $u \in C_u$ and $v \in C_v$ for some $C_u, C_v \in V(F)$, by the construction of $T_3$, we have that both $C_u$ and $C_v$ have been contracted to the same apex $\mya_i$ for some $\mya_i \in A$. Let $j \in \{0, \ldots, m\}$ be the step during which $C_u$ is contracted to $\mya_i$.
Let $B = \{ C \in V(F) : C \text{ is contracted to } \mya_i \}$.
Let $D$ be the connected component of $F[B]$ containing $C_u$. Let $D_{U'_1}^{\myin} = D[U'_1]$, $D_{U'_1}^{\myout} = D[V(D) \setminus U'_1]$, $D_{U'_2}^{\myin} = D[U'_2]$, $D_{U'_2}^{\myout} = D[V(D) \setminus U'_2]$.


We consider the following two cases:

\begin{description}
\item{Case 1: $\mya_i\notin U$.}
We know that $C_u \in U$. By the construction, we have that $z_j(C_u, \{\mya_i\}) \geq 1/a$. If $z_0(C_u, \{\mya_i\}) \geq 1/(100 \cdot a)$, we can charge $e$ to $z_0(C_u, \{\mya_i\})$ and we know that there are at most six edges in $T_{33}$ with a vertex in $C_u$. 

Otherwise, we have that $z_0(C_u, (V(D) \setminus C_u)) \geq 99/(100 \cdot a)$. If $z_0(C_u, (V(D_{U'_1}^{\myout}) \setminus C_u)) \geq 1/(100 \cdot a)$, then by the construction of $D_{U'_1}^{\myout}$ we get $z_0(C_u, (G \setminus U)) \geq 1/(100 \cdot a)$. Therefore we can charge $e$ to  $z_0(C_u, (G \setminus U))$.

Otherwise, we have that $z_0(C_u, (V(D_{U'_1}^{\myin}) \setminus C_u)) \geq 98/(100 \cdot a)$. This implies that $\deg_{D_{U'_1}^{\myin}} (C_u) \geq 98$. Now note that $D_{U'_1}^{\myin}$ is a planar graph and its average degree is at most $6$. Now a similar argument as in the $1$-apex case applies here. Let $V_1 =\{C \in D_{U'_1}^{\myin} : \deg_{D_{U'_1}^{\myin}} (C) \geq 98\}$. Let $V_2 = \{C \in D_{U'_1}^{\myin} : \deg_{D_{U'_1}^{in}} (C) \leq 5\}$. By planarity of $D_{U'_1}^{\myin}$, we have that $\left\vert V_1 \right\vert \leq \left\vert V_2 \right\vert$. For any $C \in V_2$, if $C$ is $U$-important, then $C \cap U$ is a cut for $C$ which is not tiny. Therefore we have $z(C \cap U, G \setminus U) \geq 1/(100 \cdot a)$. If $C$ is not $U$-important, we have $z(C, G \setminus U) \geq 95/(100 \cdot a)$. Now note that for any $C' \in V_1$ there are at most six edges in $T_{33}$ with a vertex in $C'$. Therefore we have $\left\vert{T_{33}}\right\vert \leq 1000 \cdot a \cdot z(\delta(U))$.

\item{Case 2: $\mya_i\in U$.}
Let $X_1 =\{C \in D_{U'_2}^{out} : \deg_{D_{U'_2}^{out}} (C) \geq 98\}$. Let $X_2 = \{C \in D_{U'_2}^{out} : \deg_{D_{U'_2}^{out}} (C) \leq 5\}$.
We know that $C_v \notin U$. We follow a similar approach as in the first case by considering $U'_2$, $D_{U'_2}^{\myin}$ and $D_{U'_2}^{\myout}$. The same argument applies here by replacing $U_1'$, $D_{U'_1}^{\myin}$, $D_{U'_1}^{\myout}$, $X_1$ and $X_2$ with $U'_2$, $D_{U'_2}^{\myout}$, $D_{U'_2}^{\myin}$, $V_1$ and $V_2$ respectively. Therefore, we get $\left\vert{T_{33}}\right\vert \leq 1000 \cdot a \cdot z(\delta(U))$.
\end{description}

Now from what we have discussed, we have $\left\vert{T_{31}}\right\vert \leq 600 \cdot a \cdot z(\delta(U))$, $\left\vert{T_{32}}\right\vert \leq 600 \cdot a \cdot z(\delta(U))$ and $\left\vert{T_{33}}\right\vert \leq 1000 \cdot a \cdot z(\delta(U))$. Therefore, we get $\left\vert{T_{3}}\right\vert \leq 2200 \cdot a \cdot z(\delta(U))$ completing the proof.
\end{proof}

\begin{lemma}\label{T_is_thin_in_G}
$T$ is a $O(a)$-thin spanning forest in $G$ with at most $a$ connected components.
\end{lemma}

\begin{proof}
By the construction, $T$ is a spanning forest and has at most $a$ connected components. Let $\alpha = \alpha_1 + \alpha_2 + \alpha_3$, where $\alpha_1$, $\alpha_2$ and $\alpha_3$ are obtained by Lemmas \ref{lem:T_1_is_thin_k-apex}, \ref{lem:T_2_is_thin_k-apex}, \ref{lem:T_3_is_thin_k-apex}. Therefore for any $U \subset V(G)$, we have $\left\vert{T}\right\vert \leq \alpha \cdot a \cdot z(\delta(U)) = 2400 \cdot a \cdot z(\delta(U))$ completing the proof.
\end{proof}

We are now ready to prove the main result of this Section.

\iffull
\begin{theorem}\label{thm:thin_a-apex}
Let $a\geq 1$ and let $G$ be a $a$-apex graph with set of apices $A=\{\mya_1,\ldots,\mya_a\}$.
Let $z:E(G)\to \mathbb{R}_{\geq 0}$ be $\beta$-thick for some $\beta>0$.
Then there exists a polynomial time algorithm which given $G$, $A$ and $z$ outputs a $O(a/\beta)$-thin spanning forest in $G$ (w.r.t.~$z$) with at most $a$ connected components.
\end{theorem}
\fi

\ifabstract
\begin{proof}[Proof of Theorem \ref{thm:thin_a-apex}]
\fi
\iffull
\begin{proof}
\fi
By Lemma \ref{T_is_thin_in_G}, for $\beta \geq 2$, we know that we can find a $(2400 a)$-thin spanning forest in $G$ (w.r.t.~$z$) with at most $a$ connected components. For any other $0 < \beta < 2$, the claim follows by scaling $z$ by a factor of $2/\beta$.
\end{proof}

\fi

\ifmain

\section{Thin forests in higher genus graphs with many apices}\label{sec:thin_forests_higher_genus}

The following theorem is implicit in the work of Erickson and Sidiropoulos \cite{erickson2014near}.

\begin{theorem}[Erickson and Sidiropoulos \cite{erickson2014near}]\label{tim:result_of_genus_g}
Let $G$ be a graph with $\eg(G)=g$, and let $z$ be a $\beta$-thick weight function on the edges of $G$ for some $\beta \geq 0$.
Then there exists a polynomial time algorithm which given $G$, $z$, and an embedding of $G$ into a surface of Euler genus $g$, outputs a $O(1/\beta)$-thin spanning forest in $G$ (w.r.t.~$z$), with at most $g$ connected components. 
\end{theorem}

In this section, we study the problem in higher genus graphs. First, the following two Lemmas can be obtained by Euler's formula.

\begin{lemma}\label{lem:minimum_degree_in_G}
Let $G$ be a graph of genus $g \geq 1$ with $|V(G)| \geq 10 g$. Then there exists $v_0 \in V(G)$ with $\deg_G (v_0) \leq 7$.
\end{lemma}


\begin{lemma}\label{lem:average_degree_in_G}
Let $G$ be an $n$ vertex graph of genus $g \geq 1$. Then the average degree of vertices of $G$ is at most $6 + 12(g-1)/n$.
\end{lemma}


For the remainder of this section, let $G$ be an $a$-apex graph with the set of apices $A = \{\mya_1, \mya_2, \ldots, \mya_a\}$ on a surface of genus $g$. Let $\Gamma = G \setminus A$, where $\Gamma$ is a graph of genus $g$. 
Let $z$ be a $2$-thick weight function on the edges of $G$.
We will find a $O(a \cdot g)$-thin (w.r.t.~$z$) spanning forest in $G$ with at most $O(a + g)$ connected components. The high level approach is similar to the case where $\Gamma$ was planar. The algorithm proceeds in $5$ phases. 

\textbf{Phase 1.}
We say that a cut $U$ is \emph{tiny} (w.r.t.~$z$) if $z(\delta(U)) < 1/(1000\cdot a \cdot g)$. We construct $\Gamma'$ the same way as in Section \ref{sec:k-apices}. The only difference here is the definition of tiny cut.

\textbf{Phase 2.} 
Similar to the planar case, for each connected component $C$ of $\Gamma'$ we find a $O(a \cdot g)$-thin forest $T_C$, with at most $g$ connected components, using Theorem \ref{tim:result_of_genus_g}

\textbf{Phase 3.}
We define $F$ and $G'$ the exact same way as in Section \ref{sec:k-apices}.

\begin{lemma}\label{lem:weight_of_edges_of_F_genus_g}
For every $\{C, C'\} \in E(F)$, we have $z(C, C') \leq 1/(1000 \cdot a \cdot g)$.
\end{lemma}

\iffull
\begin{proof}
The same argument as in Lemma \ref{lem:weight_of_edges_of_F} applies here. The only difference here is the definition of tiny cut.
\end{proof}
\fi

\textbf{Phase 4.} 
We construct a spanning forest $T'$ on $G'$, with at most $a + 10 g$ connected components. We follow a similar approach as in the planar case. Let $m = |V(F)| - 10 g$. If $m \leq 0$, we set $E(T') = \emptyset$ and we skip to the next phase. Otherwise,  we define two sequences of graphs $F_0, F_1, \ldots, F_m$, $G'_0, G'_1, \ldots, G'_m$, a sequence of weight functions $z_0, z_1, \ldots, z_m$, a sequence of forests $P_0, P_1, \ldots, P_m$ satisfying the inductive invariant (I1) the exact same way as in Section \ref{sec:k-apices}. For any $j \in \{0, 1, \ldots, m-1\}$, $C \in V(F_j)$ and $\mya_i \in A$, we also define the notion of \emph{$\mya_i$-heavy} and \emph{originally $\mya_i$-heavy} the same way as in Section \ref{sec:k-apices}. The only differences here is that $m = |V(F)| - 10 g$ instead of $|V(F)|$.

\begin{lemma}\label{lem:minimum_degree_is_heavy_genus_g}
Let $j \in \{0, 1, \ldots, m\}$. Let $C \in V(F_j)$ be a vertex of minimum degree. Then there exists $\mya_i \in A$ such that $C$ is $\mya_i$-heavy in $F_j$.
\end{lemma}

\begin{proof}
By the construction, $|V(F_j)| \geq 10  g$. Therefore by Lemma \ref{lem:minimum_degree_in_G}, we have $\deg_{F_j} (C) \leq 7$. Therefore by Lemma \ref{lem:weight_of_edges_of_F_genus_g}, we get $z_j(C, A) \geq 2 - 7/(1000 \cdot a \cdot g) \geq 1$. This implies that there exists $\mya_i \in A$ such that $z_j(C, \{\mya_i\}) \geq 1/a$.
\end{proof}

\begin{lemma}\label{lem:degree_in_P_j_genus_g}
For any $j \in \{0, \ldots, m\}$ and for any $v \in \Gamma \cap P_j$, we have $\deg_{P_j}(v) \leq 8$.
\end{lemma}

\begin{proof}
A similar argument as in the planar case applies here. The only difference here is that the minimum degree is at most $7$. Therefore, every vertex can be the parent of at most $7$ other vertices and can have at most one parent.
\end{proof}

\begin{lemma}\label{lem:Originally_a-i-heavy_or_parent_genus_g}
The inductive invariant (I1) is maintained.
\end{lemma}

\begin{proof}
The exact same argument as in Section \ref{sec:k-apices} applies here.
\end{proof}

Now we construct a forest $T'$ on $G'$ the same way as in Section \ref{sec:k-apices}.

\begin{lemma}\label{T'_has_k+10g_components}
$T'$ has at most $a + 10 g$ connected components.
\end{lemma}

\begin{proof}
If $|V(F)| \leq 10 g$, then we are done. Otherwise, we have $|F_m| = 10 g$. Now since $|A| = a$, by the construction, we have that the number of connected components of $T'$ is at most $a + 10 g$.
\end{proof}

\textbf{Phase 5.}
We construct a forest $T$ in $G$ the exact same way as in Section \ref{sec:k-apices}, by setting $E(T)=E(T') \cup \bigcup_{C\in V(F)} E(T_C)$.

This completes the description of the algorithm.

\fi

\ifappendix

\iffull
\subsection{Analysis}
\fi

\begin{lemma}\label{lem:T_is_spanning_forest_genus_g}
$T$ is a spanning forest in $G$, with at most $O(a + g)$ connected components.
\end{lemma}

\begin{proof}
For any $C \in V(F)$, let $g_C$ be the genus of $\Gamma[C]$. By Theorem \ref{tim:result_of_genus_g} we know that the number of connected components in $T_C$ is at most $g_C$. Therefore, by Lemma \ref{T'_has_k+10g_components} we have that the number of connected components in $T$ is at most $a+10g+\sum_{C\in V(F)} g_C \leq a+11g$.
\end{proof}

For the thinness of $T$, we follow a similar approach as in the planar case. There are two main differences here: First for any $j \in \{0, 1, \ldots, m\}$, by Lemma \ref{lem:average_degree_in_G} we have that the average degree of $F_j$ is at most $20 g$. Second a cut $U$ is tiny if $z(\delta(U)) < 1 / (1000 \cdot a \cdot g)$.

Let $U$ be a cut in $G$. Similar to the planar case, we partition $E(T) \cap \delta(U)$ into three subsets:

\begin{description}
\item{(1)}
$T_1 = \{\{\mya_i,v\} \in E(T)\cap \delta(U) : \mya_i \in A, v\in V(\Gamma)\}$.

\item{(2)}
$T_2 = \{\{u,v\}\in E(T)\cap \delta(U) : \text{$u$ and $v$ are in the same component of $\Gamma'$}\}$.

\item{(3)}
$T_3 = \{\{u,v\}\in E(T)\cap \delta(U) : \text{$u$ and $v$ are in different components of $\Gamma'$}\}$.
\end{description}

Also similar to the planar case, we partition $T_3$ into three subsets:

\begin{description}
\item{(1)}
$T_{31} = \{\{u,v\} \in T_3 : u \in U, v \notin U, \exists C_u, C_v \in V(F) \text{ such that } u \in C_u, v \in C_v, U \cap C_v \neq \emptyset \}$.

\item{(2)}
$T_{32} = \{\{u,v\} \in T_3 : u \in U, v \notin U, \exists C_u, C_v \in V(F) \text{ such that } u \in C_u, v \in C_v, U \cap C_u \neq \emptyset \}$.

\item{(3)}
$T_{33} = \{\{u,v\} \in T_3 : u \in U, v \notin U, \exists C_u, C_v \in V(F) \text{ such that } u \in C_u, v \in C_v, U \cap C_v = \emptyset, C_u \subseteq U \}$.
\end{description}

\begin{lemma}\label{lem:first_4_are_the_same}
For any index $i \in \{1, 2, 31, 32\}$, there exists a constant $\alpha_i$ such that $|T_i| \leq \alpha_i \cdot a \cdot g \cdot z(\delta(U)) $.
\end{lemma}

\begin{proof}
A similar argument as in Lemmas \ref{lem:T_1_is_thin_k-apex}, \ref{lem:T_2_is_thin_k-apex} and \ref{lem:T_3_is_thin_k-apex} applies here. There are two differences here. First the definition of a tiny cut is different. Second, for each $C \in V(F)$, $C$ can be the parent of at most seven vertices and can have at most one parent. Therefore we get $|T_1| \leq 1000 \cdot a \cdot g \cdot z(\delta(U))$, $|T_2| \leq 1000 \cdot a \cdot g \cdot z(\delta(U))$, $|T_{31}| \leq 8000 \cdot a \cdot g \cdot z(\delta(U))$ and $|T_{32}| \leq 8000 \cdot a \cdot g \cdot z(\delta(U))$.
\end{proof}

\begin{lemma}\label{lem:T_33_is_thin_genus-g}
There exists a constant $\alpha_{33}$ such that $|T_{33}| \leq \alpha_{33} \cdot a \cdot g \cdot z(\delta(U))$.
\end{lemma}

\begin{proof}
Let $e = \{u, v\} \in T_{33}$. We follow a similar approach as in the planar case. We define $U_1$, $U_2$, $U'_1$, $U'_2$, $B$, $D_{U'_1}^{\myin}$, $D_{U'_1}^{\myout}$, $D_{U'_2}^{\myin}$ and $D_{U'_2}^{\myout}$ the exact same way as in Lemma \ref{lem:T_3_is_thin_k-apex}. Let $V_1 =\{C \in D_{U'_1}^{\myin} : \deg_{D_{U'_1}^{\myin}} (C) \geq 98 g\}$, $V_2 = \{C \in D_{U'_1}^{\myin} : \deg_{D_{U'_1}^{\myin}} (C) \leq 20  g\}$, $X_1 =\{C \in D_{U'_2}^{out} : \deg_{D_{U'_2}^{out}} (C) \geq 98g\}$, and $X_2 = \{C \in D_{U'_2}^{out} : \deg_{D_{U'_2}^{out}} (C) \leq 5g\}$.
By Lemma \ref{lem:average_degree_in_G} we have that $|V_1| \leq |V_2|$ and $|X_1| \leq |X_2|$. With these definitions, the rest of the proof is the same as in Lemma \ref{lem:T_3_is_thin_k-apex}, and thus we get $|T_{33}| \leq 10000 \cdot a \cdot g \cdot z(\delta(U))$.
\end{proof}

\begin{lemma}\label{T_is_thin_and_spanning_genus-g}
$T$ is a $O(a \cdot g)$-thin spanning forest in $G$, with at most $O(a + g)$ connected components.
\end{lemma}

\begin{proof}
By combining Lemmas \ref{lem:T_is_spanning_forest_genus_g}, \ref{lem:first_4_are_the_same} and \ref{lem:T_33_is_thin_genus-g} and we get $|T| \leq 24000 \cdot a \cdot g \cdot z(\delta(U))$, which proves the assertion.
\end{proof}

We are now ready to prove the main result of this Section.

\iffull
\begin{theorem}\label{thm:main_theorem_genus-g}
Let $a, g \geq 1$.
Let $G$ be a graph and $A\subseteq V(G)$, with $|A|=a$, such that $H = G \setminus A$ is a graph of genus $g$.
Let $z:E(G)\to \mathbb{R}_{\geq 0}$ be $\beta$-thick for some $\beta>0$.
Then there exists a polynomial time algorithm which given $G$, $A$, an embedding of $H$ on a surface of genus $g$, and $z$ outputs a $O((a \cdot g) / \beta)$-thin spanning forest in $G$ (w.r.t.~$z$) with at most $O(a + g)$ connected components.
\end{theorem}
\fi

\ifabstract
\begin{proof}[Proof of Theorem \ref{thm:main_theorem_genus-g}]
\fi
\iffull
\begin{proof}
\fi
For $\beta \geq 2$, by Lemma \ref{T_is_thin_and_spanning_genus-g}, we can find a $(24000 \cdot a \cdot g)$-thin spanning forest with at most $a + 11 g$ connected components. For $0 < \beta < 2$, the claim follows by scaling $z$ by a factor of $2/\beta$.
\end{proof}

\fi

\ifmain

\section{Thin subgraphs in nearly-embeddable graphs}
\label{sec:main_Lemma_(a,g,1,p)-case}

\subsection{$(0, g, 1, p)$-nearly embeddable graphs}\label{subsec:planar}

For the remainder of this subsection, let $\vec{G}$ be a $(0, g, 1, p)$-nearly embeddable digraph and let $G$ be its symmetrization. Let $\vec{H}$ be the single vortex of $\vec{G}$ of width $p$, attached to some face $\vec{F}$ of $\vec{G}$. Let $H$ and $F$ be the symmetrizations of $\vec{H}$ and $\vec{F}$ respectively. Let $\{B_v\}_{v\in V(F)}$ be a path-decomposition of $H$ of width $p$. Let $\vec{W}$ be a closed walk in $\vec{G}$ visiting all vertices in $V(\vec{H})$ and let $W$ be its symmetrization.
Let $z:E(G)\to \mathbb{R}_{\geq 0}$ be $\alpha$-thick, for some $\alpha\geq 2$, and $\vec{W}$-dense. Let $G'$ be the graph obtained by contracting $F$ to a single vertex $v^*$ in $G \setminus H$.

Following \cite{erickson2014near} we introduce the following notation.
For any $u,v \in V(G)$, a \emph{ribbon} $R$ between $u$ and $v$ is the set of all parallel edges $e = \{u,v\}$ such that for every $e,e' \in R$, there exists a homeomorphism between $e$ and $e'$ on the surface. Let $R'$ be a set of parallel edges in $G$. We say that an edge $e \in R'$ is central if the total weight of edges on each side of $e$ in $R'$ (containing $e$), is at least $z(R')/2$. 

We will find a $O(1)$-thin spanning forest $S$ in $G'$ (w.r.t.~$z$), with at most $g$ connected components, such that $S$ is $O(1)$-thin in $G$ (w.r.t.~$z$). We follow a similar approach to \cite{erickson2014near} to construct $S$. We apply some modifications that assure $S$ is $O(1)$-thin in $G$ (w.r.t.~$z$).

\subsubsection{The modified ribbon-contraction argument}

If $|V(G')| \leq g$, then we set $E(S) = \emptyset$ and we are done. Otherwise, let $l = |V(G')| - g$. We define two sequences of graphs $G_0, \ldots, G_l$ and $G'_0, \ldots, G'_l$, with $G_0 = G$ and $G'_0 = G'$. For each $j \in \{0, \ldots, l\}$, $G_j$ is obtained by uncontracting $v^* \in V(G'_j)$. Let $i \geq 0$ and suppose we have defined $G'_i$. Let $R_i$ be the heaviest ribbon in $G'_i$ (w.r.t.~$z$). Let $R'_i \subseteq E(G_i)$ be the corresponding set of edges in $G_i$. We contract all the edges in $R_i$ and we let $G'_{i+1}$ be the graph obtained after contracting $R_i$. We also perform the contraction in a way such that for all $i \in \{0, \ldots, l-1\}$ we have $v^{*} \in G'_i$.

Let $i \in \{0, \ldots, l-1\}$. If $R_i = \{u, v\}$ where $u, v \neq v^*$, similar to \cite{erickson2014near}, we let $e_i$ be a central edge in $R_i$ and we add $e_i$ to $S$. Otherwise, suppose that $R_i = \{u, v^*\}$ for some $u \in V(G'_i)$. If there exists an edge $e \in R_i$ with $e \in W$ or $z(e) \geq 0.1$, we let $e_i = e$ and we add it to $S$. Otherwise, we can assume that there is no edge $e \in R$ with $e \in W$ or $z(e) \geq 0.1$.

Let $Q_i$ be the set of vertices $v \in V(F)$ with an endpoint in $R'_{i}$. By the construction, $Q_i$ is a subpath of $F$. Let $v_1, v_2 \in V(F)$ be the endpoints of $Q_i$. Let $W'_i$ be the restriction of $W$ on $\bigcup_{v \in Q_i} B_v$. Let $W''_i$ be the subgraph of $W'_i$ obtained by deleting all edges $e$ with both endpoints in $B_{v_1}$ or $B_{v_2}$.

\begin{center}
\scalebox{1.0}{\includegraphics{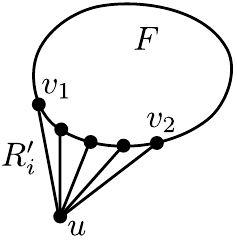}}
\end{center}

For any subgraph $C$ of $W$, we define the \emph{$i$-load} of $C$ as follows. The $i$-load of $C$ is the total weight of all edges in $R'_i$ with an endpoint in $C$. Let $C_{i}$ be the connected component of $W''_i$ with the maximum $i$-load. Let $Y_i = \{e \in R''_i : \text{ e has an endpoint in } C_{i} \}$. We let $e_i$ be a central edge in $Y_i$ and we add $e_i$ to $S$.

We set $T = S \cup W$. We will show that $T$ is a $O(p^2)$-thin spanning subgraph of $G$ (w.r.t.~$z$), with at most $g$ connected components.

\begin{lemma}\label{lem:T_is_spanning_forest_g-components}
$T$ is a spanning subgraph of $G$ with at most $g$ connected components.
\end{lemma}

\begin{proof}
By the construction, $S$ has at most $g$ connected components in $G'$. Now note that $W$ is a closed walk visiting $H$ in $G$. Therefore, all vertices of $F$ are in the same connected component in $T$. This means that $T$ has at most $g$ connected components.
\end{proof}

\begin{lemma}\label{lem:W'_has_O(p)_components}
For any $i \in \{0, \ldots, l-1\}$, $W'_i$ has at most $2  p$ connected components.
\end{lemma}

\begin{proof}
This follows immediately from the fact that $\{B_v\}_{v\in V(F)}$ is a path-decomposition of width $p$ and there is no edge in $R_i \cap W$.
\end{proof}

\begin{lemma}\label{lem:W"_has_O(p^2)_components}
For any $i \in \{0, \ldots, l-1\}$, $W''_i$ has at most $p(p+1)$ connected components.
\end{lemma}

\begin{proof}
By Lemma \ref{lem:W'_has_O(p)_components} we know that $W'_i$ has at most $2p$ connected components. $W''_i$ is obtained by deleting at most $p(p-1)$ edges of $W'_i$. Therefore, $W''_i$ has at most $p(p+1) = 2p + p(p-1)$ connected components.
\end{proof}

\begin{lemma}\label{lem:i-load_of_W'}
For any $i \in \{0, \ldots, l-1\}$, the $i$-load of $W'_i$ is at least $0.4$.
\end{lemma}

\begin{proof}
The $i$-load of $W'_i$ is $z(R'_i)$. Following \cite{erickson2014near} we know that $z(R_i) \geq 2/5$ and thus $z(R'_i) \geq 2/5$.
\end{proof}

\begin{lemma}\label{lem:i-load_of_W''}
For any $i \in \{0, \ldots, l-1\}$, the $i$-load of $W''_i$ is at least $0.2$.
\end{lemma}

\begin{proof}
By Lemma \ref{lem:i-load_of_W'} the $i$-load of $W'$ is at least $0.4$. By the construction, we have $z(\{u, v_1\}) \leq 0.1$ and $z(\{u, v_2\}) \leq 0.1$. By deleting edges with both endpoints in $B_{v_1}$ or $B_{v_2}$, we decrease the $i$-load by at most $0.2$. Therefore, the $i$-load of $W''_i$ is at least $0.2$.
\end{proof}

\begin{lemma}\label{lem:i-load_of_C_max}
For any $i \in \{0, \ldots, l-1\}$, the $i$-load of $C_{i}$ is at least $1/5p(p+1)$.
\end{lemma}

\begin{proof}
By Lemma \ref{lem:i-load_of_W''} we know that the $i$-load of $W''_i$ is at least $0.2$. By Lemma \ref{lem:W"_has_O(p^2)_components} there are at most $p(p+1)$ connected components in $W''_i$. Therefore the $i$-load of $C_{i}$ is at least $1/5p(p+1)$.
\end{proof}

\begin{lemma}\label{T'_is_thin_in_G}
There exists a constant $\beta$ such that for any $U \subseteq V(G)$, we have $|S \cap \delta(U)| \leq \beta \cdot p^2 \cdot z(\delta(U))$.
\end{lemma}

\begin{proof}
First we partition $S \cap \delta(U)$ into two subsets:

\begin{description}
\item{(1)}
$S_1 = \{ \{u, v\} \in S \cap \delta(U) : u, v \not\in V(F)\}$.

\item{(2)}
$S_2 = \{ \{u, v\} \in S \cap \delta(U) : v \in V(F)\}$.
\end{description}

By the construction, following \cite{erickson2014near} we have $|S_1| \leq 20 \cdot z(\delta(U))$. Let $e = \{u, v\} \in S_2$. Let $i \in \{0, \ldots, l-1\}$ be the step that we add $e$ to $S$. If $e \in W$, we can charge it to $z(e) \geq 1/2$ and we are done. Suppose $e \not\in W$. If there exists an edge $e' \in E(C_i) \cap \delta(U)$, we know that by the construction, $e'$ does not have both endpoints in $B_{v_1}$ or $B_{v_2}$. Therefore, for all $j \neq i$ we have $e' \notin E(C_j)$. Thus we can charge e to $z(e') \geq 1/2$ and we are done. Otherwise, suppose there is no edge in $E(C_i) \cap \delta(U)$. In this case, by the construction, for all $e'' \in R_i$ with an endpoint in $C_i$, we have $e'' \in \delta(U)$. Now we know that $e$ is the central edge in $Y_i$. By Lemma \ref{lem:i-load_of_C_max} we know that the $i$-load of $C_i$ is at least $1/5p(p+1)$. Therefore, we can charge $e$ to the $i$-load of $C_i$ and we get $|S_2| \leq 10 \cdot p^2 z(\delta(U))$. Therefore, we have $|S \cap \delta(U)| \leq 20 \cdot p^2 \cdot z(\delta(U))$.
\end{proof}

\begin{lemma}\label{lem:main_Lemma_(0,g,1,p)-case}
Let $\vec{G}$ be a $(0,g,1,p)$-nearly embeddable digraph, let $\vec{H}$ be its vortex, and let $\vec{W}$ be a walk in $\vec{G}$ visiting all vertices in $V(\vec{H})$.
Let $G$, $H$, and $W$ be the symmetrizations of $\vec{G}$, $\vec{H}$, and $\vec{W}$ respectively.
Let $z:E(G)\to \mathbb{R}_{\geq 0}$ be $\alpha$-thick for some $\alpha \geq 2$, and $\vec{W}$-dense.
Then there exists a polynomial time algorithm which given $\vec{G}$, $\vec{H}$, $\vec{W}$, $z$, and an embedding of $\vec{G}\setminus \vec{H}$ into a surface of genus $g$, outputs a subgraph $S\subseteq G\setminus H$, satisfying the following conditions:
\begin{description}
\item{(1)}
$W\cup S$ is a spanning subgraph of $G$ and has $O(g)$ connected components.

\item{(2)}
$W\cup S$ is $O(p^2)$-thin w.r.t.~$z$.
\end{description}
\end{lemma}

\begin{proof}
The assertion follows immediately by Lemmas \ref{lem:T_is_spanning_forest_g-components} and \ref{T'_is_thin_in_G}.
\end{proof}

\subsection{$(a,g,1,p)$-nearly embeddable graphs}

For the remainder of this subsection, let $a,g,k,p \geq 0$ and $\vec{G}$ be an $n$-vertex $(a,g,1,p)$-nearly embeddable digraph and let $G$ be its symmetrization. Let $\vec{H}$ be the single vortex of width $p$, attached to a face $\vec{F}$ of $\vec{G}$. Let $H$ and $F$ be the symmetrization of $\vec{H}$ and $\vec{F}$ respectively. Let $A \subseteq V(G)$ with $|A| = a$ be the set of apices of $G$, where $\Gamma = G \setminus (A \cup H)$ is a graph of genus $g$. Let $\vec{W}$ be a walk in $\vec{G}$ visiting all vertices in $V(\vec{H})$ and let $W$ be its symmetrization.
Let $z:E(G)\to \mathbb{R}_{\geq 0}$ be $\alpha$-thick, for some $\alpha\geq 2$, and $\vec{W}$-dense. Let $G'$ be the graph obtained by contracting $F$ to a single vertex $v^*$ in $G \setminus H$.

We follow a similar algorithm as in Section \ref{sec:thin_forests_higher_genus} to find a $O(a \cdot g + p^2)$-thin spanning forest $S$ in $G'$, with at most $O(a+g)$ connected components. We modify the algorithm such that $S$ is a $O(a \cdot g + p^2)$-thin subgraph of $G$.

We first start with $\Gamma$ and construct $\Gamma'$ the same way as in Section \ref{sec:thin_forests_higher_genus}. For each connected component $C$ of $\Gamma'$, we want to find a $O(p^2)$-thin spanning forest $T_C$, with at most $g$ connected components. Let $C_{v^*}$ be the connected component of $\Gamma'$ with $v^* \in C_{v^*}$. $C_{v^*}$ is a graph of genus at most $g$. For this component, we apply the modified ribbon-contraction argument on Subsection \ref{subsec:planar} to find $T_{C_{v^*}}$. Therefore, $T_{C_{v^*}}$ is a $O(p^2)$-thin spanning forest in $C_{v^*}$ with at most $g$ connected components. The rest of the algorithm is the same as in Section \ref{sec:thin_forests_higher_genus} and we find a $O(a\cdot g + p^2)$-thin spanning forest $S$ in $G'$, with at most $O(a+g)$ connected components. Let $T = S \cup W$.

\begin{lemma}\label{T_is_spanning_(0,g,1,p)-case}
$T$ is a spanning subgraph of $G$ with at most $O(a+g)$ connected components.
\end{lemma}

\begin{proof}
The same proof as in Lemma \ref{lem:T_is_spanning_forest_g-components} applies here. The only difference here is that $S$ has at most $O(a+g)$ connected components in $G'$. Therefore, $T$ has at most $O(a+g)$ connected components in $G$.
\end{proof}

\begin{lemma}\label{lem:T'_is_thin_(a,g,1,p)-case}
$S$ is a $O(a \cdot g + p^2)$-thin subgraph of $G$ (w.r.t.~$z$).
\end{lemma}

\iffull
\begin{proof}
Let $U \subseteq V(G)$ be a cut. Similar to Subsection \ref{subsec:planar}, we partition $S \cap \delta(U)$ into two subsets.

\begin{description}
\item{(1)}
$S_1 = \{ \{u, v\} \in S \cap \delta(U) : u, v \not\in V(F)\}$.

\item{(2)}
$S_2 = \{ \{u, v\} \in S \cap \delta(U) : v \in V(F)\}$.
\end{description}

First, by the construction and Lemma \ref{T_is_thin_and_spanning_genus-g}, we have $|S_1| \leq 24000 \cdot a \cdot g \cdot z(\delta(U))$. Now we partition $S_2$ into three subsets.

\begin{description}
\item{(1)}
$S_{21} = \{ \{a_j, v\} \in S_2 : a_j \in A, v \in V(F)\}$.

\item{(2)}
$S_{22} = \{ \{u, v\} \in S_2 : v \in V(F), u \text{ and } v \text{ are in different components of } \Gamma'\}$.

\item{(3)}
$S_{23} = \{ \{u, v\} \in S_2 : v \in V(F), u, v \in C_{v^*}\}$. 
\end{description}

By the construction, we know that $|S_{21}| \leq 1$ and $|S_{22}| \leq 7$. Also, by Lemma \ref{T'_is_thin_in_G} we have $|S_{23}| \leq 20 \cdot p^2 \cdot z'(\delta(U))$. Therefore, we have $|S \cap \delta(U)| \leq 8(24000 \cdot a \cdot g + 20 p^2) z(\delta(U))$.
\end{proof}
\fi

\begin{lemma}\label{lem:T_is_thin_(a,g,1,p)-case}
$T$ is a $O(a\cdot g + p^2)$-thin subgraph of $G$ (w.r.t.~$z$).
\end{lemma}

\begin{proof}
By Lemma \ref{lem:T'_is_thin_(a,g,1,p)-case} we know that $S$ is a $O(a\cdot g + p^2)$-thin subgraph of $G$ (w.r.t.~$z$). Now note that $z$ is $\vec{W}$-dense. Therefore, $T = S \cup W$ is a $O(a\cdot g + p^2)$-thin subgraph of $G$ (w.r.t.~$z$).
\end{proof}

We are now ready to prove the main result of this Section.

\begin{proof}[Proof of Lemma \ref{lem:main_Lemma_(a,g,1,p)-case}]
It follows by Lemmas \ref{T_is_spanning_(0,g,1,p)-case} and \ref{lem:T_is_thin_(a,g,1,p)-case}.
\end{proof}

\fi

\ifappendix

\fi

\iffull

\section{A preprocessing step for the dynamic program}
\label{sec:normalization}

\begin{definition}[Facial normalization]
Let $g\geq 0$, $p\geq 1$ and let $\vG$ be a $(0,g,1,p)$-nearly embeddable graph.
Let $\vF$ be the face on which the vortex is attached.
We say that $\vG$ is \emph{facially normalized} if the symmetrization of $\vF$ is a simple cycle and every $v\in V(\vF)$ has at most one incident edge that is not in $E(\vF)$.
\end{definition}

\begin{lemma}\label{lem:facial_normalization}
Let $g\geq 0$, $p\geq 1$ and let $\vG$ be a $(0,g,1,p)$-nearly embeddable graph.
There exists a polynomial-time computable $(0,g,1,p)$-nearly embeddable facially normalized graph $\vG'$ such that the following holds.
Let $\vH$ be the vortex in $\vG$ and let $\vH'$ be the vortex in $\vG'$.
Then $\OPT_{\vG}(V(\vH)) = \OPT_{\vG'}(V(\vH'))$.
Moreover there exists a polynomial-time algorithm which given any closed walk $\vW'$ in $\vG'$ that visits all vertices in $V(\vH')$, outputs some closed walk $\vW$ in $\vG$ that visits all vertices in $V(\vH)$ with $\cost_{\vG}(\vW)=\cost_{\vG'}(\vW')$.
\end{lemma}

\begin{proof}
Let $\vF$ be the face in $\vG$ that $\vH$ is attached to. Let $F$ be the symmetrization of $\vF$. We first construct a $(0,g,1,p)$-nearly embeddable graph $\vG''$, with a vortex $\vH''$ attached to a face $\vF''$ such that the symmetrization of $\vF''$ is a simple cycle. Initially we set $\vG'' = \vG$.
If $\vF$ is a simple cycle then there is nothing to be done.
Otherwise, suppose that $\vF$ is not a simple cycle. Therefore, $\partial F$ contains a family of simple cycles ${\cal C} = \{C_1, \ldots, C_k\}$ for some $k$, and a family of simple paths ${\cal P} =\{P_1, \ldots, P_l\}$ for some $l$, such that every $P_i \in {\cal P}$ is a path $x_1,x_2, \ldots, x_m$ where $x_1 \in C_\alpha$ and $x_m \in C_\beta$ for some $C_\alpha, C_\beta \in \cal{C}$. We allow ${\cal P}$ to contain paths of length $0$.

For every $P = x_1, x_2, \ldots, x_m \in \cal{P}$, where $x_1 \in C$ and $x_m \in C'$ for some $C, C' \in \cal{C}$, we update $\vG''$ as follows. Let $x'_1 \in V(C)$ and $x'_m \in V(C')$ be neighbors of $x_1$ and $x_m$. We first duplicate $P$ to get a new path $P' = y_1, y_2, \ldots, y_m$. For every edge $e \in E(P)$, we set the cost of the corresponding edge $e' \in P'$ equal to the cost of $e$. Also, for every $j \in \{1,\ldots, m\}$, we add two edges $(x_i , y_i)$ and $(y_i, x_i)$ to $\vG''$ with $\cost_{\vG''}(x_i, y_i) = \cost_{\vG''}(y_i, x_i) = 0$. Also, we delete edges $(x_1,x_1')$, $(x_1',x_1)$, $(x_m,x_m')$, and $(x_m',x_m)$ and we add edges $(y_1,x_1')$, $(x_1',y_1)$, $(y_m,x_m')$ and $(x_m',y_m)$ with the same cost respectively.

\begin{center}
\scalebox{0.85}{\includegraphics{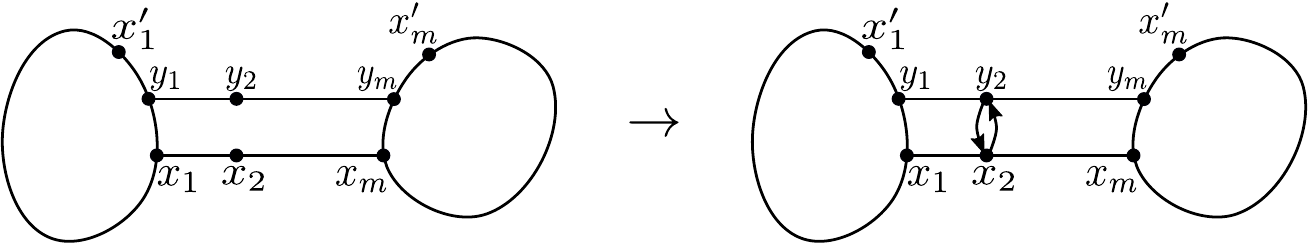}}
\end{center}

By the construction, $\vG''$ is a $(0,g,1,p)$-nearly embeddable graph, such that $\vF''$ is a simple cycle. Also suppose that $\vW''$ is a closed walk in $\vG''$ that visits all vertices in $V(\vH'')$. Then we can find a closed walk $\vW$ in $\vG$ that visits all vertices in $V(\vH)$ with $\cost_{\vG}(\vW)=\cost_{\vG''}(\vW'')$.

Now we construct a facially normalized graph $\vG'$. Initially we set $\vG' = \vG''$. For every $v \in V(\vF'')$ that has more than one incident edge in $E(\vG'') \setminus E(\vF'')$ we update $\vG'$ as follows. Let $v_{\myleft} , v_{\myright} \in V(\vF'')$ be the left and right neighbors of $v$ on $\vF''$. Let ${\cal V} = \{v_1,\ldots, v_m\} $ be the set of all neighbors of $v$ in $V(\vG'') \setminus V(\vF'')$. Let ${\cal V}' = \{v'_1, \ldots, v'_m\}$. First we delete $v$ from $\vG'$ and we add ${\cal V}'$ to $V(\vG')$. For every $(v,v_i) \in E(\vG'')$ we add $(v'_i,v_i)$ to $E(\vG')$ with $\cost_{\vG'}(v'_i,v_i)=\cost_{\vG''}(v,v_i)$. Also, for every $j \in \{1, \ldots, m-1\}$, we add $(v'_j,v'_{j+1})$ and $(v'_{j+1}, v'_j)$ to $E(\vG')$ with $\cost_{\vG'}(v'_j,v_{j+1}) = 0$. Finally we add $(v'_1, v_{\myleft})$, $(v_{\myleft}, v'_1)$, $(v'_m,v_{\myright})$ and $(v_{\myright}, v'_m)$ to $\vG'$ with the same costs as in $\vG''$.

\begin{center}
\scalebox{0.8}{\includegraphics{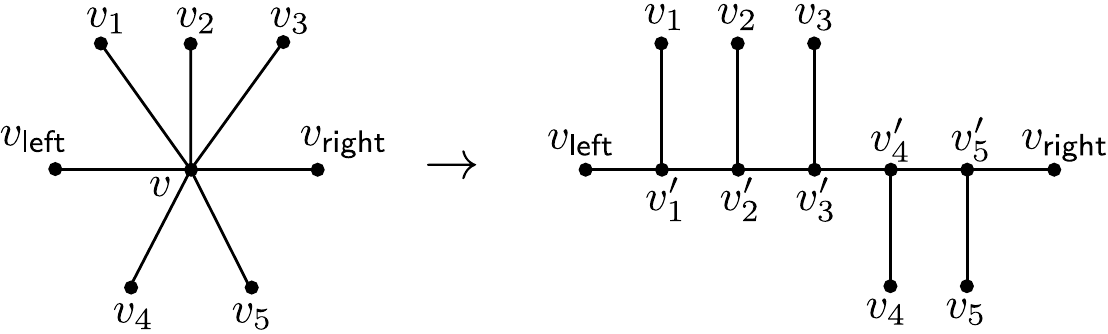}}
\end{center}

It is immediate that $\vG'$ is the desired graph.
\end{proof}

\begin{lemma}\label{lem:cross_normalization}
Let $g\geq 0$, $p\geq 1$ and let $\vG$ be a $(0,g,1,p)$-nearly embeddable graph.
There exists a polynomial-time computable $(0,g,1,p)$-nearly embeddable facially normalized graph $\vG''$ such that the following conditions hold:
\begin{description}
\item{(1)}
Let $\vH$ be the vortex in $\vG$ and let $\vH''$ be the vortex in $\vG''$.
Then $\OPT_{\vG}(V(\vH)) = \OPT_{\vG'}(V(\vH''))$.
\item{(2)}
There exists a polynomial-time algorithm which given any closed walk $\vW''$ in $\vG''$ that visits all vertices in $V(\vH'')$, outputs some closed walk $\vW$ in $\vG$ that visits all vertices in $V(\vH)$ with $\cost_{\vG}(\vW)=\cost_{\vG''}(\vW'')$.
\item{(3)}
Let $\vec{\Gamma}$ be the genus-$g$ piece of $\vG''$.
Let $\vF''$ be the face of $\vec{\Gamma}$ on which the vortex $\vH''$ is attached.
Then any $v\in V(\vec{\Gamma})\setminus V(\vF'')$ has degree at most 4.
\item{(4)}
There exists some closed walk $\vW^*$ in $\vG''$ that visits all vertices in $V(\vH'')$, with $\cost_{\vG''}(\vW^*)=\OPT_{\vG''}(V(\vH''))$, and such that every edge in $\vec{\Gamma}$ is traversed at most once by $\vW^*$.
\end{description}
We say that a graph $\vG''$ satisfying the above conditions is \emph{cross normalized}.
\end{lemma}

\begin{proof} 
We begin with computing the facially normalized graph $\vG'$ given by Lemma \ref{lem:facial_normalization}.
Clearly $\vG'$ satisfies conditions (1) and (2).

We next modify $\vG'$ so that it also satisfies (3).
This can be done as follows.
Let $\vec{\Gamma}'$ be the genus-$g$ piece of $\vG'$ and let $\vF'$ be the face on which the vortex is attached.
We replace each $v\in V(\vec{\Gamma}')\setminus V(\vF')$ of degree $d>4$ by a tree $T_v$ with $d$ leaves and with maximum degree $4$; we replace each edge incident to $v$ an edge incident to a unique leaf, and we set the length of every edge in $E(T_v)$ to 0.

It remains to modify $\vG'$ so that it also satisfies (4).
Let $\vH'$ be the vortex in $\vG'$.
Let $\vW'$ be a walk in $\vG'$ that visits all vertices in $\vH'$ with $\cost_{\vG'}(\vW') = \OPT_{\vG'}(V(\vH'))$.
We may assume w.l.o.g.~that $\vW'$ contains at most $n^2$ edges.
Thus, every vertex in $v\in V(\vec{\Gamma}')\setminus V(\vF')$ is visited at most $n^2$ times by $\vW'$.
We replace each $v\in V(\vec{\Gamma}')\setminus V(\vF')$ by a grid $A_v$ of size $3n^2\times 3n^2$, with each edge having length 0.
Each edge incident to $v$ in $\vG'$, corresponds to a unique sides of $A_v$ so that the ordering of the sides agrees with the ordering of the edges around $v$ (in $\psi$).
We replace each $(u,v)\in E(\vec{\Gamma}')$, by a matching of size $3n^2$ between the corresponding sides of $A_u$ and $A_v$, where each edge in the matching has length equal to the length of $(u,v)$.
Let $\vG''$ be the resulting graph.

We obtain the desired walk $\vW^*$ in $\vG''$ as follows.
Let $x_1,\ldots,x_{4\ell}$ be the vertices in the boundary of $A_v$, with $\ell=3n^2-1$, appearing in this order along a clockwise traversal of $A_v$, and such that $x_1$ is the lower left corner.
Then for any $i\in \{1,\ldots,4\}$, the vertices $x_{(i-1)\ell+1},\ldots,x_{(i-1)\ell+n^2}$ correspond to the copies of $v$ on the $i$-th side of $A_v$.
We traverse $\vG'$ starting at some arbitrary vertex in $V(\vH')$, and we inductively construct the walk $\vW^*$.
We consider each edge $(u,v)$ in the order that it is traversed by $\vW'$.
Suppose that $(u,v)$ is the $t$-th edge traversed by $\vW'$, for some $t\in \{1,\ldots,n^2\}$.
For each $i \in \{1,\ldots,4\}$, we let the $t$-th copy of $v$ on the $i$-th side of $A_v$ to be $x_{(i-1)\ell + t}$.
We distinguish between the following cases:
(i) If $u,v\in V(\vH')$, then $(u,v)\in E(\vG'')$ and we simply traverse $(u,v)$ in $\vG''$.
(ii) If $u\in V(\vH')$ and $v\notin V(\vH')$ then we traverse the edge in $\vG''$ that connects $u$ to the $t$-th copy of $v$ in the appropriate side of $A_v$.
(iii) If $u\notin V(\vH')$ and $v\in V(\vH')$ then we traverse the edge in $\vG''$ that connects the $t$-th copy of $u$ in the appropriate side of $A_u$ to $v$.
(iv) If $u,v\notin V(\vH')$ then we traverse the edge in $\vG''$ that connects the $t$-th copy of $u$ to the $t$-th copy of $v$ in the appropriate sides of $A_u$ and $A_v$ respectively.
Finally, for any pair of consecutive edges $(u,v)$, $(v,w)$ traversed by $\vW'$, with $v\notin V(\vH')$, we need to add a path $P$ in $A_v$ connecting two copies of $v$ in the corresponding sides of $A_v$.
Since $A_v$ is a grid of size $3n^2\times 3n^2$ this can be done so that all these paths are edge-disjoint.
More precisely, this can be done as follows.
Suppose that $(u,v)$ is the $t$-th edge traversed by $\vW'$, for some $t \in \{1,\ldots,n^2\}$.
If $P$ connects vertices $x$ and $y$ in consecutive sides of $A_v$, then we proceed as follows. We may assume w.l.o.g~that $x = x_t$ and $y = x_{\ell+t+1}$, since other cases can be handled in a similar way. We set $P$ to be the unique path starting at $x_t$, following $t+1$ horizontal edges in $A_v$, and finally following $\ell - t$ vertical edges to $x_{\ell+t+1}$.
Otherwise, if $P$ connects vertices $x$ and $y$ in opposite sides of $A_v$, then we proceed as follows.
We may assume w.l.o.g~that $x=x_t$ and $y=x_{2\ell + t +1}$, since other cases can be handled in a similar way.
We set $P$ to be the unique path starting at $x_t$, following $t+n^2$ horizontal edges in $A_v$, and then following $\ell - 2t$ vertical edges, and finally following $\ell - t - n^2 - 1$ horizontal edges to $x_{2\ell + t +1}$.
\begin{center}
\scalebox{0.45}{\includegraphics{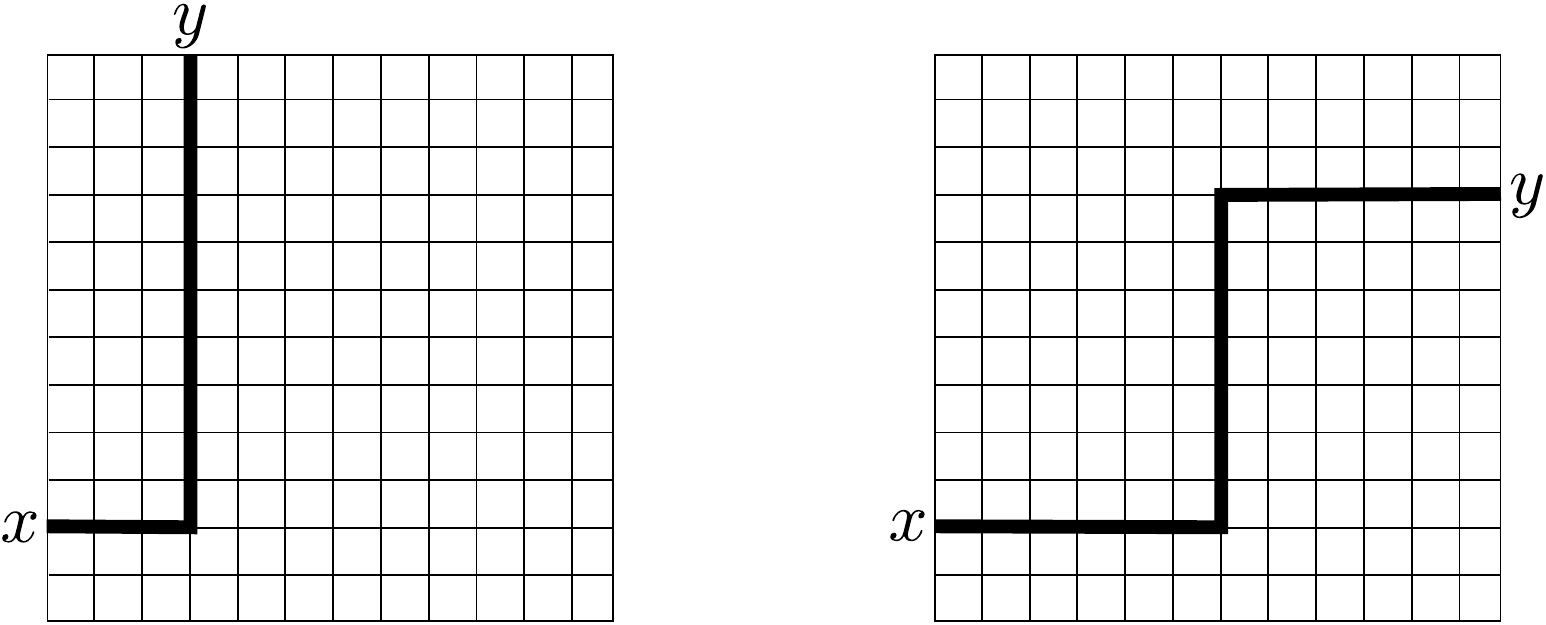}}
\end{center}
By the construction, it is immediate that all the paths $P$ constructed above are pairwise edge-disjoint, which implies that every edge in $E(\vG'')$ is visited by $\vW^*$ at most once, concluding the proof.
\end{proof}

\section{Uncrossing an optimal walk traversing a vortex}
\label{sec:uncrossing}

Let $g\geq 0$, $p\geq 1$ and
let $\vG$ be a $(0,g,1,p)$-nearly embeddable graph.
By Lemma \ref{lem:cross_normalization} we may assume w.l.o.g.~that $\vG$ is facially normalized and cross normalized.
Let $\vG'\subseteq \vG$ be the piece of genus $g$ and fix a drawing $\psi$ of $\vG'$ into a surface of genus $g$.
Let $\vH$ be the single vortex in $\vG$ and suppose that $\vH$ is attached to some face $\vF$ of $\vG'$.
Fix an optimal solution $\vW_{\OPT}$, that is a closed walk in $\vG$ that visits all vertices in $\vH$ minimizing $\cost_{\vG}(\vW)$;
if there are multiple such walks pick consistently one with a minimum number of edges.
Since $\vG$ is cross normalized we may assume w.l.o.g.~that $\vW_\OPT$ traverses every edge in $E(\vG')\setminus E(\vF)$ at most once.

\subsection{The structure of an optimal solution}

\begin{definition}[Shadow]
Let ${\cal W}$ be a collection of walks in $\vG$.
We define \emph{shadow} of ${\cal W}$ (w.r.t.~$\vG'$) to be the collection of open and closed walks obtained by restricting every walk in ${\cal W}$ on $\vG'$
(note that a walk in ${\cal W}$ can give rise to multiple walks in ${\cal W}'$, and every open walk in ${\cal W}'$ must have both endpoints in $\vF$).
\end{definition}

We say that two edje-disjoint paths $P$, $P'$ in $\vG'$ \emph{cross} (w.r.to $\psi$) if there exists $v\in V(P)\cap V(P')$ such that $v$ has degree 4 (recall that $\vG$ is cross normalized), with neighbors $u_1,\ldots,u_4$, such that the edges $\{v,u_1\},\ldots,\{v,u_4\}$ appear in this order around $v$ in the embedding $\psi$, $P$ contains the subpath $u_1,v,u_3$, and $P'$ contains the subpath $u_2,v,u_4$.
We say that two walks in $\vG'$ cross (at $v$, w.r.to $\psi$) if they contain crossing subpaths.
Finally, a walk is \emph{self-crossing} if it contains two disjoint crossing subpaths.

\begin{lemma}[Uncrossing an optimal walk of a vortex]\label{lem:uncrossing_walks_vortex}
There exists a collection ${\cal W}=\{\vec{W_1},\ldots,\vec{W_\ell}\}$ of closed walks in $\vG$ satisfying the following conditions:
\begin{description}
\item{(1)}
Every edge in $E(\vG')\setminus E(\vF)$ is traversed in total at most once by all the walks in ${\cal W}$.
\item{(2)}
$V(\vec{W_1})\cup \ldots \cup V(\vec{W_{\ell}})$ is a strongly-connected subgraph of $\vG$.
\item{(3)}
$V(\vH) \subseteq V(\vec{W_1}) \cup \ldots \cup V(\vec{W_{\ell}})$.
\item{(4)}
$\sum_{i=1}^{\ell} \cost_{\vG}(\vec{W_i}) \leq \OPT_{\vG}(V(\vH))$.
\item{(5)}
Let ${\cal W}'$ be the shadow of ${\cal W}$.
Then the walks in ${\cal W}'$ are non-self-crossing and pairwise non-crossing.
\end{description}
\end{lemma}
\begin{proof}
Initially, we set ${\cal W}=\{\vW_\OPT\}$.
Recall that since $\vG$ is cross normalized, every edge in $E(\vG')\setminus E(\vF)$ is traversed at most once by $\vW_{\OPT}$.
Clearly, this choice of ${\cal W}$ satisfies conditions (1)--(4).
We proceed to iteratively modify ${\cal W}$ until condition (4) is also satisfied, while inductively maintaining (1)--(4).

Suppose that the current choice for ${\cal W}$ does not satisfy (5).
This means that either there exist two distinct crossing walks in ${\cal W}$, or there exists some self-crossing walk in ${\cal W}$.
In either case, it follows that there exist subpaths $P$, $P'$ of the walks in ${\cal W}$ that are crossing (w.r.to $\phi$).
This means that there exists $v\in V(P)\cap V(P')$ and $e_1,e_2\in E(P)$, $e_3,e_4\in E(P')$ such that $\psi(e_1)$, $\psi(e_4)$, $\psi(e_2)$, $\psi(e_3)$ appear in this order around $\psi(v)$.
We modify $P$ and $P'$ by swapping $e_1$ and $e_3$.
It is immediate that the above operation preserves conditions (1)--(4).
Moreover, after performing the operation, the total number of crossings and self-crossings (counted with multiplicities) between the walks in ${\cal W}$ decreases by at least one. Since the original number of crossings is finite, it follows that the process terminates after a finite number of iterations.
By the inductive condition, it is immediate that when the process terminates the collection ${\cal W}$ satisfies condition (5), concluding the proof.
\end{proof}




For the remainder of this section let ${\cal W}$ and ${\cal W}'$ be as in Lemma \ref{lem:uncrossing_walks_vortex}.
Let ${\cal I}$ be a graph with $V({\cal I})={\cal W}'$ and with
\[
E({\cal I}) = \left\{\{W,W'\}\in {{\cal W}' \choose 2} : V(W)\cap V(W')\neq \emptyset\right\}.
\]

Let $\vW,\vZ$ be distinct closed walks in some digraph, and let $v\in V(\vW)\cap V(\vZ)$.
Suppose that $\vW=x_1,\ldots,x_k,v,x_{k+1},\ldots,x_{k'},x_1$ and $\vZ=y_1,\ldots,y_{r},v,y_{r+1},\ldots,y_{r'},y_1$.
Let $\vS$ be the closed walk 
$x_1,\ldots,x_k,v,y_{r+1},\ldots,y_{r'},y_1,\ldots,y_r,v,x_{k+1},\ldots,x_{k'},x_1$.
We say that $\vS$ is obtained by \emph{shortcutting} $\vW$ and $\vZ$ (at $v$).

Let ${\cal J}$ be a subgraph of ${\cal I}$.
Let ${\cal W}^{\cal J}$ be a collection of walks in $\vG$ constructed inductively as follows.
Initially we set ${\cal W}^{\cal J}={\cal W}$.
We consider all $\{W,W'\}\in E({\cal J})$ in an arbitrary order.
Note that since $\{W,W'\}\in E({\cal J})$, it follows that $W$ and $W'$ cross at some $v\in V(\vG')$.
Let $R$ and $R'$ be the walks in ${\cal W}^{\cal J}$ such that $W$ and $W'$ are sub-walks of $R$ and $R'$ respectively.
If $R\neq R'$ then we replace $R$ and $R'$ in ${\cal W}^{\cal J}$ by the walk obtained by shortcutting $R$ and $R'$ at $v$.
This completes the construction of ${\cal W}^{\cal J}$.
We say ${\cal W}^{\cal J}$ is obtained by \emph{shortcutting ${\cal W}$ at ${\cal J}$}.

\begin{lemma}\label{lem:forest}
There exists some forest ${\cal F}$ in ${\cal I}$ such that the collection of walks obtained by shortcutting ${\cal W}$ at ${\cal F}$ contains a single walk.
\end{lemma}

\begin{proof}
Let ${\cal F}$ be a forest obtained by taking a spanning subtree in each connected component of ${\cal I}$.
Let $W,W'\in {\cal W}$.
Let ${\cal W}^{\cal F}$ be obtained by shortcutting ${\cal W}$ at ${\cal F}$.
It suffices to show that $W$ and $W'$ become parts of the same walk in ${\cal W}^{\cal F}$.
By condition (2) of Lemma \ref{lem:uncrossing_walks_vortex} we have that there exists a sequence of walks $W_1,\ldots,W_t\in {\cal W}$, with $W_1=W$, $W_t=W'$, and such that for any $i\in \{1,\ldots,t-1\}$, there exists some walk $A_i\in W_i\cap \vG'$, and some walk $B_{i+1}\in W_{i+1}\cap \vG'$ such that $\{A_i,B_{i+1}\}\in E({\cal I})$.
Therefore $A_i$ and $B_{i+1}$ are in the same connected component of ${\cal I}$.
Thus there exits some tree ${\cal T}_i$ in ${\cal F}$ such that $A_i,B_{i+1}\in V({\cal T}_i)$.
It follows that after shortcutting ${\cal W}$ at ${\cal F}$, the walks $A_i$ and $B_{i+1}$ become parts of the same walk.
By induction on $i\in \{1,\ldots,t-1\}$, it follows that $A_1$ and $B_{t}$ become parts of the same walk in ${\cal W}^{\cal F}$, and thus so do $W$ and $W'$, concluding the proof.
\end{proof}

For the remainder let ${\cal F}$ be the forest given by Lemma \ref{lem:forest}.

\begin{lemma}\label{lem:leaves_of_F_face}
All leaves of ${\cal F}$ intersect $F$.
\end{lemma}

\begin{proof}
Suppose that there exists some leaf $\vW$ of ${\cal F}$ with $V(\vW)\cap V(F) = \emptyset$.
Then simply removing $\vW$ from ${\cal W}$ leaves a new collection of walks that visits all vertices in $V(\vH)$ and such that the union of all walks is a strongly-connected subgraph of $\vG$.
Thus after shortcutting all these walks we may obtain a new single walk $\vec{R}$ that visits all vertices in $V(\vH)$ with $\cost_{\vG}(\vec{R}) \leq \cost_{\vG}(\vW_{\OPT})$ and with fewer edges than $\vW_{\OPT}$, contradicting the choice of $\vW_{\OPT}$.
\end{proof}

\section{The dynamic program for traversing a vortex in a planar graph} \label{sec:vortex_planar}

For the remainder of this section let $\vG$ be a $n$-vertex $(0, 0, 1, p)$-nearly embeddable graph (that is, planar with a single vortex).
Let $\vH$ be the vortex in $\vG$ and suppose it is attached on some face $\vF$. 
Fix an optimal solution $\vW_{\OPT}$, that is a closed walk in $\vG$ that visits all vertices in $\vH$ minimizing $\cost_{\vG}(\vW)$;
if there are multiple such walks pick consistently one with a minimum number of edges.
Let $F$ be the symmetrization of $\vF$.
We present an algorithm for computing a walk traversing all vertices in $V(\vH)$ based on dynamic programming.
By Lemma \ref{lem:cross_normalization} we may assume w.l.o.g.~that $\vG$ is facially normalized and cross normalized.

Fix a path-decomposition $\{B_v\}_{v\in V(F)}$ of $\vH$ of width $p$.

Let ${\cal S}$ be a collection of walks in $\vG$.
For any $v\in V(\vG)$ we denote by $\indeg_{\cal S}(v)$ the number of times that the walks in ${\cal S}$ enter $v$; similarly, we denote by $\outdeg_{\cal S}(v)$ the number of times that the walks in ${\cal S}$ exit $v$.
We define $\vG[{\cal S}]$ to be the graph with 
$V(\vG[{\cal S}]) = \bigcup_{W \in {\cal S}} V(W)$ and $E(\vG[{\cal S}]) = \bigcup_{W \in {\cal S}} E(W)$.

\subsection{The dynamic program}
Let ${\cal P}$ be the set of all subpaths of $F$, where we allow allow for simplicity in notation that a path be closed.
Let $u$, $v$ be the endpoints of $P$.
Let 
\[
\vH_P = \vH\left[\bigcup_{x\in V(P)} B_x\right].
\]
Let ${\cal C}_P$ be the set of all possible partitions of $B_u\cup B_v$.
Let ${\cal D}^{\myin}_P = \{0,\ldots,n\}^{B_u\cup B_v}$, ${\cal D}^{\myout}_P = \{0,\ldots,n\}^{B_u\cup B_v}$, that is, every element of ${\cal D}^{\myin}_P\cup {\cal D}^{\myout}_P$ is a function $f:B_u\cup B_v \to \{0,\ldots,n\}$.
Let ${\cal A} = V(F)^2$.

\subsubsection{The dynamic programming table}
The dynamic programming table is indexed by all pairs $(P, \phi)$ where $P\in {\cal P}$ and
\[
\phi = (C, f^{\myin}, f^{\myout}, a,l,r,p) \in {\cal I}_P
\]
where
\[
{\cal I}_P = {\cal C}_P \times {\cal D}_{P}^{\myin} \times {\cal D}_{P}^{\myout} \times ( {\cal A} \cup ({\cal A} \times {\cal A}) \cup \nil) \times (V(\vG) \cup \nil) \times (V(\vG) \cup \nil) \times (V(\vG) \cup \nil).
\]
A \emph{partial solution} is a collection of walks in $\vG$.

We say that a partial solution ${\cal S}$  is \emph{compatible} with $(P,\phi)$ if the following conditions are satisfied:
\begin{description}
\item{(T1)}
For every $x\in V(\vH_P)$ there exists some walk in ${\cal S}$ that visits $x$.
That is $V(\vH_P) \subseteq \bigcup_{Q\in {\cal S}} V(Q)$.

\item{(T2)}
If $a\neq \nil$ and $a \in {\cal A}$, let $a=(u',v')$.
Let $Q_1$ be the shortest path from $u'$ to $l$ in $\vG$.
Let $Q_2$ be the shortest path from $l$ to $r$ in $\vG$.
Let $Q_3$ be the shortest path from $r$ to $v'$ in $\vG$.
Let $Q^*_1$ be the walk from $u'$ to $v'$ obtained by the concatenation of $Q_1$, $Q_2$ and $Q_3$.
Then $Q^*_1$ is a sub-walk of some walk in ${\cal S}$.
We refer to $Q^*$ as the \emph{grip} of ${\cal S}$, and in this case we say that it is an \emph{unbroken grip}.
If $a\in V(P)^2$ then we say that the unbroken grip is \emph{closed} and otherwise we say that it is \emph{open}.
Otherwise, if $a \in ({\cal A} \times {\cal A})$, let $a = ((u'_1,v'_1),(u'_2,v'_2))$.
Let $Q'_1$ be the shortest path from $u'_1$ to $l$ in $\vG$.
Let $Q'_2$ be the shortest path from $l$ to $v'_1$ in $\vG$.
Let $Q^*_2$ be the path from $u'_1$ to $v'_1$ obtained by the concatenation of $Q'_1$ and $Q'_2$.
Let $Q''_1$ be the shortest path from $u'_2$ to $r$ in $\vG$.
Let $Q''_2$ be the shortest path from $r$ to $v'_2$ in $\vG$.
Let $Q^*_3$ be the path from $u'_2$ to $v'_2$ obtained by the concatenation of $Q''_1$ and $Q''_2$.
Then $Q^*_2$ and $Q^*_3$ are sub-walks of some walks in ${\cal S}$.
We refer to $(Q^*_2,Q^*_3)$ as the \emph{broken grip} of ${\cal S}$.


\item{(T3)}
If $a=\nil$ or $a \in {\cal A}$, then every open walk in ${\cal S}$ has both endpoints in $B_u \cup B_v$, except possibly for one walk $W \in {\cal S}$ that contains the grip as a sub-walk.
If $a =(Q_1,Q_2) \in ({\cal A} \times {\cal A})$, then every open walk in ${\cal S}$ has both endpoints in $B_u \cup B_v$, except possibly for at most two walks $W_1,W_2 \in {\cal S}$ that contain $Q_1$ and $Q_2$ as sub-walks (note that $Q_1$ and $Q_2$ might be sub-walks of the same walk in ${\cal S}$).

\item{(T4)}
For all $x\in B_u\cup B_v$ we have
\[
f^{\myin}(x) = \indeg_{{\cal S}}(x) \text{~~~~ and ~~~~} f^{\myout}(x) = \outdeg_{{\cal S}}(x).
\]

\item{(T5)}
For any $x,y\in B_u\cup B_v$ we have that if $x$ and $y$ are in the same set of the partition $C$ then they are in the same weakly-connected component of $\bigcup_{W\in {\cal S}} W$.
Moreover for any $z\in V(\vH_P)$ there exists $z'\in B_u\cap B_v$ such that $z$ and $z'$ are in the same weakly-connected component of $\bigcup_{W\in {\cal S}} W$.
\end{description}

\subsubsection{Merging partial solutions}
We compute the values of the dynamic programming table inductively as follows.
Let $P,P_1,P_2\in {\cal P}$
such that $E(P_1)\neq \emptyset$, $E(P_2)\neq \emptyset$, $E(P_1)\cap E(P_2)=\emptyset$, and $P=P_1\cup P_2$.
Let $u\in V(P_1)$, $w\in V(P_1)\cap V(P_2)$, $v\in V(P_2)$ such that $u,w$ are the endpoints of $P_1$ and $w,v$ are the endpoints of $P_2$.
Let
\[
\phi = (C, f^{\myin}, f^{\myout}, a,l,r,p) \in {\cal I}_P,
\]
\[
\phi_1 = (C_1, f_1^{\myin}, f_1^{\myout}, a_1,l_1,r_1,p_1) \in {\cal I}_{P_1},
\]
\[
\phi_2 = (C_2, f_2^{\myin}, f_2^{\myout}, a_2,l_2,r_2,p_2) \in {\cal I}_{P_2}.
\]

For any $i\in \{1,2\}$ let ${\cal S}_i$ be a partial solution that is compatible with $(P_i,\phi_i)$.
We proceed to compute a collection of walks ${\cal S}$ that is compatible with $(P,\phi)$.
This is done in phases, as follows:

\begin{description}
\item{\textbf{Merging phase 1: Joining the walks.}}
We check that for all $x\in B_w$ we have
$f_1^{\myin}(x) = f_2^{\myout}(x)$ and $f_2^{\myin}(x) = f_1^{\myout}(x)$.
If not then the merging procedure return $\nil$.
For any $x\in B_w$ and for any $i\in \{1,2\}$ let $E_i(x)^{\myin}$ (resp.~$E_i(x)^{\myout}$) be the multiset of all edges in all walks in ${\cal S}_i$ that are incoming to (resp.~outgoing from) $x$ counted with multiplicities.
Since
$f_1^{\myout}(x)=f_2^{\myin}(x)$
and
$f_2^{\myout}(x)=f_1^{\myin}(x)$
it follows that
$|E_1^{\myin}(x)\cup E_2^{\myin}(x)|=|E_1^{\myout}(x) \cup E_2^{\myout}(x)|$.
Pick an arbitrary bijection
$\sigma_x:E_1^{\myin}(x) \cup E_2^{\myin} \to E_1^{\myout}(x) \cup E_2^{\myout}(x)$.
We initially set ${\cal S}={\cal S}_1\cup {\cal S}_2$.
For each $x\in B_w$ we proceed as follows.
For each $e\in E_1^{\myin(x)}\cup E_2^{\myin}(x)$ we modify the walk traversing $e$ so that immediately after traversing $e$ it continues with the walk traversing $\sigma_x(e) \in E_1^{\myout}\cup E_2^{\myout}$.

\item{\textbf{Merging phase 2: Updating the grip.}}
We check that at least one of the following conditions is satisfied:
\begin{description}
\item{(1)}
Suppose that
$a= l = r = p = a_1 = l_1 = r_1 = p_1 =a_2 = l_2 = r_2 = p_2=\nil$.
Then there is nothing to do.

\item{(2)}
Suppose that 
$a_1= l_1 = r_1 = p_1 = \nil$, $l=l_2$, $r=r_2$, $p=p_2$ and $a=a_2=(u_2^*,v_2^*) \in {\cal A}$ with $\{u_2^*,v_2^*\}\cap V(P_1) \subseteq \{u,w\}$, or $a_2=l_2 = r_2 = p_2 = \nil$, $l=l_1$, $r=r_1$, $p=p_1$ and $a=a_1=(u_1^*,v_1^*) \in {\cal A}$ with $\{u_1^*,v_1^*\}\cap V(P_2) \subseteq \{w,v\}$.
Then there is nothing to do.

\item{(3)}
Suppose that 
$a_1\neq \nil$, $a_2\neq \nil$, and $a\neq \nil$.
Suppose $a_1=a_2=a=(u^*, v^*) \in {\cal A}$, 
with $a\in V(P_1)\times V(P_2)$ or $a\in V(P_2)\times V(P_1)$.
Suppose $l_1=l_2=l$, $r_1=r_2=r$ and $p_1=p_2=p$.
Then we proceed as follows to ensure that (T2) holds.
We may assume w.l.o.g.~that $a\in V(P_1)\times V(P_2)$ since the remaining case can be handled in a similar way.
Let $Q^*$ be the grip between $u^*$ and $v^*$ in $\vG$.
It follows by (T2) that for any $i\in \{1,2\}$ there exists a walk $\vW_i\in {\cal S}_i$ that contains $Q^*$ as a sub-walk.
It follows by the definition of the merging phase 1 that for any $i\in \{1,2\}$ there exists $\vW_i'\in {\cal S}$ that contains $Q^*$ as a sub-walk.
We will modify ${\cal S}$ in order to ensure that (T2) holds.
We remove $Q^*$ from $\vW_2'$ and we merge $\vW_1'$ with $\vW_2'\setminus Q^*$ (via concatenation).

\begin{center}
\scalebox{0.95}{\includegraphics{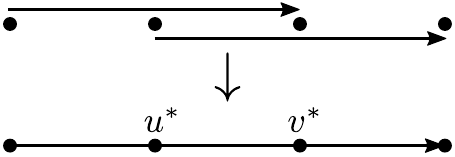}}
\end{center}

\item{(4)}
Suppose that 
$a_1, a_2, a \in {\cal A}$.
Suppose that $a_1=(u_1^*, v_1^*)$, 
$a_2=(u_2^*, v_2^*)$,
$a=(u^*, v^*)$,
with $u^*\in \{u_1^*, u_2^*\}$,
and $v^*\in \{v_1^*, v_2^*\}$.
Suppose that $l_1 = l$, $l_2=r_2=r$ and $p_1=p_2=p$, or $l_2 = l$, $l_1=r_1=r$ and $p_1=p_2=p$, or $l=r=p_1=p_2$ and $l_2=r_2$, or $l=r=p_1=p_2$ and $l_1=r_1$.
Then we proceed as follows to ensure that (T2) holds.
We may assume w.l.o.g.~that $a=(u_2^*,v_1^*)$, $l_1 = l$, $l_2=r_2=r$ and $p_1=p_2=p$ since the other cases can be handled in a similar way.
For any $i\in \{1,2\}$ 
let $Q_i^*$ be the grip of ${\cal S}_i$.
It follows by (T2) that for any $i\in \{1,2\}$ there exists a walk $\vW_i\in {\cal S}_i$ that contains $Q_i^*$ as a sub-walk.
It follows that for any $i\in \{1,2\}$ there exists $\vW_i'\in {\cal S}$ that contains $Q_i^*$ as a sub-walk.
We will modify ${\cal S}$ in order to ensure that (T2) holds.
If $a=a_1$ or $a=a_2$ then there is nothing left to do.
Otherwise, let $R_1'$ be the shortest path in $\vG$ from $u_1^*$ to $r_1$.
Let $R_1''$ be the shortest path in $\vG$ from $r_1$ to $l_2$.
Let $R_1'''$ be the shortest path in $\vG$ from $l_2$ to $v_2^*$.
Let $R_1^*$ be the path in $\vG$ from $u_1^*$ to $v_2^*$ obtained by concatenation of $R_1'$, $R_1''$ and $R_1'''$.
Let $R_2'$ be the shortest path in $\vG$ from $u_2^*$ to $p$.
Let $R_2''$ be the shortest path in $\vG$ from $p$ to $l_1$.
Let $R_2'''$ be the shortest path in $\vG$ from $l_1$ to $v_1^*$.
Let $R_2^*$ be the path in $\vG$ from $u_2^*$ to $v_1^*$ obtained by concatenation of $R_2'$, $R_2''$ and $R_2'''$.
We remove $Q_1^*$ and $Q_2^*$ from $\vW_1'$ and $\vW_2'$ and replace them by $R_1^*$ and $R_2^*$.

\begin{center}
\scalebox{0.65}{\includegraphics{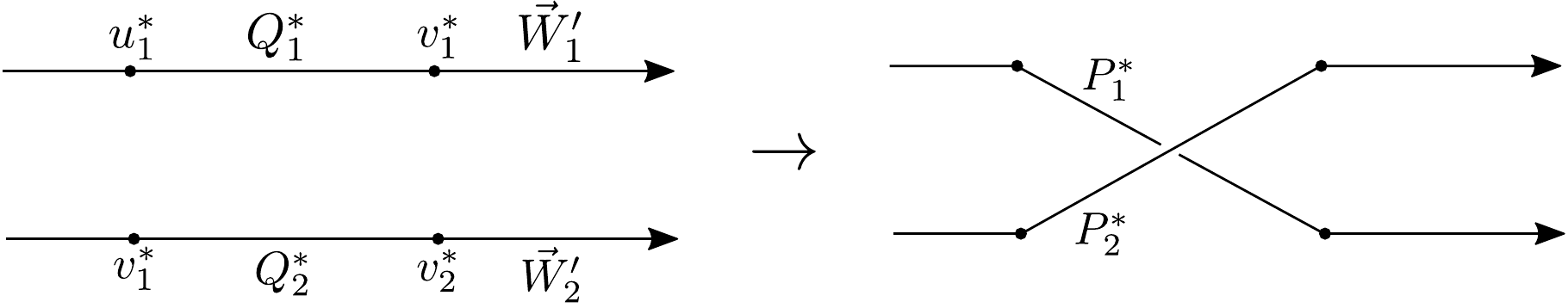}}
\end{center}

\item{(5)}
Suppose that 
$a_1=\nil$ and $a=a_2 \in ({\cal A} \times {\cal A})$, or $a_2=\nil$ and $a=a_1 \in ({\cal A} \times {\cal A})$.
Then there is nothing to do.

\item{(6)}
Suppose that
$a_1 \in {\cal A}$, $a_2 \in {\cal A}$ and $a \in ({\cal A} \times {\cal A})$.
Suppose $a_1 = (u_1,v_1)$, $a_2 = (u_2,v_2)$ and $a=((u',v'),(u'',v''))$, with $v' \in\{v_1,v_2\}$ and $u'' \in\{u_1,u_2\}$. We may assume w.l.o.g~that $a=((u',v_1),(u_2,v''))$.
Suppose that $l=l_1=r_1$, $r=l_2=r_2$ and $p=p_1=p_2$.
For any $i \in \{1,2\}$, let $Q_i$ be the grip of ${\cal S}_i$.
It follows by (T2) that for any $i\in \{1,2\}$ there exists a walk $\vW_i\in {\cal S}_i$ that contains $Q_i$ as a sub-walk.
It follows that for any $i\in \{1,2\}$ there exists $\vW_i'\in {\cal S}$ that contains $Q_i$ as a sub-walk.
Let $Q'_1$ be the shortest path in $\vG$ from $u_1$ to $l_1$.
Let $Q'_2$ be the shortest path in $\vG$ from $l_1$ to $l_2$.
Let $Q'_3$ be the shortest path in $\vG$ from $l_2$ to $v_2$.
Let $Q^*_1$ be the path in $\vG$ from $u_1$ to $v_2$ obtained by concatenation of $Q'_1$, $Q'_2$ and $Q'_3$.
Let $Q''_1$ be the shortest path in $\vG$ from $u'$ to $l_1$.
Let $Q''_2$ be the shortest path in $\vG$ from $l_1$ to $v_1$.
Let $Q^*_2$ be the path in $\vG$ from $u'$ to $v_1$ obtained by concatenation of $Q''_1$ and $Q''_2$.
Let $Q'''_1$ be the shortest path in $\vG$ from $u_2$ to $l_2$.
Let $Q'''_2$ be the shortest path in $\vG$ from $l_2$ to $v''$.
Let $Q^*_3$ be the path in $\vG$ from $u_2$ to $v''$ obtained by concatenation of $Q'''_1$ and $Q'''_2$.
Then we remove $Q_1$ and $Q_2$ from $\vW_1'$ and $\vW_2'$, and replace them by $Q^*_1$, $Q^*_2$ and $Q^*_3$.

\begin{center}
\scalebox{0.7}{\includegraphics{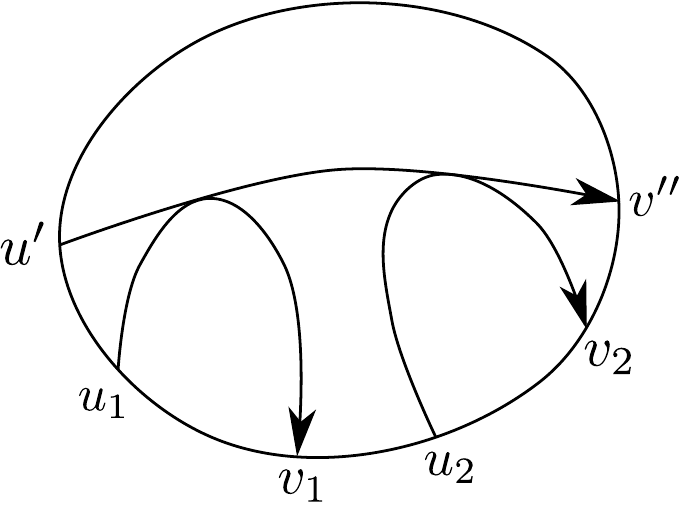}}
\end{center}

\item{(7)}
Suppose that
$a_1 \in ({\cal A} \times {\cal A})$, $a_2 \in {\cal A}$ and $a \in ({\cal A} \times {\cal A})$, or $a_1 \in {\cal A}$, $a_2 \in ({\cal A} \times {\cal A})$ and $a \in ({\cal A} \times {\cal A})$. We may assume w.l.o.g~that $a_1 \in ({\cal A} \times {\cal A})$, $a_2 \in {\cal A}$ and $a \in ({\cal A} \times {\cal A})$.
Suppose that $a_1 = ((u_1,v_1),(u'_1,v'_1))$, $a_2 = (u_2,v_2)$ and $a = ((u,v),(u',v'))$, with $(u,v) = (u_1,v_1)$,
$u' \in \{u_1,u_2\}$, and $v' \in \{v_1,v_2\}$, or $(u,v) = (u_2,v_2)$,
$u' \in \{u_1,u_2\}$, and $v' \in \{v_1,v_2\}$. We may assume w.l.o.g~that $(u,v) = (u_1,v_1)$ and $(u',v') = (u_2,v'_1)$.
Suppose that $l = l_1$, $r=l_2=r_2$ and $p=p_1=p_2$. Let $(Q_1,Q'_1)$ be the grip of ${\cal S}_1$ and let $Q_2$ be the grip of ${\cal S}_2$.
Let $R_1$ be the shortest path in $\vG$ from $u'_1$ to $r_1$.
Let $R_2$ be the shortest path in $\vG$ from $r_1$ to $r_2$.
Let $R_3$ be the shortest path in $\vG$ from $r_2$ to $v_2$.
Let $R'$ be the path in $\vG$ from $u'_1$ to $v_2$ obtained by concatenation of $R_1$, $R_2$ and $R_3$.
Let $R'_1$ be the shortest path in $\vG$ from $u_2$ to $r_2$.
Let $R'_2$ be the shortest path in $\vG$ from $r_2$ to $v'_1$.
Let $R''$ be the path in $\vG$ from $u_2$ to $v'_1$ obtained by concatenation of $R'_1$ and $R'_2$.
Then we remove $Q'_1$ and $Q_2$, and replace them by $R'$ and $R''$.


\end{description}
If none of the above holds then the merging procedure returns $\nil$.

\item{\textbf{Merging phase 3: Checking connectivity.}}
We check that condition (T5) holds for ${\cal S}$ and we return $\nil$ if it does not.
\end{description}

\begin{lemma}
If the merging procedure outputs some partial solution ${\cal S}$ then ${\cal S}$ is compatible with $\phi$.
\end{lemma}

\begin{proof}
It follows immediately by the definition of compatibility.
\end{proof}

\subsubsection{Initializing the dynamic programming table}
For all $P\in {\cal P}$ containing at most one edge and for all $\phi \in {\cal I}_P$ with $\phi=(C, f^{\myin}, f^{\myout}, a,l,r,p)$ we proceed as follows.
We enumerate all partial solutions ${\cal S}$ that are compatible with $(P,\phi)$ and have minimum cost.
Any walk in any such partial solution can intersect $\vG\setminus \vH$ only on the at most two oppositely-directed edges in $E(P)$.
Moreover there are at most $O(n^7)$ possibilities for $a$, $l$, $r$ and $p$.
Thus the enumeration can clearly be done in time $n^{O(1)}$ by ensuring that for all walks $W\in {\cal S}$, their sub-walks that do not intersect $E(P)$ are shortest paths between vertices in $B_u\cup B_v$.
The total running time of this initialization step is therefore $n^{O(p)}$.

\subsubsection{Updating the dynamic programming table}
For all $P \in {\cal P}$ containing $m > 1$ edges, and for all $P_1,P_2 \in {\cal P}$ with $E(P_1) \neq \emptyset$, $E(P_2) \neq \emptyset$, $E(P_1) \cap E(P_2) = \emptyset$ and $P_1 \cup P_2 = P$, and for all $\phi_1 \in {\cal I}_{P_1}$ and $\phi_2 \in {\cal I}_{P_2}$ we proceed as follows. Suppose that for all paths $P'$ containing $m' < m$ edges and all $\phi' \in {\cal I}_{P'}$ we have computed the partial solutions in the dynamic programming table at $(P',\phi')$. If there exist partial solutions ${\cal S}_1$ and ${\cal S}_2$ at $(P_1,\phi_1)$ and $(P_2,\phi_2)$ respectively, we call the merging process to merge ${\cal S}_1$ and ${\cal S}_2$. Suppose that the merging process returns a partial solution ${\cal S}$ at $(P,\phi)$ for some $\phi \in {\cal I}_P$.
If there is no partial solution stored currently at $(P,\phi)$ then we store ${\cal S}$ at that location.
Otherwise if there there exists a partial solution ${\cal S}'$ stored at $(P,\phi)$ and the cost of ${\cal S}$ is smaller than the cost of ${\cal S}'$ then we replace ${\cal S}'$ with ${\cal S}$.

\subsection{Analysis}

Let ${\cal W}$ be the collection of walks given by Lemma \ref{lem:uncrossing_walks_vortex}.
Let ${\cal W}'$ be the shadow of ${\cal W}$.
Let ${\cal F}$ be the forest obtained by Lemma \ref{lem:forest}.
For every connected component ${\cal T}$ of ${\cal F}$ pick some $v_{\cal T}\in V({\cal T})$ and consider ${\cal T}$ to be rooted at $v_{\cal T}$.

Let $\vG'$ be the planar piece of $\vG$, that is $\vG' = \vG\setminus (V(\vH)\setminus (\vF))$.
Fix some planar drawing $\psi$ of $G'$.
Let ${\cal D}$ be the disk with $\partial {\cal D} = \psi(F)$ with $\psi(\vG) \subset {\cal D}$.

Let $P\in {\cal P}$ with endpoints $u$, $v$, and let ${\cal T}$ be a subtree of some tree in ${\cal F}$.
We say that \emph{$P$ covers ${\cal T}$} if for all $D\in V({\cal T})$ we have $V(D)\cap V(F)\subseteq V(P)$.
We say that \emph{$P$ avoids ${\cal T}$} if for all $D\in V({\cal T})$ we have $V(D)\cap V(P)\subseteq \{u,v\}$.

\begin{definition}[Basic path]\label{defn:basic}
Let $P\in {\cal P}$.
Let $u,v \in V(P)$ be the endpoints of $P$.
We say that $P$ is \emph{basic} (w.r.t.~${\cal W}$)
if either $P \setminus \{u,v\}$ does not intersect any of the walks in ${\cal W}$ (in this case we call $P$ \emph{empty basic}) or the following holds.
There exists some tree ${\cal T}$ in ${\cal F}$ and some $D\in V({\cal T})$, 
with children $D_1, \ldots, D_k$, intersecting $D$ in this order along a traversal of $D$,
such that the following conditions are satisfied (see Figure \ref{fig:basic_path}):
\begin{description}
\item{(1)}
For any $i \in \{1, \ldots, k\}$, let ${\cal T}_{D_i}$ be the subtree of ${\cal T}$ rooted at ${D_i}$ and let ${\cal T}_D$ be the subtree of $\cal T$ rooted at $D$. Then at least one of the following two conditions is satisfied:

\begin{description}

\item{(1-1)}
$P$ covers ${\cal T}_D$ and avoids ${\cal T}\setminus {\cal T}_D$.

\item{(1-2)}
There exists $j \in \{1, \ldots, k\}$ such that for all $l \leq j$,
$P$ covers ${\cal T}_{D_l}$
and $P$ avoids ${\cal T}_D \setminus \bigcup_{m=1}^j {\cal T}_{D_m}$.
\end{description}

\item{(2)}
Let ${\cal T}'$ be a tree in ${\cal F}$ with ${\cal T}'\neq {\cal T}$.
Then either $P$ covers ${\cal T}'$ or $P$ avoids ${\cal T}'$.
\end{description}
\end{definition}

\begin{figure}
\begin{center}
\scalebox{0.8}{\includegraphics{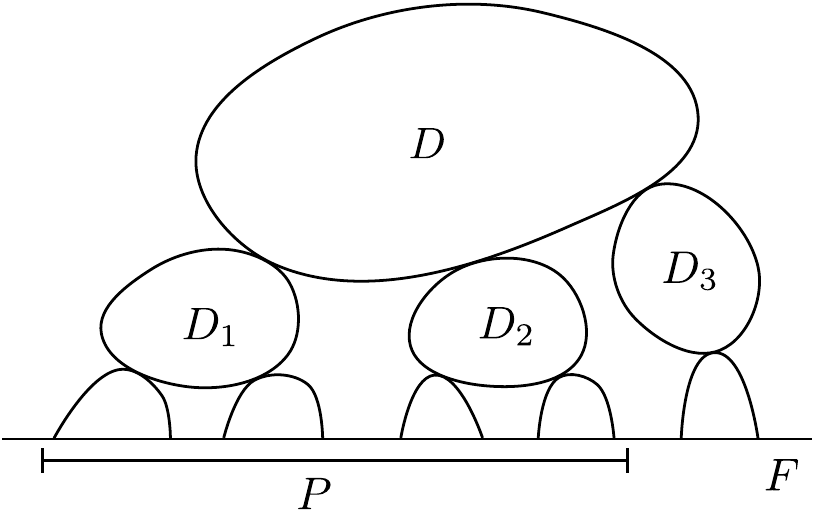}}
\end{center}
\caption{Example of a basic path.\label{fig:basic_path}}
\end{figure}

\begin{definition}[Facial restriction]\label{defn:facial_restriction}
Let $P\in {\cal P}$ be basic.
Let ${\cal W}'$ be the collection of walks obtained by restricting every $W \in {\cal W}$ on $H_P$.
If $P$ is empty, we say that ${\cal W}'$ is the \emph{$P$-facial restriction of $\cal W$}.
Suppose that $P$ is not empty.
Let $\cal T$, $D$, ${\cal T}_D$, $k$ and $j$ be as in Definition \ref{defn:basic}.
Let ${\cal F}' = \{{\cal T}' \in {\cal F} : P \text{ covers } {\cal T}'\}$ and let ${\cal R} = \bigcup_{{\cal T}' \in {\cal F}'} V({\cal T}')$.
If $P$ covers ${\cal T}_D$ and avoids ${\cal T} \setminus {\cal T}_D$, we say that ${\cal R} \cup {\cal W}' \cup V({\cal T}_D)$ is the \emph{$P$-facial restriction of $\cal W$}.
Otherwise, we say that ${\cal R} \cup {\cal W}' \cup \{D\} \cup \bigcup_{i=1}^{j} V({\cal T}_{D_i})$ is the \emph{$P$-facial restriction of $\cal W$}.
Figure \ref{fig:P-facial-restriction} depicts an example of a $P$-facial restriction.
\end{definition}

\begin{figure}
\begin{center}
\scalebox{1}{\includegraphics{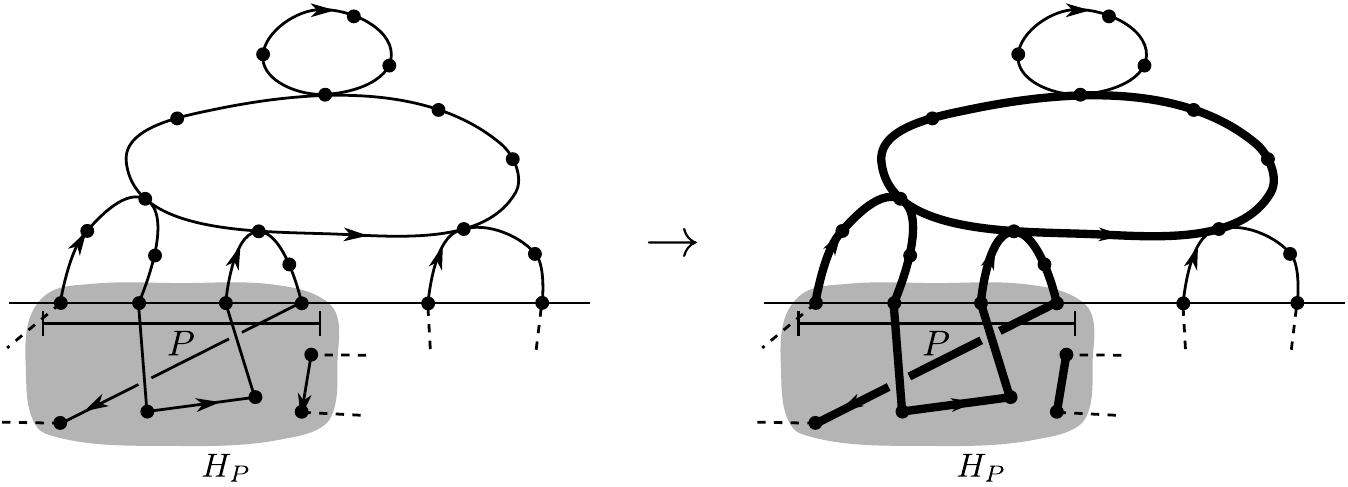}}
\caption{Part of the collection of walks in ${\cal W}$ (left) and the corresponding $P$-facial restriction of ${\cal W}$ depicted in bold (right).}
\end{center}
\end{figure}

\begin{definition}[Important walk]\label{defn:important_path}
Let $P\in {\cal P}$ be a basic path.
Let ${\cal W}'$ be the $P$-facial restriction of $\cal W$.
Let ${\cal W}''$ be the shadow of ${\cal W}'$.
Let $Q$ be a walk in $G[{\cal W}'']$.
We say that $Q$ is \emph{$P$-important} (w.r.to.~${\cal W}$) if the following conditions hold:
\begin{description}
\item{(1)}
Both endpoints of $Q$ are in $V(P)$.
\item{(2)}
$Q$ is the concatenation of walks $Q_1,\ldots,Q_{\ell}$ such that for each $i\in \{1,\ldots,\ell\}$ there exists some $W_i\in {\cal W}$ such that $Q_i$ is a sub-walk of $W_i$, and for each $j\in \{1,\ldots,\ell-1\}$ we have $\{W_j,W_{j+1}\}\in E({\cal F})$.
\end{description}
\end{definition}


\begin{proposition}\label{prop:important_walk}
For any $u,v\in V(P)$, there exists at most one $P$-important walk from $u$ to $v$.
\end{proposition}

\begin{proof}
It follows immediately by the fact that ${\cal F}$ is a forest.
\end{proof}

\begin{lemma}\label{lem:vortex_DP_induction}
Let $P\in {\cal P}$ be basic w.r.t.~${\cal W}$ with endpoints $u,v\in V(F)$, where $u=v$ if $P$ is closed.
Let ${\cal W}_P$ be the $P$-facial restriction of $\cal W$.
Let $\Gamma=\bigcup_{\vW\in {\cal W}_P} \vW$.
Let $C$ be the partition of $B_u\cup B_v$ that corresponds to the weakly-connected components of $\Gamma$.
For any $x\in B_u\cup B_v$ let 
$f^{\myin}(x) = \indeg_{{\cal W}_P}(x)$
and 
$f^{\myout}(x) = \outdeg_{{\cal W}_P}(x)$.
Then there exists some $a\in {\cal A} \cup ({\cal A} \times {\cal A}) \cup \nil$ and $l,r,p \in (V(\vG) \cup \nil)$ such that the dynamic programming table contains some partial solution ${\cal S}$ at location $(P,(C,f^\myin, f^\myout, a,l,r,p))$, 
such that $\cost_{\vG}({\cal S}) \leq \cost_{\vG}({\cal W}_P)$.
\end{lemma}

\begin{proof}
Let us first assume that $\cal F$ contains only one tree. We will deal with the more general case later on. First suppose that $P$ is an empty basic path. Let $P = x_1,x_2, \ldots, x_m$, where $x_1 = u$ and $x_m = v$. We will prove the assertion by induction on $m$. For the base case, suppose that $m=2$, and thus $P$ contains only one edge. Let $a = l = r = p =  \nil$. In this case, a partial solution $\cal S$ at location $(P,(C,f^\myin, f^\myout, a,l,r,p))$ is computed in the initialization step of the dynamic programming table and clearly we have $\cost_{\vG}({\cal S}) \leq \cost_{\vG}({\cal W}_P)$ and we are done.
Now suppose that $m > 2$ and we have proved the assertion for all $m' < m$. We first decompose $P$ into two edge-disjoint paths $P_1$ and $P_2$, such that $V(P_1) \cap V(P_2) = w$ for some $1 < j < m$ and $w = x_j$.
For $i \in \{1,2\}$ let ${\cal W}_{P_i}$ be the $P_i$-facial restriction of $\vW_{\OPT}$ and let $\Gamma_i=\bigcup_{\vW\in {\cal W}_{P_i}} \vW$.
Let $C_1$ be the partition of $B_u\cup B_w$ that corresponds to the weakly-connected components of $\Gamma_1$ and let $C_2$ be the partition of $B_w\cup B_v$ that corresponds to the weakly-connected components of $\Gamma_2$.
For any $x\in B_u\cup B_w$ let 
$f_1^{\myin}(x) = \indeg_{{\cal W}_{P_1}}(x)$
and 
$f_1^{\myout}(x) = \outdeg_{{\cal W}_{P_1}}(x)$. Also for any $x\in B_w\cup B_v$ let 
$f_2^{\myin}(x) = \indeg_{{\cal W}_{P_2}}(x)$
and 
$f_2^{\myout}(x) = \outdeg_{{\cal W}_{P_2}}(x)$.
By the induction hypothesis, there exists partial solutions ${\cal S}_1$ and ${\cal S}_2$ that are compatible with $(P_1,(C_1,f_1^\myin, f_1^\myout, \nil,\nil,\nil,\nil))$ and $(P_2,(C_2,f_2^\myin, f^\myout, \nil,\nil,\nil,\nil))$ respectively, and we have $\cost_{\vG}({\cal S}_1) \leq \cost_{\vG}({\cal W}_{P_1})$, $\cost_{\vG}({\cal S}_2) \leq \cost_{\vG}({\cal W}_{P_2})$.
The dynamic program will merge ${\cal S}_1$ and ${\cal S}_2$ to get the desired $\cal S$. Note that by the construction, for every $x \in B_w$ we have $f_1^{\myin}(x) = f_2^{\myout}(x)$ and $f_2^{\myin}(x) = f_1^{\myout}(x)$. Therefore, by the first merging phase, we can merge walks in ${\cal S}_1$ and ${\cal S}_2$. Also, we let $a=l=r=p=\nil$ and we proceed the second phase of merging. Finally, by the construction and definition of $P$-facial restriction, (T5) holds and the merging process returns a partial solution $\cal S$ compatible with $(P,(C,f^\myin, f^\myout, a,l,r,p))$ with $\cost_{\vG}({\cal S}) \leq \cost_{\vG}({\cal W}_P)$, as desired.

Now suppose that $P$ is not empty basic. Let ${\cal T}$, $D$, ${\cal T}_D$, $k$ and $j$ be as in Definition \ref{defn:basic}.
We will prove the assertion by induction on $\cal T$. For the base case, suppose that $D$ is a leaf of ${\cal T}$.
Suppose that $P = x_1,x_2, \ldots, x_m$ for some $m > 0$, where $x_1 = u$ and $x_m = v$, and $D$ is a (possibly closed) walk from $x_i\in V(P)$ to $x_j\in V(P)$. We may assume w.l.o.g.~that $i \leq j$. There are some possible cases here based on $i$ and $j$. First suppose that $i > 1$, $j < m$ and $j-i \geq 3$. In this case, let $P_1 = x_1,\ldots,x_i$, $P_2 = x_i,x_{i+1}$, $P_3 = x_{i+1}, \ldots, x_{j-1}$, $P_4 = x_{j-1},x_j$ and $P_5=x_j,\ldots,x_m$. Let $P_6 = P_1 \cup P_2$, $P_7 = P_6 \cup P_3$ and $P_8 = P_7 \cup P_4$. For $i \in \{1,2,3,4,5,6,7,8\}$ let ${\cal W}_{P_i}$ be the $P_i$-facial restriction of $\vW_{\OPT}$ and let $\Gamma_i=\bigcup_{\vW\in {\cal W}_{P_i}} \vW$. We also define $C_i$, $f_i^{\myin}$ and $f_i^{\myout}$ as in the previous case.
Note that $P_1$, $P_3$ and $P_5$ are empty basic paths. We let $a_1 = a_3 = a_5 =\nil$, $l_1 = l_3 = l_5 =\nil$, $r_1 = r_3 = r_5 =\nil$, $p_1 = p_3 = p_5 =\nil$ and thus we can find partial solutions ${\cal S}_1$, ${\cal S}_3$ and ${\cal S}_5$ compatible with $(P_1,(C_1,f_1^\myin, f_1^\myout, a_1,l_1,r_1,p_1))$, $(P_3,(C_3,f_3^\myin, f_3^\myout, a_3,l_3,r_3,p_3))$ and $(P_5,(C_5,f_5^\myin, f_5^\myout, a_5,l_5,r_5,p_5))$ respectively. We also have $\cost_{\vG}({\cal S}_1) \leq \cost_{\vG}({\cal W}_{P_1})$, $\cost_{\vG}({\cal S}_3) \leq \cost_{\vG}({\cal W}_{P_3})$ and $\cost_{\vG}({\cal S}_5) \leq \cost_{\vG}({\cal W}_{P_5})$.
Let $a_2 = a_4 = a_6 = a_7 = a_8 =  (x_i,x_j)$. If $D$ does not have a parent in $\cal T$, then we let $l_2=l_4 = l_6=l_7=l_8 = r_2=r_4 = r_6=r_7=r_8 = p_2=p_4 = p_6=p_7=p_8 \in V(D)$ to be an arbitrary vertex of $D$. Otherwise, suppose that $D$ has a parent $D'$ in $\cal T$. If $D'$ does not have a parent in $\cal T$, then we let $l_2=l_4 = l_6=l_7=l_8 = r_2=r_4 = r_6=r_7=r_8 = p_2=p_4 = p_6=p_7=p_8 \in V(D) \cap V(D')$. Otherwise, suppose that $D'$ has a parent $D''$ in $\cal T$. In this case, we let $l_2=l_4 = l_6=l_7=l_8 = r_2=r_4 = r_6=r_7=r_8 \in V(D) \cap V(D')$ and $p_2=p_4 = p_6=p_7=p_8 \in V(D') \cap V(D'')$.
Therefore, by computing the initialization step, we can find a partial solution ${\cal S}_2$ compatible with $(P_2,(C_2,f_2^\myin, f_2^\myout, a_2,l_2,r_2,p_2))$ and a partial solution ${\cal S}_4$ compatible with $(P_4,(C_4,f_4^\myin, f_4^\myout, a_4,l_4,r_4,p_4))$, and we have $\cost_{\vG}({\cal S}_2) \leq \cost_{\vG}({\cal W}_{P_2})$ and $\cost_{\vG}({\cal S}_4) \leq \cost_{\vG}({\cal W}_{P_4})$. Now by merging ${\cal S}_1$ and ${\cal S}_2$, we get a partial solution ${\cal S}_6$ compatible with $(P_6,(C_6,f_6^\myin, f_6^\myout, a_6,l_6,r_6,p_6))$. By merging ${\cal S}_6$ and ${\cal S}_3$, we get a partial solution ${\cal S}_7$ compatible with $(P_7,(C_7,f_7^\myin, f_7^\myout, a_7,l_7,r_7,p_7))$. By merging ${\cal S}_7$ and ${\cal S}_4$, we get a partial solution ${\cal S}_8$ compatible with $(P_8,(C_8,f_8^\myin, f_8^\myout, a_8,l_8,r_8,p_8))$, and finally by merging ${\cal S}_8$ and ${\cal S}_5$, we get the desired partial solution ${\cal S}$ compatible with $(P,(C,f^\myin, f^\myout, a,l,r,p))$ with $\cost_{\vG}({\cal S}) \leq \cost_{\vG}({\cal W}_P)$. If $i = 1$ or $j=m$, we will follow a similar approach. The only different is that instead of dividing $P$ into five paths, we divide it into four paths. Finally, the last case is when $j-i < 3$. In this case, if $i \neq j$, we define the same subpaths $P_1$, $P_2$, $P_4$ and $P_5$, and we follow a similar approach. Otherwise, suppose that $i = j$. In this case, we let $P_1 = x_1, \ldots, x_{i}$ and $P_2 = x_i, \ldots, x_m$ and by following the same approach by mering two partial solutions, we get the desired $\cal S$.

Now suppose that $D \in V({\cal T})$ is non-leaf. In this case, we prove the assertion by induction on $j$, where $j$ comes from Definition \ref{defn:basic}. Note that we perform a second induction inside the first induction. For the base case, suppose that $j=1$. In this case, $D_1$ is a child of $D$ and $P$ covers ${\cal T}_{D_1}$ and avoids ${\cal T}_D \setminus {\cal T}_{D_1}$. Therefore, by using the first induction hypothesis on $D_1$, there exists $a \in {\cal A} \cup ({\cal A} \times {\cal A}) \cup \nil$ and $l,r,p \in V(\vG) \cup \nil$ such that the dynamic programming table contains some partial solution $\cal S$ at location $(P,(C,f^\myin, f^\myout, a,l,r,p))$ with $\cost_{\vG}({\cal S}) \leq \cost_{\vG}({\cal W}_P)$. Now for the same $a$, $l$, $r$, $p$ and $\cal S$, we have that $\cal S$ is compatible with $(P,(C,f^\myin, f^\myout, a,l,r,p))$, as desired.
Now suppose that we have proved the assertion for all $ 1 \leq j' < j$. By Definition \ref{defn:basic}, there exists a basic path $P_1 \subseteq P$, where $u$ is the first vertex of $P_1$, such that for all $l \leq j-1$, $P_1$ covers ${\cal T}_{D_l}$ and avoids ${\cal T}_D \setminus (\bigcup_{m=1}^{j-1} {\cal T}_{D_m})$. Also there exists a basic path $P_2 \subseteq P$, where $v$ is the last vertex of $P_2$, such that $P_2$ covers ${\cal T}_{D_j}$ and avoids ${\cal T} \setminus {\cal T}_{D_j}$.
Let $u' \in V(F)$ and $v' \in V(F)$ be the other endpoints of $P_1$ and $P_2$ respectively. Let $P_3 \in {\cal P}$ be the path between $u'$ and $v'$ that does not contain $v$ and let $P_4 = P_1 \cup P_3$. By the construction, $P_3$ and $P_4$ are basic. For $i \in \{1,2,3,4\}$, let ${\cal W}_{P_i}$ be the $P_i$-facial restriction of $\vW_{OPT}$ and let $\Gamma_i=\bigcup_{\vW\in {\cal W}_{P_i}} \vW$. Let $C_1, C_2, C_3$ and $C_4$ be the partitions of $B_u\cup B_{u'}$, $B_{v'}\cup B_v$, $B_{u'}\cup B_{v'}$ and $B_u \cup B_{v'}$ that corresponds to the weakly connected components of $\Gamma_1$, $\Gamma_2$, $\Gamma_3$ and $\Gamma_4$ respectively.
For any $x\in B_u\cup B_{u'}$ let 
$f_1^{\myin}(x) = \indeg_{{\cal W}_{P_1}}(x)$
and 
$f_1^{\myout}(x) = \outdeg_{{\cal W}_{P_1}}(x)$,
for any $x\in B_{v'}\cup B_v$ let 
$f_2^{\myin}(x) = \indeg_{{\cal W}_{P_2}}(x)$
and 
$f_2^{\myout}(x) = \outdeg_{{\cal W}_{P_2}}(x)$,
for any $x \in B_{u'} \cup B_{v'}$, let $f_3^{\myin}(x) = \indeg_{{\cal W}_{P_3}}(x)$ and $f_3^{\myout}(x) = \outdeg_{{\cal W}_{P_3}}(x)$,
and for any $x\in B_{u}\cup B_{v'}$ let 
$f_4^{\myin}(x) = \indeg_{{\cal W}_{P_4}}(x)$
and 
$f_4^{\myout}(x) = \outdeg_{{\cal W}_{P_4}}(x)$.
By the second induction hypothesis, there exists some $a_1\in {\cal A} \cup ({\cal A} \times {\cal A}) \cup \nil$ and $l_1,r_1,p_1 \in V(\vG) \cup \nil$, such that the dynamic programming table contains some partial solution ${\cal S}_1$ at location $(P_1,(C_1,f_1^\myin, f^\myout, a_1,l_1,r_1,p_1))$, with
$\cost_{\vG}({\cal S}_1) \leq \cost_{\vG}({\cal W}_{P_1})$.
Also by the first induction hypothesis, there exists some $a_2\in {\cal A} \cup ({\cal A} \times {\cal A}) \cup \nil$ and $l_2,r_2,p_2 \in V(\vG) \cup \nil$, such that the dynamic programming table contains some partial solution ${\cal S}_2$ at location $(P_2,(C_2,f_2^\myin, f^\myout, a_2,l_2,r_2,p_2))$, with
$\cost_{\vG}({\cal S}_2) \leq \cost_{\vG}({\cal W}_{P_2})$.
Let $a_3 = l_3 = r_3 = p_3= \nil$. Let $a_4 = a_1$, $l_4 = l_1$, $r_4 = r_1$ and $p_4 = p_1$. Since $P_3$ is basic, there exists a partial solution ${\cal S}_3$ compatible with $(P_3,(C_3,f_3^\myin, f_3^\myout, a_3,l_3,r_3,p_3))$. Now we merge ${\cal S}_1$ and ${\cal S}_3$ to get a partial solution ${\cal S}_4$ compatible with $(P_4, (C_4, f_4^\myin, f_4^\myout, a_4,l_4,r_4,p_4))$. Note that for every $x \in B_{u'}$, we have $f_3^{\myin}(x) = f_1^{\myout}(x)$ and $f_3^{\myout}(x) = f_1^{\myin}(x)$. Therefore, we can apply the first merging phase. Also we have $a_4 = a_1$, $l_4 = l_1$, $r_4 = r_1$ and $p_4 = p_1$, and thus we can apply the second merging phase. Finally, by the construction (T5) holds and we get a partial solution ${\cal S}_4$ compatible with $(P_4, (C_4, f_4^\myin, f_4^\myout, a_4,l_4,r_4,p_4))$.
Now, we merge two partial solutions ${\cal S}_4$ and ${\cal S}_2$ to get the desired ${\cal S}$. Clearly, for every $x \in B_{v'}$ we have $f_4^\myin(x) = f_2^\myout(x)$ and $f_4^\myout(x) = f_2^\myin(x)$. Therefore, we can apply the first phase of merging. 
If $a_2 = l_2 = r_2 = p_2 = a_4 = l_4 = r_4 = p_4 = \nil$, then we let $a = l = r =p = \nil$.
Otherwise, if $a_4 = l_4 = r_4 = p_4 = \nil$ and $a_2 = (u_2^*,v_2^*) \in {\cal A}$ with $\{u_2^*,v_2^*\}\cap V(P_4) \subseteq \{u,v'\}$, then we let $a = a_2$, $l=l_2$, $r=r_2$ and $p=p_2$.
If $a_2 = l_2 = r_2 = p_2 = \nil$ and $a_4 = (u_4^*,v_4^*)$ with $\{u_4^*,v_4^*\}\cap V(P_2) \subseteq \{v',v\}$, then we let $a = a_4$, $l=l_4$, $r=r_4$ and $p=p_4$.
Otherwise, if $a_2 \neq \nil$, $a_4 \neq \nil$, $a_4 = a_2$, $l_4 = l_2$, $r_4 = r_2$ and $p_4 = p_2$, then we let $a = a_4$, $l=l_4$, $r=r_4$ and $p=p_4$.
Otherwise, if $a_2 = (u_2^*,v_2^*) \in {\cal A}$, $a_4 = (u_4^*,v_4^*) \in {\cal A}$, $l_4=r_4$ and $p_1 = p_2$,
then we let $a = (u^*,v^*)$, where $u^* \in \{u_2^*, u_4^*\}$ and $v^* \in \{v_2^*,v_4^*\}$, $l=l_2$, $r=r_4$ and $p=p_2$.
Otherwise, if $a_2 = \nil$ and $a_4 \in ({\cal A} \times {\cal A})$, then we let $a = a_4$, $l=l_4$, $r=r_4$ and $p=p_4$.
Otherwise, if $a_4 = \nil$ and $a_2 \in ({\cal A} \times {\cal A})$, then we let $a = a_2$, $l=l_2$, $r=r_2$ and $p=p_2$.
Otherwise, if $a_2 = (u_2^*,v_2^*) \in {\cal A}$, $a_4 = (u_4^*,v_4^*) \in {\cal A}$, $l_2=r_2$, $l_4=r_4$ and $p_2=p_4$, then we let $a=((u',v'),(u'',v''))$ where $v' \in \{v_2,v_4\}$ and $u'' \in \{u_2,u_4\}$, $l=l_2$, $r=r_2$ and $p=p_2$.
Otherwise, if $a_4 = ((u_4,v_4),(u'_4,v'_4)) \in ({\cal A} \times {\cal A})$, $a_2=(u_2,v_2) \in {\cal A}$, $l_2=r_2$ and $p_2=p_4$, then we let $a=((u_4,v_4),(u_2,v'_4))$, $l=l_4$, $r=r_2$ and $p=p_2$.
Therefore, after applying the second merging phase, we get a partial solution $\cal S$ compatible with $(P,(C,f^\myin, f^\myout, a,l,r,p))$.

Now we have to show that $\cost_{\vG}({\cal S}) \leq \cost_{\vG}({\cal W}_P)$. Let us first suppose that $D$ is a closed walk. We will deal with the case where $D$ is an open walk later on. If $a_2 = l_2 = r_2 = p_2 = \nil$ or $a_4 =l_4=r_4=p_4=\nil$, then this is immediate.
 If $a_2 = a_4$, $l_2=l_4$, $r_2=r_4$ and $p_2=p_4$, then we have $a = a_2$, $l=l_2$, $r=r_2$ and $p=p_2$, and thus this case is also immediate. Suppose that $a_2 = (u_2^*,v_2^*) \neq \nil$, $a_4 = (u_4^*,v_4^*) \neq \nil$, $l_4=r_4$ and $p_1 = p_2$, and
 $a = (u^*,v^*)$, where $u^* \in \{u_2^*, u_4^*\}$ and $v^* \in \{v_2^*,v_4^*\}$. We may assume w.l.o.g~that $a = (u_2^*, v_4^*)$. We have that $l=l_2$, $r=r_4$ and $p=p_2$.
 For $i\in \{2,4\}$ let $Q_i^*$ be the grip of ${\cal S}_i$, and let $\vW_i\in {\cal S}_i$ that contains $Q_i^*$ as a sub-walk.
Let $Y_1$ be the shortest-path in $\vG$ from $u_2^*$ to $l_2$.
Let $Y_2$ be the shortest-path in $\vG$ from $l_2$ to $l_4$. 
Let $Y_3$ be the shortest-path in $\vG$ from $l_4$ to $v_4^*$.
Let $R_1^*$ be the path in $\vG$ from $u_2^*$ to $v_4^*$ obtained by concatenation of $Y_1$, $Y_2$ and $Y_3$.
Let $Y'_1$ be the shortest-path in $\vG$ from $u_4^*$ to $r_4$.
Let $Y'_2$ be the shortest-path in $\vG$ from $r_4$ to $l_2$.
Let $Y'_3$ be the shortest-path in $\vG$ from $l_2$ to $v_2^*$.
Let $R_2^*$ be the path in $\vG$ from $u_4^*$ to $v_2^*$ obtained by concatenation of $Y'_1$, $Y'_2$ and $Y'_3$. Now by the construction, we have that $\cost_{\vG}({\cal S}) = \cost_{\vG}({\cal S}_2) + \cost_{\vG}({\cal S}_4) + \cost_{\vG}(R_1^*) + \cost_{\vG}(R_2^*) - \cost_{\vG}(Q_2^*) - \cost_{\vG}(Q_4^*)$. 
Note that by the construction, there exists a $P$-important walk from $u_2^*$ to $v_4^*$ (w.r.t.~$\cal W$) and a $P$-important walk from $u_4^*$ to $v_2^*$ (w.r.t.~$\cal W$). By Proposition \ref{prop:important_walk} these walks are unique. Let $R'_1$ be the $P$-important walk from $u_2^*$ to $v_4^*$ (w.r.t.~$\cal W$) and let $R'_2$ be the $P$-important walk from $u_4^*$ to $v_2^*$ (w.r.t.~$\cal W$). 
Also, there exists a $P_2$-important walk from $u_2^*$ to $v_2^*$ (w.r.t.~$\cal W$) and a $P_4$-important walk from $u_4^*$ to $v_4^*$ (w.r.t.~$\cal W$), and thus by Proposition \ref{prop:important_walk} these walks are unique.
Let $Q'_2$ be the $P_2$-important walk from $u_2^*$ to $v_2^*$ (w.r.t.~$\cal W$) and let $Q'_4$ be the $P_4$-important walk from $u_4^*$ to $v_4^*$ (w.r.t.~$\cal W$). By the definition of important walks, we have that $\cost_{\vG}(R_1^*) \leq \cost_{\vG}(R'_1)$, $\cost_{\vG}(R_2^*) \leq \cost_{\vG}(R'_2)$, $\cost_{\vG}(Q_2^*) \leq \cost_{\vG}(Q'_2)$ and $\cost_{\vG}(Q_4^*) \leq \cost_{\vG}(Q'_4)$. Also by the construction, we have that $\cost_{\vG}(R'_1) + \cost_{\vG}(R'_2) = \cost_{\vG}(Q'_2) + \cost_{\vG}(Q'_4)$.
Therefore, we have

\begin{align*}
\cost_{\vG}({\cal S}) &= \cost_{\vG}({\cal S}_2) + \cost_{\vG}({\cal S}_4) + \cost_{\vG}(R_1^*) + \cost_{\vG}(R_2^*) - \cost_{\vG}(Q_2^*) - \cost_{\vG}(Q_4^*) \\
 &\leq \cost_{\vG}({\cal S}_2) + \cost_{\vG}({\cal S}_4) + \cost_{\vG}(R'_1) + \cost_{\vG}(R'_2) - \cost_{\vG}(Q_2^*) - \cost_{\vG}(Q_4^*) \\
 &= (\cost_{\vG}({\cal S}_2) - \cost_{\vG}(Q_4^*)) + (\cost_{\vG}({\cal S}_4) - \cost_{\vG}(Q_2^*) ) + \cost_{\vG}(R'_1) + \cost_{\vG}(R'_2) \\
 &\leq \cost_{\vG}({\cal W}_{P_1}) + \cost_{\vG}({\cal W}_{P_2}) \\
 &= \cost_{\vG}({\cal W}_P).
\end{align*}

\begin{center}
\scalebox{1}{\includegraphics{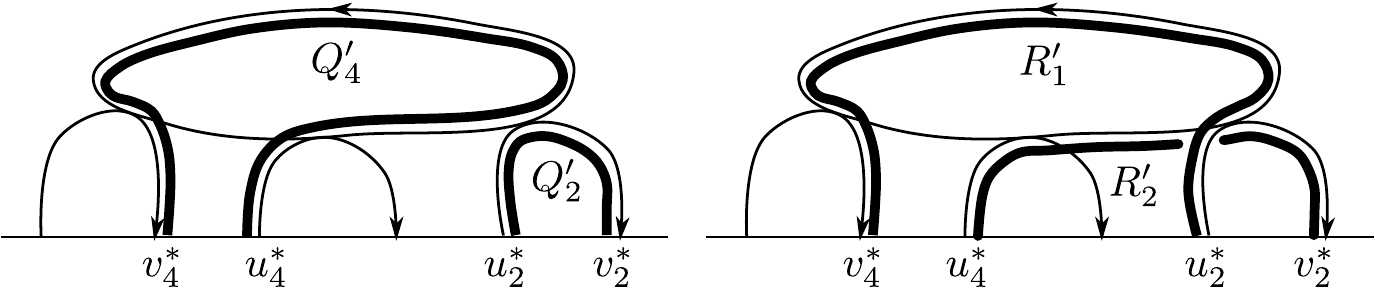}}
\end{center}

Now suppose that $D$ is an open walk. If $a_2 = \nil$ and $a_4 \in ({\cal A} \times {\cal A})$, or $a_4 = \nil$ and $a_2 \in ({\cal A} \times {\cal A})$, then this is immediate. 
If $a_2 = (u_2^*,v_2^*) \in {\cal A}$, $a_4 = (u_4^*,v_4^*) \in {\cal A}$, $l_2=r_2$, $l_4=r_4$ and $p_2=p_4$, then we have $a=((u',v'),(u'',v''))$ where $v' \in \{v_2^*,v_4^*\}$ and $u'' \in \{u_2^*,u_4^*\}$. We may assume w.l.o.g~that $v' = v_2^*$ and $u'' = u_4^*$. Then we have $l=l_2$, $r=r_2$ and $p=p_2$.
Let $R_1$ be the shortest-path in $\vG$ from $u'$ to $l_2$.
Let $R_2$ be the shortest-path in $\vG$ from $l_1$ to $v'$.
Let $R'$ be the path from $u'$ to $v'$ obtained by the concatenation of $R_1$ and $R_2$.
Let $Y_1$ be the shortest path in $\vG$ from $u''$ to $l_4$.
Let $Y_2$ be the shortest path in $\vG$ from $l_4$ to $v''$.
Let $Y'$ be the path from $u''$ to $v''$ obtained by the concatenation of $Y_1$ and $Y_2$.
Let $Z_1$ be the shortest path in $\vG$ from $u_2^*$ to $l_2$.
Let $Z_2$ be the shortest path in $\vG$ from $l_2$ to $l_4$.
Let $Z_3$ be the shortest path in $\vG$ from $l_4$ to $v_4^*$.
Let $Z'$ be the path from $u_2^*$ to $v_4^*$ obtained by the concatenation of $Z_1$, $Z_2$ and $Z_3$.
Let $Q_2^*$ be the grip of ${\cal S}_2$ and let $Q_4^*$ be the grip of ${\cal S}_4$.
By the construction, we have that $\cost_{\vG}({\cal S}) = \cost_{\vG}({\cal S}_2) + \cost_{\vG}({\cal S}_4) + \cost_{\vG}(R') + \cost_{\vG}(Y') + \cost_{\vG}(Z') - \cost_{\vG}(Q_2^*) - \cost_{\vG}(Q_4^*)$.
By the construction, there exists a $P$-important walk from $u_2^*$ to $v_4^*$, a $P$-important walk from $u'$ to $v_2^*$, and a $P$-important walk from $u_4^*$ to $v''$. Therefore by Proposition \ref{prop:important_walk}, these important walks are unique. Let $R''$, $Y''$ and $Z''$ be the important walks from $u'$ to $v'$, $u''$ to $v''$, and $u_2^*$ to $v_4^*$ respectively. Therefore, we have that

\begin{align*}
\cost_{\vG}({\cal S}) &= \cost_{\vG}({\cal S}_2) + \cost_{\vG}({\cal S}_4) + \cost_{\vG}(R') + \cost_{\vG}(Y') + \cost_{\vG}(Z') - \cost_{\vG}(Q_2^*) - \cost_{\vG}(Q_4^*) \\
 &\leq \cost_{\vG}({\cal S}_2) + \cost_{\vG}({\cal S}_4) + \cost_{\vG}(R'') + \cost_{\vG}(Y'') + \cost_{\vG}(Z'') - \cost_{\vG}(Q_2^*) - \cost_{\vG}(Q_4^*) \\
 &\leq \cost_{\vG}({\cal W}_{P_1}) + \cost_{\vG}({\cal W}_{P_2}) \\
 &= \cost_{\vG}({\cal W}_P).
\end{align*}


If $a_4 = ((u_4,v_4),(u'_4,v'_4)) \in ({\cal A} \times {\cal A})$, $a_2=(u_2,v_2) \in {\cal A}$, $l_2=r_2$ and $p_2=p_4$, then we have $a=((u_4,v_4),(u_2,v'_4))$, $l=l_4$, $r=r_2$ and $p=p_2$.
Let $(Q_4^*,Q_4^{**})$ be the grip of ${\cal S}_4$, and let $Q_2^*$ be the grip of ${\cal S}_2$.
We may assume w.l.o.g~that $Q_4^*$ is a path from $u_4$ to $v_4$, and $Q_4^{**}$ is a path from $u'_4$ to $v'_4$.
Let $Y_1$ be the shortest path in $\vG$ from $u_2$ to $l_2$.
Let $Y_2$ be the shortest path in $\vG$ from $l_2$ to $v'_4$.
Let $Y'$ be the path from $u_2$ to $v'_4$ obtained by the concatenation of $Y_1$ and $Y_2$.
Let $Z_1$ be the shortest path in $\vG$ from $u'_4$ to $r_4$.
Let $Z_2$ be the shortest path in $\vG$ from $r_4$ to $l_2$.
Let $Z_3$ be the shortest path in $\vG$ from $l_2$ to $v_2$.
Let $Z'$ be the path from $u'_4$ to $v_2$ obtained by the concatenation of $Z_1$, $Z_2$ and $Z_3$. By the construction, we have that $\cost_{\vG}({\cal S}) = \cost_{\vG}({\cal S}_2) + \cost_{\vG}({\cal S}_4) + \cost_{\vG}(Y') + \cost_{\vG}(Z') - \cost_{\vG}(Q_2^*)  - cost_{\vG}(Q_4^{**})$.
By the construction, there exists a $P$-important walk from $u'_4$ to $v_2$, and a $P$-important walk from $u_2$ to $v'_4$. Therefore by Proposition \ref{prop:important_walk}, these important walks are unique. Let $Y''$ and $Z''$ be the important walks from $u_2$ to $v'_4$, and $u'_4$ to $v_2$ respectively. Therefore we have

\begin{align*}
\cost_{\vG}({\cal S}) &= \cost_{\vG}({\cal S}_2) + \cost_{\vG}({\cal S}_4) + \cost_{\vG}(Y') + \cost_{\vG}(Z') - \cost_{\vG}(Q_2^*)  - cost_{\vG}(Q_4^{**}) \\
 &\leq \cost_{\vG}({\cal S}) = \cost_{\vG}({\cal S}_2) + \cost_{\vG}({\cal S}_4) + \cost_{\vG}(Y'') + \cost_{\vG}(Z'') - \cost_{\vG}(Q_2^*)  - cost_{\vG}(Q_4^{**}) \\
 &\leq \cost_{\vG}({\cal W}_{P_1}) + \cost_{\vG}({\cal W}_{P_2}) \\
 &= \cost_{\vG}({\cal W}_P).
\end{align*}


Now suppose that $\cal F$ contains more than one tree. Let $A = \{{\cal T}_1, \ldots, {\cal T}_m\}$ be the set of all trees in $\cal F$. For every ${\cal T} \in A$ we define the \emph{level} of ${\cal T}$, $L({\cal T})$, as follows. Let $D$ be the root of ${\cal T}$. We set $L({\cal T}) = 0$, if there exists a basic path $P'$ 
that covers ${\cal T}$ and avoids all ${\cal T}'\in {\cal F}\setminus \{{\cal T}\}$.
We call a minimal such path, a \emph{corresponding} basic path for $\cal T$ and we denote it by $P_{\cal T}$; it is immediate that there is a unique such minimal path.
Let ${\cal F}_0 = \{{\cal T} \in {\cal F}: L({\cal T}) = 0\}$.
Now for $i \geq 0$, suppose that we have defined trees of level $i$ and ${\cal F}_i$.
Suppose that $L({\cal T}) \notin \{0, \ldots, i\}$.
We set $L({\cal T}) = i+1$ if there exists a basic path $P'$ 
that covers ${\cal T}$
such that for all ${\cal T}'\in \bigcup_{j=0}^i {\cal F}_j$, $P'$ either avoids or covers ${\cal T}'$,
and for all ${\cal T}'' \in ({\cal F} \setminus \{{\cal T}\}) \setminus \bigcup_{j=0}^i {\cal F}_j$, $P'$ avoids ${\cal T}''$.
We also call a minimal such path corresponding basic for $\cal T$ and we denote it by $P_{\cal T}$.
Let ${\cal F}_{i+1} = \{{\cal T} \in {\cal F}: L({\cal T}) = i+1\}$.


We say that some ${\cal T}\in {\cal F}$ is \emph{outer-most} if there is no ${\cal T}' \in {\cal F}$ such that $L({\cal T}') > L({\cal T})$ and $P_{\cal T} \subset P_{{\cal T}'}$. 

Let us first suppose that there exists only one outer-most tree ${\cal T} \in {\cal F}$, such that $P_{\cal T} \subseteq P$. We will deal with the more general case later. Also, suppose that ${\cal T} \in {\cal F}_m$ for some $m \geq 0$. We will prove the assertion by induction on $m$. We also prove that for this case, we have $a = l = r = p = \nil$. For the base case, if $m = 0$, then by the construction, $\cal T$ is the only tree with a corresponding basic path $P_{\cal T} \subseteq P$. For this case, we have already established the assertion and we are done.
Now suppose that we have proved the assertion for all $m' < m$. Let ${\cal F}' = {\cal F} \setminus \{{\cal T}\}$. Let ${\cal T}_1, \ldots, {\cal T}_t \in {\cal F}'$ be all outer-most trees in ${\cal F}'$, such that for $j \in \{1, \ldots, t\}$ we have $P_{{\cal T}_j} \subseteq P$, and they intersect $F$ in this order. By the construction, for every $j \in \{1, \ldots, t\}$ we have $L({\cal T}_j) < m$. For every $j \in \{1, \ldots, t\}$ let $P_{{\cal T}_j}$ be a corresponding basic path for ${\cal T}_j$. By the induction hypothesis, for every $j \in \{1, \ldots, t\}$ there exists a partial solution ${\cal S}_j$ for $P_{{\cal T}_j}$. Now, we can apply the same argument when we had only one tree ${\cal T} \in {\cal F}$ for $\cal T$. Note that for each $j \in \{1, \ldots, t\}$, we have $a_j = l_j=r_j=p_j= \nil$. The only difference is that the intermediate basic paths here are not necessarily empty basic paths, and each ${\cal T}_j$ appears as an intermediate basic path, with a partial solution ${\cal S}_j$. Therefore, by following a similar approach to the previous cases and merging appropriately we get a partial solution $\cal S$, as desired.

\begin{center}
\scalebox{0.80}{\includegraphics{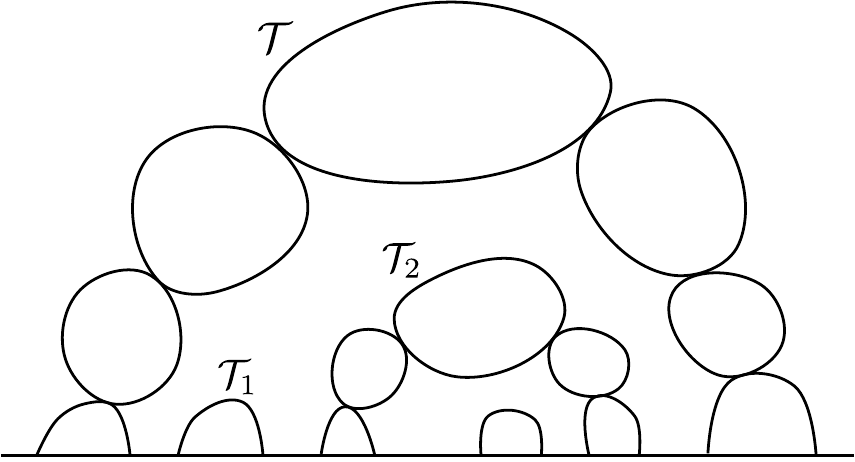}}
\end{center}

Now, suppose that there exist more than one outer-most trees ${\cal T} \in {\cal F}$, where $P_{\cal T} \subseteq P$. Let $B = \{{\cal T}_1, \ldots, {\cal T}_t\}$ be the set of all such trees, where they intersect consecutive subpaths of $P$ in this order. In this case, we prove the assertion by induction on $t$. For the base case where $t =1$, we have already proved the assertion.
Now suppose that we have proved the assertion for all $t' < t$. Let $B' = \{{\cal T}_1, \ldots, {\cal T}_{t-1}\}$. Let $P' \subseteq P$ be the subpath of $P$ such that $B'$ is the set of all outer-most trees, with $P_{{\cal T}_{j}} \subseteq P'$ for all $1 \leq j \leq t-1$. By the induction hypothesis, there exists a partial solution ${\cal S}'$ for $P'$. Let $P'' = P \subseteq P'$. By the construction, $P''$ is a corresponding basic path for ${\cal T}_t$, and thus by the induction hypothesis, there exists a partial solution $S''$ for $P''$. Therefore, by merging ${\cal S}'$ and ${\cal S}''$ we get a partial solution $\cal S$ for $P$, as desired.
\end{proof}

\begin{theorem}\label{thm:tour_visiting_a_single_vortex}
Let $\vG$ be an $n$-vertex $(0, 0, 1, p)$-nearly embeddable graph (that is, planar with a single vortex) and let $\vH$ be the vortex of $\vG$.
Then there exists an algorithm which computes a walk $\vW$ visiting all vertices in $V(\vH)$ of total length $\OPT_{\vG}(V(\vH))$ in time $n^{O(p)}$.
\end{theorem}

\begin{proof}
We have $F\in {\cal P}$ and it is immediate to check that $F$ is basic.
Since $F$ is a cycle, both the endpoitns of $F$ are some vertex $v^\circ$.
It follows by Lemma \ref{lem:vortex_DP_induction} that there exists some $\phi\in {\cal I}_F$ such that the dynamic programming table contains a partial solution ${\cal S}$ at location $(F,\phi)$.
Let $\phi=(C,f^{\myin},f^{\myout},a,l,r,p)$.
It follows by condition (T4) that for all $x\in B_{v^\circ}$ we have $\indeg_{\cal S}(x) = \outdeg_{\cal S}(x)$.
Thus by repeatedly merging pairs of walks that respectively terminate to and start from the same vertex,
we obtain a collection ${\cal S}'$ of closed walks with $\cost_{\vG}({\cal S}') = \cost_{\vG}({\cal S})$ that visit all vertices in $V(\vH)$.
By Lemma \ref{lem:forest} we can assume that $C$ is the trivial partition containing only one cluster, that is $C=\{\{B_{v^\circ}\}\}$.
Since $C$ is trivial, it follows by (T5) that all vertices in $V(\vH)$ are in the same weakly-connected component of $\bigcup_{W\in {\cal S}} W$.
Thus all vertices in $V(\vH)$ are in the same strongly-connected component of $\bigcup_{W\in {\cal S}'} W$.
Since all walks in ${\cal S}'$ are closed, by repeatedly shortcutting pairs of intersecting walks, we obtain a walk $\vW$ with that visits all vertices in $V(\vH)$ with 
$\cost_{\vG}(\vW) = \cost_{\vG}({\cal S}') = \cost_{\vG}({\cal S}) \leq \cost_{\vG}(\vW_{\OPT}) = \OPT_{\vG}$.

The running time is polynomial in the size of the dynamic programming table, which is at most $|{\cal P}| \cdot \max_{P\in {\cal P}} |{\cal I}_P| = O(n^2) \cdot p^{O(p)} \cdot n^{O(p)} \cdot O(n^2) = n^{O(p)}$.
\end{proof}

\section{The dynamic program for traversing a vortex in a bounded genus graph}
\label{sec:vortex_genus}

For the remainder of this section let $\vG$ be a $n$-vertex $(0, g, 1, p)$-nearly embeddable graph.
Let $\vH$ be the vortex in $\vG$, attached to some face $\vF$.
Let $\vG' = \vG\setminus (V(\vH)\setminus (\vF))$ and fix some embedding $\psi$ of $\vG'$ on a surface $S$ of genus $g$.
Let $F$ be the symmetrization of $\vF$.
Let $\vW_{\OPT}$ be a closed walk in $\vG$ that visits all vertices in $\vH$ with minimum $\cost_{\vG}(\vW)$.
Fix a path-decomposition $\{B_v\}_{v\in V(F)}$ of $\vH$ of width $p$.
We present a similar algorithm as in Section \ref{sec:vortex_planar} for computing a walk traversing all vertices in $V(\vH)$ based on dynamic programming.
By Lemma \ref{lem:cross_normalization} we may assume w.l.o.g.~that $\vG$ is facially normalized and cross normalized.

\subsection{The dynamic program}
Let ${\cal Q}$ be the set of all subpaths of $F$, where we allow a path to be closed. For every integer $m$, let 
\[
{\cal P}_{m} = \{ A \subseteq {\cal Q} : |A| \leq m, \text{ for every } Q,Q' \in A \text{ we have } V(Q) \cap V(Q') = \emptyset\}.
\]
And let
\[
{\cal P}_{\infty} = \{A \subseteq {\cal Q} : \text{ For every } Q,Q' \in A \text{ we have } V(Q) \cap V(Q') = \emptyset\}.
\]
 For every $P =\{Q_1, \ldots, Q_m\} \in {\cal P}_{\infty}$, let $E(P) = \bigcup_{i=1}^{m} E(Q_i)$ and let $V(P) = \bigcup_{i=1}^{m} V(Q_i)$.
For each $i \in \{1, \ldots, m\}$, let $u_i$ and $v_i$ be the endpoints of $Q_i$.
Similar to the planar case, let 
\[
\vH_P = \vH\left[\bigcup_{x\in V(P)} B_x\right].
\]
Let ${\cal B} = \bigcup_{i=1}^{m} (B_{u_i}\cup B_{v_i})$. Let ${\cal C}_P$ be the set of all possible partitions of ${\cal B}$.
Let ${\cal D}^{\myin}_P = \{0,\ldots,n\}^{{\cal B}}$, ${\cal D}^{\myout}_P = \{0,\ldots,n\}^{\cal B}$, that is, every element of ${\cal D}^{\myin}_P\cup {\cal D}^{\myout}_P$ is a function $f:{\cal B} \to \{0,\ldots,n\}$.

\subsubsection{The dynamic programming table.}
Let ${\cal P}={\cal P}_{324000 g^4}$.
With these definitions, the dynamic programming table is indexed the exact same way as in the planar case. Also,
a \emph{partial solution} is again a collection of walks in $\vG$.

We say that a partial solution ${\cal S}$  is \emph{compatible} with $(P,\phi)$ if the same conditions (T1)-(T5) as in the planar case are satisfied. The only difference here is that instead of $B_u \cup B_v$, we have $\cal B$.

\paragraph{Merging partial solutions.}

We follow a similar approach as in the planar case. Let $P = \{Q_1, \ldots, Q_m\},P_1=\{Q'_1, \ldots, Q'_{m'}\},P_2=\{Q''_1, \ldots, Q''_{m''}\} \in {\cal P}$ such that $E(P_1)\neq \emptyset$, $E(P_2)\neq \emptyset$, $E(P_1)\cap E(P_2)=\emptyset$, and $E(P) = E(P_1) \cup E(P_2)$.
Let $\phi = (C, f^{\myin}, f^{\myout}, a,l,r,p) \in {\cal I}_P$, $\phi_1 = (C_1, f_1^{\myin}, f_1^{\myout}, a_1,l_1,r_1,p_1) \in {\cal I}_{P_1}$, $\phi_2 = (C_2, f_2^{\myin}, f_2^{\myout}, a_2,l_2,r_2,p_2) \in {\cal I}_{P_2}$. 

Let ${\cal S}_1$ and ${\cal S}_2$ be partial solutions compatible with $(P_1, \phi_1)$ and $(P_2, \phi_2)$ respectively. Similar to the planar case, we compute a partial solution $\cal S$ compatible with $(P,\phi)$ as follows.

\begin{description}
\item{\textbf{Merging phase 1: Joining the walks.}}
For every $w \in V(P_1) \cap V(P_2)$, we check that for all $x \in B_w$ we have $f_1^{\myin}(x) = f_2^{\myout}(x)$ and $f_2^{\myin}(x) = f_1^{\myout}(x)$.
If not then the merging procedure returns $\nil$. Otherwise, for every $w \in V(P_1) \cap V(P_2)$ and every $x \in B_w$, we follow the exact same approach as in the planar case.

\item{\textbf{Merging phase 2: Updating the grip.}}
This phase is exactly the same as in the planar case.

\item{\textbf{Merging phase 3: Checking connectivity.}}
Similar to the planar case, we check that condition (T5) holds for ${\cal S}$ and we return $\nil$ if it does not.

\end{description}

\subsubsection{Initializing the dynamic programming table.}

For all $P \in {\cal P}_1$ with $|E(P)| \leq 1$, we follow the same approach as in the planar case.

\subsubsection{Updating the dynamic programming table.}
For all $P\in {\cal P}$ with $|E(P)| > 1$, and for all $P_1,P_2 \in {\cal P}$ with $E(P_1) \neq \emptyset$, $E(P_2) \neq \emptyset$, $E(P_1) \cap E(P_2) = \emptyset$ and $E(P_1) \cup E(P_2) = E(P)$, and for all $\phi_1 \in {\cal I}_{P_1}$ and $\phi_2 \in {\cal I}_{P_2}$ we proceed as follows. Suppose that for all $P' \in {\cal P}$ with $|E(P')| < |E(P)|$ and all $\phi' \in {\cal I}_{P'}$, we have computed the partial solutions in the dynamic programming table at $(P',\phi')$. Now similar to the planar case, if there exists partial solutions ${\cal S}_1$ and ${\cal S}_2$ at $(P_1,\phi_1)$ and $(P_2,\phi_2)$ respectively, we call the merging process to (possibly) get a partial solution ${\cal S}$ at $(P,\phi)$ for some $\phi \in {\cal I}_P$. Now similar to the planar case, if there is no partial solution at $(P,\phi)$ then we store ${\cal S}$ at $(P,\phi)$.
Otherwise if there there exists a partial solution ${\cal S}'$ stored at $(P,\phi)$ and $\cost_{\vG} ({\cal S}) < \cost_{\vG} ({\cal S}')$ then we replace ${\cal S}'$ with ${\cal S}$.

\subsection{Analysis}
Let ${\cal W}$ be the collection of walks given by Lemma \ref{lem:uncrossing_walks_vortex}.
Let ${\cal F}$ be the forest given by Lemma \ref{lem:forest}.
Let $\cal T$ be a subtree of $\cal F$. We say that $\cal T$ is \emph{trivial} if 
$\psi(F) \cup \bigcup_{D\in V({\cal T})}\psi(D)$ is contractible.
Otherwise, we say that ${\cal T}$ is \emph{non-trivial}.


Let $Q\in {\cal Q}$, and let ${\cal T}$ be a subtree of some tree in ${\cal F}$.
We define the terms \emph{$Q$ covers ${\cal T}$} and \emph{$Q$ avoids ${\cal T}$} the exact same way as in the planar case.
Let $P = \{Q_1, \ldots, Q_m\}\in {\cal P}_{\infty}$. We say that \emph{$P$ covers $\cal T$} if for all $D \in V({\cal T})$ we have $V(D) \cap V(F) \subseteq V(Q_1\cup \ldots\cup Q_m)$. We say that \emph{$P$ avoids $\cal T$} if for all $Q_i \in P$ we have that $Q_i$ avoids $\cal T$.

\begin{definition}[Basic family of paths]\label{dfn:basic_family_paths}
Let $P = \{Q_1, \ldots, Q_m\} \in {\cal P}_{\infty}$. For each $i \in \{1,\ldots, m\}$, let $u_i$ and $v_i$ be the endpoints of $Q_i$. We say that $P$ is \emph{basic} (w.r.t.~$\cal W$) if either $V(P) \setminus (\bigcup_{i=1}^{m} \{u_i, v_i\} )$ does not intersect any of the walks in $\cal W$ (in which case we call it \emph{empty basic}) or there exists ${\cal T} \in {\cal F}$ and $D \in V({\cal T})$, with children $D_1, \ldots, D_k$, intersecting $D$ in this order along a traversal of $D$, such that the exact same conditions (1) \& (2) as in Definition \ref{defn:basic} hold, and $|P|$ is minimal subject to the following:
\begin{description}
\item{(3)}
If $P$ covers ${\cal T}_D$ and avoids ${\cal T} \setminus {\cal T}_D$, let ${\cal T}[P] = {\cal T}_D$, and otherwise let ${\cal T}[P] = \bigcup_{i=1}^{j} {\cal T}_{D_i}$. For every two disjoint subtrees ${\cal T}_1$ and ${\cal T}_2$ of ${\cal T}[P]$, the following holds.
If there exists $Q_i \in P$ such that $Q_i$ covers ${\cal T}_1 \cup {\cal T}_2$ and avoids ${\cal T} \setminus ({\cal T}_1 \cup {\cal T}_2)$, then there exist edge-disjoint subpaths $Q'_i, Q''_i$ of $Q_i$ such that $Q_i = Q'_i \cup Q''_i$, $Q'_i$ covers ${\cal T}_1$ and avoids ${\cal T} \setminus {\cal T}_1$, and $Q''_i$ covers ${\cal T}_2$ and avoids ${\cal T} \setminus {\cal T}_2$.
\end{description}
Moreover if for each non-trivial tree ${\cal T}'$ in ${\cal F}$ with ${\cal T}' \neq {\cal T}$, we have that $P$ avoids ${\cal T}'$, then we say that $P$ is $\emph{elementary}$. Furthermore, if for each non-trivial tree ${\cal T}'$ in ${\cal F}$, we have that $P$ avoids ${\cal T}'$, then we say that $P$ is \emph{trivial}.
\end{definition}

\begin{center}
\scalebox{0.9}{\includegraphics{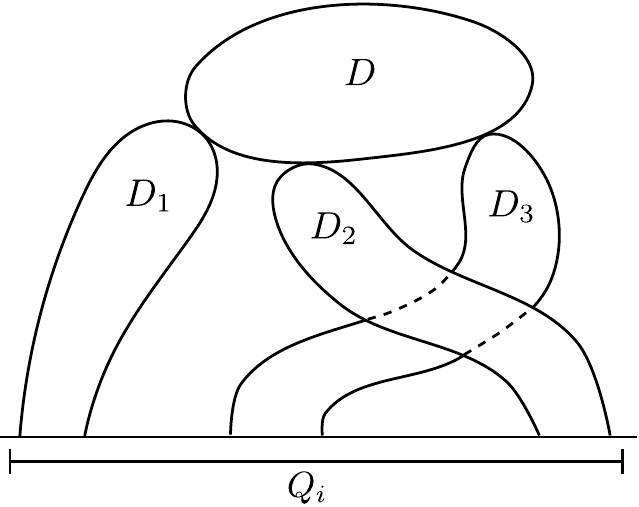}}
\end{center}

\begin{definition}[Twins]
Let $P,P'\in {\cal P}_{\infty}$.
We say that $P$ and $P'$ are \emph{twins} if for each subtree ${\cal T}$ of ${\cal F}$,
$P$ covers ${\cal T}$ if and only if $P'$ covers ${\cal T}$,
and 
$P$ avoids ${\cal T}$ if and only if $P'$ avoids ${\cal T}$.
\end{definition}

\begin{definition}[Succinctness]\label{def:succ}
Let $P = \{Q_1,\ldots, Q_m\} \in {\cal P}_{\infty}$ be basic.
For each $i \in \{1,\ldots, m\}$, let $u_i$ and $v_i$ be the endpoints of $Q_i$, let $Q'_i = Q_i \setminus \{u_i\}$ and let $Q''_i = Q_i \setminus \{v_i\}$.
We say that $P$ is \emph{succinct} if for all $i \in \{1,\ldots, m\}$, $P$ and $\{Q_1,\ldots,Q_{i-1},Q'_i,Q_{i+1},\ldots,Q_m\}$ are not twins, and moreover $P$ and $\{Q_1,\ldots,Q_{i-1},Q''_i,Q_{i+1},\ldots,Q_m\}$ are not twins.
\end{definition}

\begin{lemma}\label{lem:unique_succinct_basic}
Let $P \in {\cal P}_{\infty}$ be non-empty basic. Then there exists some succinct and basic $P' \in {\cal P}_{\infty}$ such that $P$ and $P'$ are twins with $E(P')\subseteq E(P)$.
\end{lemma}

\begin{proof}
It is immediate by Definition \ref{def:succ}.
\end{proof}

\begin{lemma}\label{lem:decomposing}
Let $P \in {\cal P}_{\infty}$ be non-empty basic elementary and succinct. Let $\cal T$, $D$, $j$, $D_1, \ldots, D_j$ be as in Definition \ref{dfn:basic_family_paths}. 
Then there exist $P_1,P_2\in {\cal P}_{\infty}$ satisfying the following conditions:
\begin{description}
\item{(1)}
$P_1$ and $P_2$ are non-empty basic elementary and succinct.
\item{(2)}
$P_1$ covers $\bigcup_{i=1}^{j-1} {\cal T}_{D_i}$ and avoids ${\cal T} \setminus (\bigcup_{i=1}^{j-1} {\cal T}_{D_i})$.
\item{(3)}
$P_2$ covers ${\cal T}_{D_j}$ and avoids ${\cal T} \setminus {\cal T}_{D_j}$.
\item{(4)}
$E(P_1) \subseteq E(P)$, $E(P_2) \subseteq E(P)$, and $E(P_1) \cap E(P_2) = \emptyset$.
\end{description}
\end{lemma}

\begin{proof}
We begin by defining auxiliary $Z_1,Z_2\in {\cal P}_{\infty}$.
The desired $P_1$ and $P_2$ will be succinct twins of $Z_1$ and $Z_2$.
First we define $Z_1$.
Initially, we set $Z_1=P$ and we inductively modify $Z_1$ until it covers $C_1=\bigcup_{i=1}^{j-1} {\cal T}_{D_i}$ and avoids $A_1={\cal T} \setminus (\bigcup_{i=1}^{j-1} {\cal T}_{D_i})$ as follows:
If $Z_1$ contains a path $Q$ that does not intersect $C_1$ then we remove $Q$ from $Z_1$. 
If $Z_1$ contains a path $Q$ that intersects both $C_1$ and $A_1$ then we proceed as follows: let $R$ be the collection of paths obtained from $Q$ by deleting all vertices in $A_1\cap Q$; let $R'$ be the collection obtained from $R$ by removing all paths that do not intersect $C_1$; let $R''$ be the collection obtained from $R'$ by replacing each $Q'\in R'$ be the minimal subpath $Q''\subseteq Q'$ with $V(Q'')\cap C_1=V(Q')\cap C_1$.
We repeat the above process until the resulting $Z_1$ covers $C_1=\bigcup_{i=1}^{j-1} {\cal T}_{D_i}$ and avoids $A_1={\cal T} \setminus (\bigcup_{i=1}^{j-1} {\cal T}_{D_i})$.
In a similar fashion we define $Z_2$ that covers $C_2={\cal T}_{D_j}$ and avoids $A_2={\cal T} \setminus {\cal T}_{D_j}$.
It is immediate by construction that $E(Z_1)\subset E(P)$, $E(Z_2)\subset E(P)$ and $E(Z_1)\cap E(Z_2)=\emptyset$.

Next we argue that $Z_1$ is basic.
It is immediate that conditions (1) \& (2) of Definition \ref{dfn:basic_family_paths} are satisfied.
It remains to establish condition (3) of Definition \ref{dfn:basic_family_paths}.
Let $Q\in P_1$.
By construction there exists $Q'\in P$ such that $Q\subseteq Q'$.
Let ${\cal T}[Z_1]$ and ${\cal T}[P]$ be as in Definition \ref{dfn:basic_family_paths}.
Let ${\cal T}_1$ and ${\cal T}_2$ be disjoint subtrees of ${\cal T}[Z_1]$ such that $Q$ covers ${\cal T}_1\cup {\cal T}_2$ and avoids ${\cal T} \setminus ({\cal T}_1\cup {\cal T}_2)$.
Let $\cT_{Q'}$ be the minimal subtree of $\cT[P]$ that contains all the nodes of $\cT[P]$ that are covered by $Q'$.
Let $\cT'=\cT_{Q'}\cup \cT_1 \cup \cT_2$.
By definition we have that $\cT'$ is a subtree of $\cT[P]$.
Therefore there exist disjoint subtrees $\cT_1'$ and $\cT_2'$ of $\cT'$ such that $\cT_1\subseteq \cT_1'$, $\cT_2\subseteq \cT_2'$, and $\cT'=\cT_1'\cup\cT_2'$.
Since $P$ is basic, it follows by condition (3) of Definition \ref{dfn:basic_family_paths} that there exist edge-disjoint subpaths $Q_1',Q_2'$ of $Q'$ such that $Q_1'$ covers $\cT_1'$ and avoids $\cT\setminus \cT_1'$ and $Q_2'$ covers $\cT_2'$ and avoids $\cT\setminus \cT_2'$.
Let $Q_1=Q\cap Q_1'$ and $Q_2=Q\cap Q_2'$.
It now follows that $Q_1$ covers $\cT_1$ and avoids $\cT\setminus \cT_1$ and $Q_2$ covers $\cT_2$ and avoids $\cT\setminus \cT_2$, establishing Condition (3) of Definition \ref{dfn:basic_family_paths}.

It remains to show that $|P_1|$ is minimal subject to condition (3) of Definition \ref{dfn:basic_family_paths}.
Suppose not. Then there are $Q,Q'\in P_1$ that are consecutive in $F$ such that we can replace $Q$ and $Q'$ in $P_1$ by some subpath $Q''$ that contains $Q$ and $Q'$.
If $Q$ and $Q'$ are subpaths of the same path in $P$ then this is a contradiction because by construction there must exist a vertex in $V(Q'')\setminus (V(Q)\cup V(Q'))$ that $P_1$ must avoid.
Therefore there must exist distinct $R,R'\in P$ such that $Q\subseteq R$ and $Q'\subseteq R'$.
However this means that we can replace $R$ and $R'$ in $P$ by some subpath containing both $R$ and $R'$, contradicting the fact that $P$ is basic.
This establishes that $|P_1|$ is minimal subject to condition (3) of Definition \ref{dfn:basic_family_paths}.
We have thus obtained that $Z_1$ is basic.
It is also immediate by construction that since $P$ is elementary, $Z_1$ is also elementary.
By the exact same argument it follows that $Z_2$ is also basic and elementary.

For any $i\in\{1,2\}$ let $P_i$ be a succinct twin of $Z_i$ obtained by Lemma \ref{lem:unique_succinct_basic}.
Since $Z_i$ is basic and elementary, it follows that $P_i$ is also basic and elementary, and thus condition (1) is satisfied.
Since $P_i$ is a twin of $Z_i$, it follows that conditions (2) \& (3) are satisfied.
Since $E(Z_i)\subset E(P)$ and $E(P_i)\subseteq E(Z_i)$, we get $E(P_i)\subset E(P)$.
Finally, since $E(Z_1)\cap E(Z_2)=\emptyset$, we have $E(P_1)\cap E(P_2)=\emptyset$, and thus condition (4) is satisfied, which concludes the proof.
\end{proof}

\begin{definition}[$P$-facial restriction]
Let $P = \{Q_1, \ldots, Q_m\} \in {\cal P}_{\infty}$ be basic. We define the \emph{$P$-facial restriction of $\cal W$} the exact same way as in Definition \ref{defn:facial_restriction}. 
We remark that $P$ is now a family of paths, while in the planar case $P$ is a single path; the definition remains the same by replacing the notion of basic path given in Definition \ref{defn:basic} by the notion of basic family of paths given in Definition \ref{dfn:basic_family_paths}.
\end{definition}

\begin{lemma}[Malni\v{c} and Mohar \cite{malnic1992generating}]\label{lem:mohar_nonseparating}
Let $S$ be an either orientable or non-orientable surface of Euler genus $g$, and let $x\in S$.
Let ${\cal X}$ be a collection of noncontractible curves.
Suppose that at least one of the following holds: 
(i) The curves in ${\cal X}$ are disjoint and (freely) nonhomotopic.
(ii) There exist $x\in S$ such that for every $C,C'\in {\cal X}$, we have $C\cap C' = x$, and the curves in ${\cal X}$ are nonhomotopic (in $\pi_1(S,x)$).
Then, 
\[
|{\cal X}|\leq \left\{\begin{array}{ll}
0 & \text{ if } $S$ \text{ is the 2-sphere}\\
1 & \text{ if } $S$ \text{ is the torus or the projective plane}\\
3(g-1) & \text{ otherwise}
\end{array}
\right.
\]
\end{lemma}

We recall the following result on the genus of the complete bipartite graph \cite{harary1994graph}.

\begin{lemma}\label{lem:K33}
For any $n,m\geq 1$, the Euler genus of $K_{m,n}$ is $\left\lceil(m-2)(n-2)/4\right\rceil$.
\end{lemma}

\begin{lemma}\label{lem:succinct_and_elementary}
Let $P \in {\cal P}_{\infty}$ be elementary and succinct.
Then $|P| \leq 18000 g^3$.
\end{lemma}

\begin{proof}
Let $\cT[P]$ be as in Definition \ref{dfn:basic_family_paths}.
If ${\cal T}[P]$ is trivial, we can find $Q \in {\cal Q}$ such that $Q$ covers ${\cal T}[P]$ and avoids ${\cal T} \setminus {\cal T}[P]$, and thus $|P| = 1$ and we are done. 
Now suppose
that ${\cal T}[P]$ is non-trivial.
Let $m = |P|$ and suppose that $P = \{Q_1, \ldots, Q_m\}$ such that $Q_1, \ldots, Q_m$ are subpaths of $F$ in this order along a traversal of $F$. 
Let $Z_1,\ldots,Z_{m'}$ be a sequence of disjoint subsets of $P$ such that $P=\bigcup_{i=1}^{m'} Z_i$, and each $Z_i$ is maximal subject to the following condition: Let $Z_i'$ be the minimal subpath of $F$ that contains all the paths in $Z_i$ and does not intersect any other paths in $P$;
then $Z_i'$ avoids all non-trivial tress in ${\cal F}\setminus \cT$.
For every $j \in \{1, \ldots, m'-1\}$, let $P_j$ be the subpath between $Z_j'$ and $Z_{j+1}'$ and let ${\cal P}' = \{P_1,\ldots, P_{m-1}\}$. Note that since $P$ is basic, for every two consecutive  $Z_j'$ and $Z_{j+1}'$, there exists a non-trivial tree ${\cal T}_j$ such that ${\cal T}_j$ intersects $P_j$.

We construct a set $R$ of non-trivial trees as follows.
We initially set $R =\{{\cal T}_1\}$.
In each step $i > 1$, if there exists ${\cal T'} \in R$ such that ${\cal T'}$ intersects $P_i$, we continue to the next step. Otherwise, we let ${\cal T}'_i$ be some non-trivial tree that intersects $P_i$ and we add ${\cal T}'_i$ to $R$ and we continue to the next step. We argue that $|R| \leq 20g$. Suppose not.
For each ${\cal T}' \in R$ where ${\cal T}'$ intersects some $P' \in {\cal P}'$ at some $x' \in V(P')$, we construct a path $\gamma_{\cal T'}$ in the surface, with both endpoints on $\psi(F)$, and such that after contracting $\psi(F)$ into single point, the loop resulting from $\gamma_{\cT'}$ is non-contractible, as follows.
First suppose that $\psi(\cT')$ contains some non-contractible loop $\gamma$.
Then let $\zeta$ be a path in $\psi(\cT')$ between $x'$ and some point in $\gamma$.
We set $\gamma_{\cT'}$ to be the path starting at $x'$, traversing $\zeta$, followed by $\gamma$, followed by the reversal of $\zeta$, and terminating at $x'$.
Otherwise, suppose that $\psi(cT')$ does not contain any non-contractible loops.
Since $\cT'$ is non-trivial, it follows that $\psi(\cT')$ contains some path $\xi$ with both endpoints in $\psi(F)$ such that after contracting $\psi(F)$ into a single point, the loop obtained from $\xi$ is non-contractible.
Let $\xi'$ be some path in $\psi(\cT')$ between $x'$ and some point in $\xi$.
By Thomassen's 3-path condition \cite{thomassen1990embeddings} it follows that $\xi\cup \xi'$ contains some path $\xi''$ with both endpoints in $\psi(F)$, such that one of these endpoints is $x'$, and such that after contracting $\psi(F)$ into a single point, the loop resulting from $\xi''$ is non-contractible.
We set $\gamma_{\cT'}$ to be $\xi''$.

Let ${\cal L}$ be the set of all loops obtained from the paths $\gamma_{\cT'}$ as follows.
Pick a point $x$ in the interior of the disk bounded by $\psi(F)$.
Connect $x$ to both endpoints of each $\gamma_{\cT'}$ by paths such that all chosen paths are interior disjoint.
During this process each path $\gamma_{\cT'}$ gives rise to a non-contractible loop in ${\cal L}$ such that any two loops in ${\cal L}$ intersect only at $x$.
Let $L',L'',L'' \in {\cal L}$ be distinct.
We show that $L'$, $L''$ and $L'''$ can not be all homotopic. Suppose not.
Since they are interior-disjoint, by removing them from the surface we obtain three connected components.
Since $\psi(\cT)$ does not intersect any of $L'$, $L''$ and $L'''$, it has to be inside one of the three connected components completely. We may assume w.l.o.g~that ${\cal T}$ is inside the component which is bounded by $L'$ and $L''$. Therefore, there is no path from $T$ to $L'''$ without crossing $L'\cup L''$, which is a contradiction.


Let ${\cal L}' \subseteq {\cal L}$ be a maximal subset such that for all $L', L'' \in {\cal L}'$ we have that $L'$ and $L''$ are non-homotopic. Since $|{\cal L}| > 20g$ and for every three loops in ${\cal L'}$ we know that at most two of them are homotopic, we have that $|{\cal L}'| > 10g$, which contradicts Lemma \ref{lem:mohar_nonseparating}. Therefore, we have that $|R| \leq 20g$ and thus there exists ${\cal T}_0 \in R$ such that ${\cal T}_0$ intersects at least $10g = 200g^2/20g$ elements of ${\cal P}'$. Let $x_1$ be a point inside the face. Let $x_2$ be a point in the root of ${\cal T}_0$ and let $x_3$ be a point in $\psi(D)$. There exists $P'_1, \ldots, P'_{10g} \in {\cal P}'$ such that for every $i \in \{1,\ldots,10g\}$, ${\cal T}_0$ intersects $P'_i$. For every $i \in \{1,\ldots,10g\}$, let $y_i$ be a point on $P'_i$. By the construction, for every $y_i$ we can find non-crossing paths to $x_1$, $x_2$ and $x_3$. Therefore, we get an embedding of $K_{3,10g}$ in a surface of genus $g$ which contradicts Lemma \ref{lem:K33}.
This establishes that $m'\leq 200 g^2$.


Let $i\in \{1,\ldots,m'\}$.
We next we bound $|Z_i|$.
Suppose $Z_i=\{Q_a,Q_{a+1},\ldots,Q_{a+\ell}\}$.
Let $x$ be an arbitrary point in $\psi(D)$.
For each $j\in \{0,\ldots,\lfloor (\ell-1)/2\rfloor\}$ pick an arbitrary point $x_j \in \psi(Q_{a+2j})\cap \psi(\cT)$; we define a path in the surface between $x_j$ and $x$ as follows: we start from the vertex $D'$ of $\cT$ containing $x_j$; if $D'$ is a closed walk then we traverse $D'$ clockwise until we reach the point that connects $D'$ to its parent; otherwise we traverse the unique path between $x_j$ and the point that connects $D'$ to its parent. We continue in this fashion until we reach $D$, and we finally traverse $D$ clockwise until we reach $x$.
This completes the definition of the path $\gamma_j$.
It is immediate that for all $j\neq j'$, the paths $\gamma_j$ and $\gamma_{j'}$ are non-crossing.

Let $S'$ be the surface obtained by contracting $\psi(F)$ into a single point $y$, and identifying $x$ with $y$ (note that $S'$ has Euler genus at most $g+2$).
For each $j\in \{0,\ldots,\lfloor (\ell-1)/2\rfloor\}$ let $\gamma_j'$ be the loop in $S'$ obtained from $\gamma_j$ after the above contraction and identification.
We argue that there can be at most 4 loops $\gamma_j'$ that are pairwise homotopic.
Suppose for the sake of contradiction that there are at least 5 such loops.
It follows that there must exist $t\in \{0,\ldots,\lfloor (\ell-1)/2\rfloor\}$ such that $\gamma_t'$, $\gamma_{t+1}'$, and $\gamma_{t+2}'$ are all pairwise homotopic.
We are going to obtain a contradiction by arguing that the paths $Q_{a+2t}$ and $Q_{a+2t+1}$ violate the fact that $P$ is basic.
To that end, let $Q'$ be the minimal subpath of $F$ that contains both $Q_{a+2t}$ and $Q_{a+2t+1}$ and does not intersect any other paths in $P$.
We will show that $(P \setminus \{Q_{a+2t},Q_{a+2t+1}\}) \cup \{Q'\}$ is basic, thus violating the fact that $|P|$ is minimal subject to condition (3) of definition \ref{dfn:basic_family_paths}.
Let $\cT_1$, $\cT_2$ be disjoint subtrees of $\cT[P]$ such that $Q'$ covers $\cT_1\cup \cT_2$ and avoids $\cT\setminus (\cT_1\cup \cT_2)$.
We need to show that there exist edge-disjoint subpaths $Q_1'$ and $Q_2'$ of $Q'$ such that $Q'=Q_1'\cup Q_2'$, $Q_1'$ covers $\cT_1$ and avoids $\cT\setminus \cT_1$, and $Q_2'$ covers $\cT_2$ and avoids $\cT\setminus \cT_2$.
Let $\lambda$ be the subpath of $\psi(F)$ in $S$ between $x_t$ and $x_{t+1}$ that contains $x_{t+1}$.
Since $\gamma_t'$, $\gamma_{t+1}'$, and $\gamma_{t+2}'$ are homotopic, it follows that $\gamma_t \cup \lambda \cup \gamma_{t+2}$ bounds a disk $\Psi$ in $S$.
Each $x_s$ is contained in the image of a unique vertex $D_s'$ of $\cT[P]$.
Let $B$ be the path in $\cT[P]$ between $D_{t}'$ and $D_{t+2}'$.
For each $r\in \{1,2\}$, ${\cal T}_r$ intersects $B$ into some possibly empty subpath $B_r$.
It follows that there exist disjoint disks $\Psi_1,\Psi_2\subset \Psi$ such that for each $r\in \{1,2\}$, $\psi(\cT_r)\cap \Psi\subset \Psi_r$.
Therefore there exist edge-disjoint subpaths $Q_1'$ and $Q_2'$ of $Q'$
with 
$Q'=Q_1'\cup Q_2'$
such that
$\psi(Q')\cap \Psi_1 \subseteq \psi(Q_1')$ and 
$\psi(Q')\cap \Psi_2 \subseteq \psi(Q_2')$.
It follows that 
$Q_1'$ covers $\cT_1$ and avoids $\cT\setminus \cT_1$,
and
$Q_2'$ covers $\cT_2$ and avoids $\cT\setminus \cT_2$.
This contradicts the fact that $P$ is basic, and concludes the proof that there are at most four loop $\gamma_t'$ that are pairwise homotopic.

\begin{center}
\scalebox{0.7}{\includegraphics{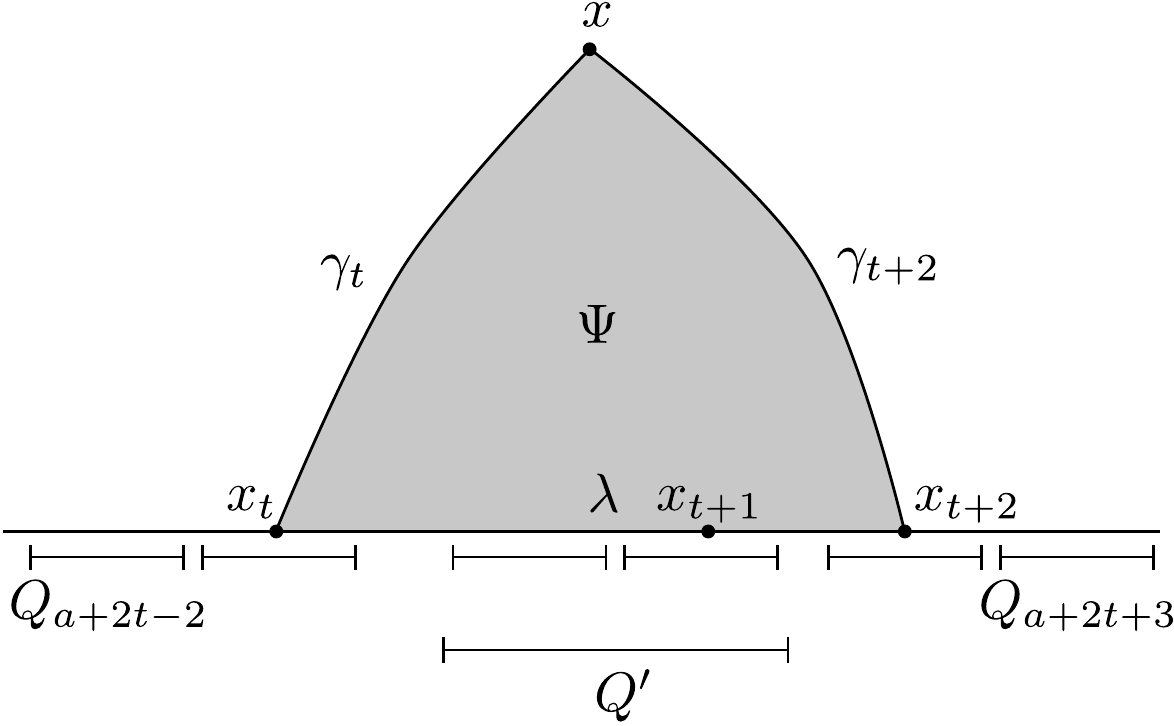}}
\end{center}

Pick $I\subseteq \{0,\ldots,\lfloor (\ell-1)/2\rfloor\}$, with $|I|\geq \lfloor (\ell-1)/10\rfloor$ such that for all $t\neq t'\in I$, we have that $\gamma'_t$ and $\gamma'_{t'}$ are non-homotopic.
Since the paths $\gamma_t$ are non-crossing, we may assume w.l.o.g.~that the paths $\gamma'_t$ are interior disjoint after an infinitesimal perturbation.
By Lemma \ref{lem:mohar_nonseparating} on $S'$ it follows that $|I|\leq 3(g+1)$.
Since $|I|\geq \ell/15$, we get $\ell \leq 45 (g+1)$.
Therefore $|Z_i|\leq 45 (g+1)$.

We conclude that $m\leq \sum_{i=1}^{m'} |Z_i| \leq m' \cdot \max_{i\in \{1,\ldots,m'\}} |Z_i| \leq 9000 g^2 (g+1) \leq 18000 g^3$. 
\end{proof}

\begin{lemma}\label{lem:partial_solution_twins}
Let $P_1 = \{Q_1, \ldots, Q_m\} \in {\cal P}$ be succinct. Let $P_2 = \{Q'_1, \ldots, Q'_l\} \in {\cal P}$ such that $P_1$ and $P_2$ are twins, $V(P_1) \subseteq V(P_2)$ and $E(P_1) \subseteq E(P_2)$.
For every $i \in \{1,\ldots, m\}$, let $u_i$ and $v_i$ be the endpoints of $Q_i$.
Let ${\cal B}_1 = \bigcup_{i=1}^{m} (B_{u_i} \cup B_{v_i})$.
Let ${\cal W}_{P_1}$ be the $P_1$-facial restriction of $\cal W$.
Let $\Gamma_1=\bigcup_{\vW\in {\cal W}_{P_1}} \vW$.
Let $C_1$ be the partition of ${\cal B}_1$ that corresponds to the weakly-connected components of $\Gamma_1$.
For any $x\in {\cal B}_1$ let 
$f_1^{\myin}(x) = \indeg_{{\cal W}_{P_1}}(x)$
and 
$f_1^{\myout}(x) = \outdeg_{{\cal W}_{P_1}}(x)$.
We similarly define $C_2,f_2^\myin, f_2^\myout$ and ${\cal W}_{P_2}$ for $P_2$.
Suppose that there exists some $a_1\in {\cal A} \cup ({\cal A} \times {\cal A}) \cup \nil$ and $l_1,r_1,p_1 \in (V(\vG) \cup \nil)$ such that the dynamic programming table contains some partial solution ${\cal S}_1$ at location $(P_1,(C_1,f_1^\myin, f_1^\myout, a_1,l_1,r_1,p_1))$, with
$\cost_{\vG}({\cal S}_1) \leq \cost_{\vG}({\cal W}_{P_1})$.
Then there exists some $a_2\in {\cal A} \cup ({\cal A} \times {\cal A}) \cup \nil$ and $l_2,r_2,p_2 \in (V(\vG) \cup \nil)$ such that the dynamic programming table contains some partial solution ${\cal S}_2$ at location $(P_2,(C_2,f_2^\myin, f_2^\myout, a_2,l_2,r_2,p_2))$, with
$\cost_{\vG}({\cal S}_2) \leq \cost_{\vG}({\cal W}_{P_2})$.
\end{lemma}

\begin{proof}
Let $E_1 = E(P_2) \setminus E(P_1)$.
For every $e \in E_1$, let $Q_e \in {\cal Q}$ be the path containing a single edge $e$, and let $P_e = \{Q_e\}$. Note that $P_e$ is an empty basic family of paths and $|E(P_e)|=1$. Therefore for every $e \in E_1$, the initialization step of the dynamic programming, finds a partial solution ${\cal S}_e$ for $P_e$. By merging ${\cal S}_1$ with all these partial solutions sequentially in an arbitrary order, we get a partial solution ${\cal S}_2$ for $P_2$, as desired.
\end{proof}

\begin{lemma}\label{lem:partial_solution_for_trivial}
Let $P = \{Q_1 \ldots, Q_m\} \in {\cal P}$ be basic and trivial (w.r.t.~${\cal W}$).
For every $i \in \{1,\ldots, m\}$, let $u_i$ and $v_i$ be the endpoints of $Q_i$.
Let ${\cal B} = \bigcup_{i=1}^{m} (B_{u_i} \cup B_{v_i})$.
Let ${\cal W}_{P}$ be the $P$-facial restriction of $\cal W$.
Let $\Gamma=\bigcup_{\vW\in {\cal W}_P} \vW$.
Let $C$ be the partition of $\cal B$ that corresponds to the weakly-connected components of $\Gamma$.
For any $x\in {\cal B}$ let 
$f^{\myin}(x) = \indeg_{{\cal W}_P}(x)$
and 
$f^{\myout}(x) = \outdeg_{{\cal W}_P}(x)$.
Then there exists some $a\in {\cal A} \cup ({\cal A} \times {\cal A}) \cup \nil$ and $l,r,p \in (V(\vG) \cup \nil)$ such that the dynamic programming table contains some partial solution ${\cal S}$ at location $(P,(C,f^\myin, f^\myout, a,l,r,p))$, with
$\cost_{\vG}({\cal S}) \leq \cost_{\vG}({\cal W}_P)$.
\end{lemma}

\begin{proof}
Since $P$ is trivial, we have that $P \in {\cal P}_1$, and thus the exact same argument as in Lemma \ref{lem:vortex_DP_induction} applies here.
\end{proof}

\begin{lemma}\label{lem:partial_solution_for_succinct_elementary}
Let $P = \{Q_1 \ldots, Q_m\} \in {\cal P}$ be succinct and elementary (w.r.t.~${\cal W}$).
For every $i \in \{1,\ldots, m\}$, let $u_i$ and $v_i$ be the endpoints of $Q_i$.
Let ${\cal B} = \bigcup_{i=1}^{m} (B_{u_i} \cup B_{v_i})$.
Let ${\cal W}_{P}$ be the $P$-facial restriction of $\cal W$.
Let $\Gamma=\bigcup_{\vW\in {\cal W}_P} \vW$.
Let $C$ be the partition of $\cal B$ that corresponds to the weakly-connected components of $\Gamma$.
For any $x\in {\cal B}$ let 
$f^{\myin}(x) = \indeg_{{\cal W}_P}(x)$
and 
$f^{\myout}(x) = \outdeg_{{\cal W}_P}(x)$.
Then there exists some $a\in {\cal A} \cup ({\cal A} \times {\cal A}) \cup \nil$ and $l,r,p \in (V(\vG) \cup \nil)$ such that the dynamic programming table contains some partial solution ${\cal S}$ at location $(P,(C,f^\myin, f^\myout, a,l,r,p))$, with
$\cost_{\vG}({\cal S}) \leq \cost_{\vG}({\cal W}_P)$.
\end{lemma}

\begin{proof}
Let $\cal T$, $D$, $k$, $D_1, \ldots, D_k$ and $j$ be as in Definition \ref{dfn:basic_family_paths}. We prove the assertion by induction on $\cal T$. For the base case, where $D$ is a leaf of $\cal T$, the same argument as in Lemma \ref{lem:vortex_DP_induction} applies here. Suppose that $D$ is non-leaf. In this case, we prove the assertion by another induction on $j$. For the base case, where $j=1$, the same argument as in Lemma \ref{lem:vortex_DP_induction} applies here. Now suppose that we have proved the assertion for all $j' < j$. Let $P_1 \in {\cal P}_{\infty}$ such that $V(P_1) \subseteq V(P)$, $E(P_1) \subseteq E(P)$, $P_1$ covers $\bigcup_{i=1}^{j-1} {\cal T}_{D_i}$ and avoids ${\cal T} \setminus (\bigcup_{i=1}^{j-1} {\cal T}_{D_i})$. Let $P_2 \in {\cal P}_{\infty}$ such that $E(P_2) = E(P) \setminus E(P_1)$. By the construction, $P_1$ and $P_2$ are elementary. Therefore, by Lemma \ref{lem:unique_succinct_basic}, there exists elementary and succinct $P'_1\in {\cal P}_{\infty}$ and $P'_2 \in {\cal P}_{\infty}$ for $P_1$ and $P_2$ respectively. Moreover, by Lemma \ref{lem:decomposing}, we have that $E(P'_1) \subseteq E(P)$, $E(P'_2) \subseteq E(P)$ and $E(P'_1) \cap E(P'_2) = \emptyset$
. By Lemma \ref{lem:succinct_and_elementary} we have that $P, P'_1, P'_2 \in {\cal P}_{18000g^3}$.
We next define some $P''_1\in \cP_{\infty}$.
Let $\Gamma$ be the graph obtained as the union of all paths in $P$, that is $\Gamma=\bigcup_{Q\in P} Q$.
Similarly, let $\Gamma'=\bigcup_{Q\in P'_2} Q$,
let $\Gamma''=\Gamma\setminus E(\Gamma')$,
and let $\Gamma'''$ be the graph obtained from $\Gamma''$ by deleting all isolated vertices.
We define $P''_1$ to be the set of connected components in $\Gamma'''$.
Note that $\Gamma'''$ has at most $36000 g^3$ connected components, and each such component is a path.
Thus $P''_1 \in {\cal P}_{36000g^3}$, and $E(P'_1) \subseteq E(P''_1)$. By the induction hypothesis, there exists a partial solution ${\cal S}'_1$ for $P'_1$, and thus by Lemma \ref{lem:partial_solution_twins}, there exists a partial solution ${\cal S}''_1$ for $P''_1$. Also, by the induction hypothesis, there exists a partial solution ${\cal S}'_2$ for $P'_2$. By merging ${\cal S}''_1$ and ${\cal S}'_2$, we get a partial solution $\cal S$ for $P$, as desired.
\end{proof}

We define the following three types of non-trivial trees in $\cal F$: 
\begin{description}
\item{(1)}
We say that a non-trivial tree ${\cal T}$ is of the \emph{first type}, if there exists a non-leaf $D \in V({\cal T})$, such that $D$ is a closed walk, with $V(D)\cap V(F)=\emptyset$, such that $\psi(D)$ is non-contractible.

\item{(2)}
We say that a non-trivial tree $\cal T$ is of the \emph{second type}, if it is not of the first type and there exists a leaf $D \in V({\cal T})$ such that $\psi(D)\cup \psi(F)$ is non-contractible.

\item{(3)}
We say that a non-trivial tree $\cal T$ is of the \emph{third type}, if it is not of the first type nor of the second type.
\end{description}

We say that a non-trivial tree $\cal T$ is \emph{good}, if at least one of the following conditions holds:
\begin{description}
\item{(1)}
$\cal T$ is of the second type, and for every $D_1,D_2 \in V({\cal T})$ where $\psi(D_1)\cup \psi(F)$ and $\psi(D_1)\cup \psi(F)$ are non-contractible, we have that the loops obtained from $\psi(D_1)$ and $\psi(D_1)$ by contracting $\psi(F)$ into a single poit $x$ are homotopic in $\pi_1(S,x)$.
Let $D \in V({\cal T})$ such that $\psi(D)\cup \psi(F)$ is non-contractible. We let $\beta({\cal T})$ to be the homotopy class of the loop $\psi(D)$ in the surface obtained after contracting $\psi(F)$ into a single point.

\item{(2)}
$\cal T$ is of the third type and the following holds.
Let $X=\psi(F) \cup \bigcup_{D\in V({\cal T})}\psi(D)$ and let $X'$ be the image of $X$ after contracting $\psi(F)$ into a single point $x$.
Then and all non-contractible loops in $X'$ are homotopic in $\pi_1(S,x)$. We let $\beta({\cal T})$ be the homotopy class in $\pi_1(S,x)$ of all non-contractible loops in $X'$; note that we may always take a non-contractible loop in $X'$ that contains the basepoint $x$ since $X$ is connected.
\end{description}
Otherwise, we say that $\cal T$ is a \emph{bad} tree.

Let ${\cal T}_0$ and ${\cal T}_1$ be non-trivial good trees in $\cal F$. We say that ${\cal T}_0$ is a \emph{friend} of ${\cal T}_1$ if $\beta({\cal T}_0) = \beta({\cal T}_1)$ (see Figure \ref{fig:friend} for an example).

\begin{figure}
\begin{center}
\scalebox{0.8}{\includegraphics{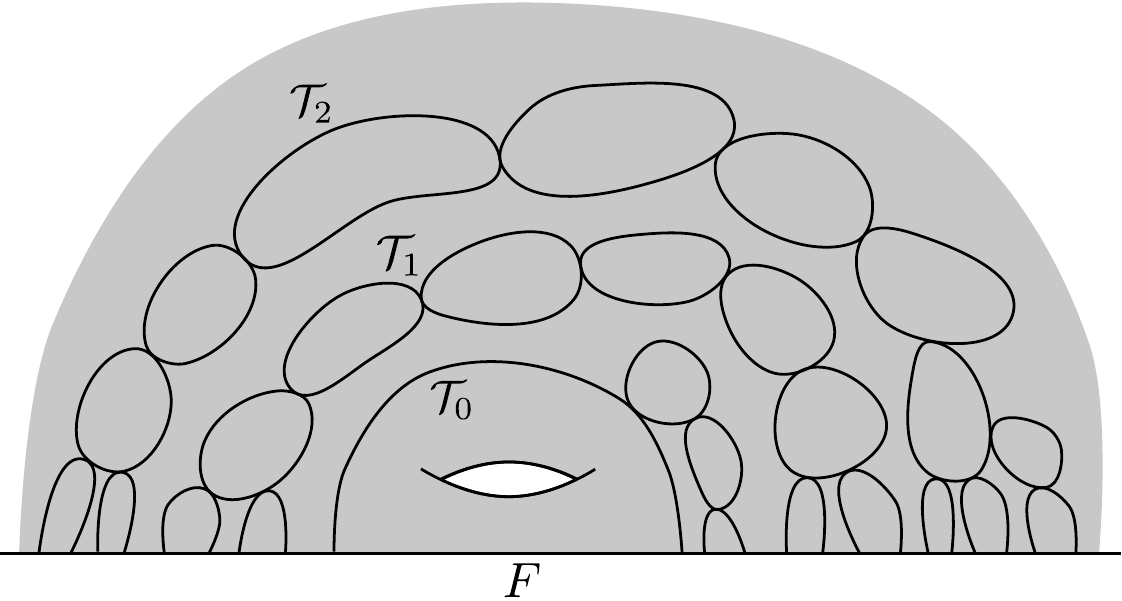}}
\caption{Example of good non-trivial trees ${\cal T}_0$, ${\cal T}_1$, and ${\cal T}_2$ that are pairwise friends; note that ${\cal T}_0$ is of the second type while ${\cal T}_1$ and ${\cal T}_2$ are of the third type.}\label{fig:friend}
\end{center}
\end{figure}

\begin{definition}[Friendly]
Let $P \in {\cal P}_{\infty}$.
We say that $P$ is \emph{friendly} if the following holds.
\begin{description}
\item{(1)}
$P$ avoids all non-trivial bad trees.
\item{(3)}
For any two non-trivial good trees ${\cal T}_0$ and ${\cal T}_1$ in $\cal F$ that $P$ covers, we have that $\beta({\cal T}_0) = \beta({\cal T}_1)$.
\item{(3)}
If $P$ covers a non-trivial good tree ${\cal T}_0$, then $P$ covers all non-trivial good trees ${\cal T}_1$ with $\beta({\cal T}_0) = \beta({\cal T}_1)$.
\end{description}
\end{definition}

\begin{lemma}\label{lem:bad_trees}
There exists at most $12g$ bad trees in $\cal F$.
\end{lemma}

\begin{proof}
We partition all bad trees in $\cal F$ into three sets:
\begin{description}
\item{(1)}
$F_1 = \{{\cal T} \in {\cal F}: {\cal T} \text{ is a bad tree of the first type}\}$.
\item{(2)}
$F_2 = \{{\cal T} \in {\cal F}: {\cal T} \text{ is a bad tree of the second type}\}$.
\item{(3)}
$F_3 = \{{\cal T} \in {\cal F}: {\cal T} \text{ is a bad tree of the third type}\}$.
\end{description}

We first bound $|F_1|$. Let ${\cal T} \in F_1$. We let $\beta({\cal T})$ to be the homotopy class of some non-leaf $D \in V(T)$, such that $D$ is a closed walk, with $V(D)\cap V(F)=\emptyset$, and $\psi(D)$ is non-contractible. Note that by the definition of trees of the first type, such a non-leaf $D \in V({\cal T})$ exists. Let $\cT_1, \cT_2, \cT_3 \in F_1$ be distinct.
We show that $\beta(\cT_1) = \beta(\cT_2) = \beta(\cT_3)$ cannot happen. Suppose not, and we have that $\beta(\cT_1) = \beta(\cT_2) = \beta(\cT_3)$. For any $i\in\{1,2,3\}$, let $D_i \in V(\cT_i)$ such that $\beta(\cT_i)$ is the homotopy class of $D_i$. By removing $\psi(D_1)$, $\psi(D_2)$ and $\psi(D_3)$ from the surface we obtain three connected components. We may assume w.l.o.g~that $\psi(D_2)$ is inside an annulus bounded by $\psi(D_1)$ and $\psi(D_3)$. Therefore, there is no path in the surface from $\psi(D_2)$ to $\psi(F)$, that does not intersect $\psi(D_1\cup D_3)$, which is a contradiction. Therefore by Lemma \ref{lem:mohar_nonseparating} we have that $|F_1| \leq 6g$.

Next we bound $|F_2|$ and $|F_3|$. Let ${\cal T} \in F_2$. By the construction, there exist $D_1,D_2 \in V({\cal T})$ where $\psi(D_1)\cup \psi(F)$ and $\psi(D_2)\cup \psi(F)$ are non-contractible, and the loops obtained from $\psi(D_1)$ and $\psi(D_2)$ after contracting $\psi(F)$ into a single point $y$ are non-homotopic in $\pi_1(S,y)$. Let $x$ be a point in the interior of the disk bounded by $\psi(F)$. For every $\cT \in F_2$, similarly to Lemma \ref{lem:succinct_and_elementary}, we construct two non-homotopic loops $\gamma_{\cT'}$ and $\gamma'_{\cT'}$ in the surface, corresponding to $D_1$ and $D_2$ respectively, such that they only intersect at $x$. Let $\cal L$ be the set of all these loops. Let $L', L'', L''' \in {\cal L}$ be distinct. Similarly to Lemma \ref{lem:succinct_and_elementary}, we show that $L'$, $L''$ and $L'''$ can not be all homotopic. Suppose not. We may assume w.l.o.g.~that $L''$ is inside the disk bounded by $L'$ and $L'''$. Let $\cT'' \in F_2$ be the tree corresponding to $L''$.
Now note that $L''$ is inside the disk bounded by $L'$ and $L'''$, and thus all non-contractible loops in $\psi(F) \cup \bigcup_{D\in V({\cal T})}\psi(D)$ are in the same homotopy class as $L''$. This contradicts the fact that $\psi(D_1)\cup \psi(F)$ and $\psi(D_2)\cup \psi(F)$ are non-homotopic. Therefore, by Lemma \ref{lem:mohar_nonseparating} we have that $|F_2| \leq 3g$. Using a similar argument, we can show that $|F_3| \leq 3g$.

Therefore, we have that $|F_1| + |F_2|+ |F_3| \leq 12g$, completing the proof.
\end{proof}

\begin{lemma}\label{lem:vortex_DP_friendy_succinct_basic}
Let $P = \{Q_1 \ldots, Q_m\} \in {\cal P}_{\infty}$ be friendly, succinct and basic (w.r.t.~${\cal W}$).
For every $i \in \{1,\ldots, m\}$, let $u_i$ and $v_i$ be the endpoints of $Q_i$.
Let ${\cal B} = \bigcup_{i=1}^{m} (B_{u_i} \cup B_{v_i})$.
Let ${\cal W}_{P}$ be the $P$-facial restriction of $\cal W$.
Let $\Gamma=\bigcup_{\vW\in {\cal W}_P} \vW$.
Let $C$ be the partition of $\cal B$ that corresponds to the weakly-connected components of $\Gamma$.
For any $x\in {\cal B}$ let 
$f^{\myin}(x) = \indeg_{{\cal W}_P}(x)$
and 
$f^{\myout}(x) = \outdeg_{{\cal W}_P}(x)$.
Then there exists some $a\in {\cal A} \cup ({\cal A} \times {\cal A}) \cup \nil$ and $l,r,p \in (V(\vG) \cup \nil)$ such that the dynamic programming table contains some partial solution ${\cal S}$ at location $(P,(C,f^\myin, f^\myout, a,l,r,p))$, with
$\cost_{\vG}({\cal S}) \leq \cost_{\vG}({\cal W}_P)$.
\end{lemma}

\begin{proof}
If $P$ does not cover any non-trivial trees in $\cal F$, then $P$ is trivial and thus we have that $P \in {\cal P}_1$ and by Lemma \ref{lem:partial_solution_for_trivial} there exists a partial solution $\cal S$ for $P$. Otherwise suppose that $P$ covers a non-trivial good tree ${\cal T}$ in $\cal F$. By the definition of friendly, for every non-trivial good tree ${\cal T}'$ in $\cal F$ with $\beta({\cal T}) = \beta({\cal T}')$, we have that $P$ covers ${\cal T}'$. Let $A = \{{\cal T}' \in {\cal F} : \beta({\cal T}) = \beta({\cal T}')\}$. Suppose that $A=\{{\cal T}_0, \ldots, {\cal T}_m\}$ such that they intersect $F$ in this order along a traversal of $F$. For any $i \in \{0,\ldots,m\}$, let $P_i \in {\cal P}_{\infty}$ be elementary and succinct such that $P_i$ covers ${\cal T}_i$ and avoids any other non-trivial trees in $\cal F$. By Lemma \ref{lem:succinct_and_elementary} we have that $P_i \in {\cal P}_{18000g^3}$. Moreover, by Lemma \ref{lem:partial_solution_for_succinct_elementary} we have that there exists a partial solution ${\cal S}_i$ for $P_i$.
For any $i\in\{0,\ldots,m\}$, let $P'_i \in {\cal P}_{\infty}$ be succinct and basic, such that $P'_i$ covers $\bigcup_{j=0}^{i} {\cal T}_j$ and avoids any other non-trivial trees in $\cal F$. By the construction, each $P'_i$ can be obtained by merging $P'_{i-1}$ with $P_i$ and some trivial and empty basic family of paths, and thus we have that $P'_{i} \in {\cal P}_{36000g^3}$. Therefore, by induction and Lemma \ref{lem:partial_solution_for_trivial}, there exists a partial solution ${\cal S}'_i$ for each $P'_i$, and thus there exists a partial for solution ${\cal S}'_m$ for $P'_m$. Note that $P'_m = P$. Therefore, we get a partial solution $\cal S$ for $P$, as desired.
\end{proof}

\begin{lemma}\label{lem:vortex_DP_solution_for_F}
Let $P = \{F\} $.
Let $v^\circ\in V(F)$.
Let ${\cal W}_{P}$ be the $P$-facial restriction of $\cal W$.
Let $\Gamma=\bigcup_{\vW\in {\cal W}_P} \vW$.
Let $C$ be the partition of $B_{v^\circ}$ that corresponds to the weakly-connected components of $\Gamma$.
For any $x\in B_{v^\circ}$ let 
$f^{\myin}(x) = \indeg_{{\cal W}_P}(x)$
and 
$f^{\myout}(x) = \outdeg_{{\cal W}_P}(x)$.
Then there exists some $a\in {\cal A} \cup ({\cal A} \times {\cal A}) \cup \nil$ and $l,r,p \in (V(\vG) \cup \nil)$ such that the dynamic programming table contains some partial solution ${\cal S}$ at location $(P,(C,f^\myin, f^\myout, a,l,r,p))$, with
$\cost_{\vG}({\cal S}) \leq \cost_{\vG}({\cal W}_P)$.
\end{lemma}

\begin{proof}
We first partition $F$ to the sets of basic families of paths as follows:
\begin{description}
\item{(1)}
$F_1=\{P' \in {\cP}_{\infty} : P' \text{ is elementary and succinct}, P' \text{ covers some bad tree } {\cal T}\}$.

\item{(2)}
$F_2= \{P' \in {\cP}_{\infty} : P' \text{ is basic, succinct and friendly}\}$.

\item{(3)}
$F_3=\{P' \in {\cP}_{\infty} : P' \text{ is basic, succinct and trivial}\}$.

\item{(4)}
$F_4=\{P' \in {\cP}_{\infty} : P' \text{ is basic, succinct and empty}\}$.
\end{description}

For every bad tree $\cal T$ in $\cal F$, let $P_1 \in F_1$ such that $P_1$ covers $\cal T$ and avoids any other non-trivial tree in $\cal F$.
Since $P_1$ is elementary and succinct, by Lemma \ref{lem:succinct_and_elementary} we have that $|P_1| \leq 18000g^3$. Moreover by Lemma \ref{lem:partial_solution_for_succinct_elementary} there exists a partial solution ${\cal S}_1$ for $P_1$. Also by Lemma \ref{lem:bad_trees}, there exist at most $12g$ bad trees in $\cal F$, and thus $|F_1| \leq 12g$.
Therefore there exists $P'_1 \in {\cal P}_{216000g^4}$, such that for any ${\cal T}' \in {\cal F}$, $P'_1$ covers (resp.~avoids) ${\cal T}'$ if and only if there exists $P_1 \in F_1$ such that $P_1$ covers (resp.~avoids) ${\cal T}'$. Moreover the dynamic programming table contains a partial solution ${\cal S}'_1$ for $P'_1$.
Now we will show that $|F_2| \leq 3g$.
For every $P_2 \in F_2$, let $\cT_2$ be some non-trivial good tree that $P_2$ covers. Let $x$ be a point in the interior of the disk bounded by $\psi(F)$. Similar to Lemma \ref{lem:succinct_and_elementary} we construct a non-contractible loop $\gamma_{P_2}$ that contains $x$, corresponding to $\cT_2$. Let $L$ be the set of all these interior-disjoint loops. For every $P',P'' \in F_2$, since $P'$ and $P''$ are succinct and friendly, $\gamma_{P'}$ and $\gamma_{P''}$ are not homotopic. Therefore by Lemma \ref{lem:mohar_nonseparating} we have that $|L| \leq 3g$, and thus $|F_2| \leq 3g$. Furthermore for every $P_2 \in F_2$, by Lemma \ref{lem:vortex_DP_friendy_succinct_basic}, we have that $P_2 \in \cP_{36000g^3}$, and there exists a partial solution ${\cal S}_2$ for $P_2$. Therefore there exists $P'_2 \in {\cal P}_{108000g^4}$, such that for any ${\cal T}' \in {\cal F}$, $P'_2$ covers (resp.~avoids) ${\cal T}'$ if and only if there exists $P_2 \in F_2$ such that $P_2$ covers (resp.~avoids) ${\cal T}'$. Moreover the dynamic programming table contains a partial solution ${\cal S}'_2$ for $P'_2$.
Let $P'_{1,2} \in \cP_{324000g^4}$ such that for any ${\cal T}' \in {\cal F}$, $P'_{1,2}$ covers (resp.~avoids) ${\cal T}'$ if and only if there exists $P' \in F_1 \cup F_2$ such that $P'$ covers (resp.~avoids) ${\cal T}'$. By merging ${\cal S}'_1$ and ${\cal S}'_2$ we get a partial solution ${\cal S}'_{1,2}$ for $P'_{1,2}$.

Let 
\[
F_5 = \{\{Q\} \in F_3 \cup F_4 : Q \text{ is maximal subject to } E(Q) \cap E(P'_{1,2}) = \emptyset \}.
\]
For every $P_3 \in F_3$, by Lemma \ref{lem:partial_solution_for_trivial} we have that $P_3 \in {\cal P}_1$ and there exists a partial solution ${\cal S}_3$ for $P_3$.
Finally, for every $P_4 \in F_4$, we have that $P_4 \in {\cal P}_1$ and also the dynamic programming table contains a partial solution ${\cal S}_4$ for $P_4$. Therefore for every $\{Q\} \in F_5$, the dynamic programming table contains a partial solution for $\{Q\}$. By merging these partial solutions with $P'_{1,2}$ in an arbitrary order, we get a partial solution ${\cal S}$, as desired. We remark that after merging the current partial solution with the partial solution for some $\{Q\}\in F_5$, we obtain a new partial solution for some $P\in P_{\infty}$ with fewer paths. Therefore for all intermediate $P\in P_{\infty}$ we have $P\in P_{324000 g^4}$, concluding the proof.
\end{proof}

\begin{theorem}\label{thm:vortex_genus}
Let $\vG$ be an $n$-vertex $(0, g, 1, p)$-nearly embeddable graph and let $\vH$ be the vortex of $\vG$.
Then there exists an algorithm which computes a walk $\vW$ visiting all vertices in $V(\vH)$ of total length $\OPT_{\vG}(V(\vH))$ in time $n^{O(pg^4)}$.
\end{theorem}

\begin{proof}
Let $P=\{F\}$. By Lemma \ref{lem:vortex_DP_solution_for_F}, there exists a partial solution for $P$. Now the exact same argument as in Theorem \ref{thm:tour_visiting_a_single_vortex} applies here.
\end{proof}

\section{The algorithm for traversing a vortex in a nearly-embeddable graph}
\label{sec:vortex_nearly}

In this Section we obtain an exact algorithm for computing a closed walk that traverses all the vertices in the single vortex of a nearly-embeddable graph.

\begin{proof}[Proof of Theorem \ref{thm:vortex_nearly}]
The algorithm is as follows.
Let $A\subset V(\vG)$ be the set of apices and let $\vF$ be the face on which the vortex is attached.
For each $A'\subseteq A$ we construct a $(0, g, 1, p+a)$-nearly embeddable graph $\vG_{A'}$ as follows:
Initially we set $\vG_{A'}=\vG \setminus A$.
We then add $A'$ to $V(\vG_{A'})$ and for every $u\in A'$ and every $v\in V(\vF)$ we add edges $(u,v)$ and $(v,u)$ of length $d_{\vG}(u,v)$ and $d_{\vG}(v,u)$ respectively; let $E_{A'}$ be the set of all these new edges.
We consider $A'$ as being part of the vortex and 
we modify the path decomposition of the vortex by adding $A'$ to each one of its bubbles; it is immediate that the result is a path decomposition of width at most $a+p$.
We then run the algorithm of Theorem \ref{thm:vortex_genus} on $\vG_{A'}$ and find an optimal closed walk visiting all vertices in $V(\vH) \cup A'$.
Let $\vW_{A'}$ be the resulting walk in $\vG_{A'}$.
We obtain a walk $\vW_{A'}'$ in $\vG$ visiting all vertices in $V(\vH)$ with $\cost_{\vG}(\vW'_{A'}) = \cost_{\vG_{A'}}(\vW_{A'})$ by replacing every edge in $(u,v)\in E_{A'}$ by a shortest path from $u$ to $v$ in $\vG$.
After considering all subsets $A'\subseteq A$, we output the walk $\vW'_{A'}$ of minimum cost that we find.
This completes the description of the algorithm.

It is immediate that the running time is $O(2^a n^{(a+p)g^4}) = O(n^{(a+p)g^4})$.
It remains to show that the algorithm computes an optimum walk.
Let $\vR$ be the walk computed by the algorithm.
Let $\vW_{\OPT}$ be a walk of minimum cost in $\vG$ visiting all vertices in $V(\vH)$.
Let $A^*$ be the set of apices visited by $\vW_{\OPT}$, that is $A^* = A\cap V(\vW_{\OPT})$.
Let $\vG'$ be the genus-$g$ piece of $\vG$.
We can construct a walk $\vZ$ in $\vG_{A^*}$ that visits all the vertices in $V(\vH)\cup A^*$ of cost $\cost_{\vG}(\vW_{\OPT})$ as follows: we replace every sub-walk of $\vW_{\OPT}$ that is contained in $\vG'$, has endpoints $u,v\in V(\vF)$, and visits some apex $w\in A^*$, by the path $u,w,v$ in $\vG_{A^*}$.
It is immediate that the resulting walk does not traverse any apices and therefore it is contained in $\vG_{A^*}$.
Thus we have
\begin{align*}
\cost_{\vG}(\vR) &\leq \cost_{\vG}(\vW'_{A^*}) = \cost_{\vG_{A^*}}(\vW_{A^*})\leq \cost_{\vG}(\vW_{\OPT}) \leq \OPT_{\vG},
\end{align*}
which concludes the proof.
\end{proof}

\section{The lower bound for graphs of bounded pathwidth}\label{sec:lower}

\makeatletter
\newcommand{\xRightarrow}[2][]{\ext@arrow 0359\Rightarrowfill@{#1}{#2}}
\makeatother

In this section we present the proof of Theorem \ref{thm:lower}.
This is done via the following chain of reductions:

\newcommand{\mybox}[1]
{\fbox{\parbox[][1.8cm][c]{3cm}{\begin{center}\small #1\end{center}}}}

\vspace{3mm}
\noindent
\resizebox{\linewidth}{!}{
{\fbox{\parbox[][1cm][c]{1.3cm}{\begin{center}\small\probclique\end{center}}}}
$\xRightarrow{\textup{folklore}}$
{\fbox{\parbox[][1cm][c]{2.5cm}{\begin{center}\small\probmcb\end{center}}}}
$\xRightarrow{\textup{Lemma~\ref{lem:balancedhard}}}$
{\fbox{\parbox[][1cm][c]{2cm}{\begin{center}\small\probbalanced\end{center}}}}
$\xRightarrow{\textup{Lemma~\ref{lem:nonrep-tw}}}$
{\fbox{\parbox[][1cm][c]{2.5cm}{\begin{center}\small\probnonrep\end{center}}}}
$\xRightarrow{\textup{Lemma~\ref{lem:walktoATSP}}}$
{\fbox{\parbox[][1cm][c]{1.3cm}{\begin{center}\small\probATSP\end{center}}}}
}
\vspace{3mm}

\subsection{\probbalanced}

Let $D$ be a directed graph and let $\chi:E(D)\to \mathbb{Z}^+$ be an
assignment of integers to the edges. We can extend $\chi$ to a set
$F\subseteq E(D)$ of edges in the obvious way by defining
$\chi(F)=\sum_{e\in F}\chi(e)$. Let $\delta_D^+(v)$ and $\delta_D^-(v)$ be
the set of outgoing and incoming edges of $v$, respectively.  We say that
$\chi$ is {\em balanced} at $v\in V(D)$ if
$\sum_{e\in \delta_D^+(v)}\chi(e)= \sum_{e\in \delta_D^-(v)}\chi(e)$
holds and we say that $\chi$ is {\em balanced} if it is balanced at
every $v\in V(D)$.

\begin{center}
\fbox{\parbox{0.9\linewidth}{
\probbalanced: Given a directed graph $D$ and a set $X_e$ of positive integers for every edge $e\in E(D)$, the task is to select a $\chi(e)\in X_e$ for every edge $e\in E(D)$ such that $\chi$ is balanced.}}
\end{center}

An easy counting argument shows that it is sufficient to require that $\chi$ balanced at all but one vertex:
\begin{proposition}\label{prop:balanced0}
If $\chi$ is balanced at every vertex of $V(D)\setminus \{v\}$, then it is also balanced at $v$.
\end{proposition}

We give a lower bound for \probbalanced\ by a parameterized reduction
from \probmcb.

\begin{center}
\fbox{\parbox{0.9\linewidth}{
\probmcb: Given a graph $G$ with a partition $(V_1,\dots,V_{2k})$ of
vertices, find a vertex $v_i\in V_i$ for every $1\le i \le 2k$ such that $v_{i_1}$ and $v_{i_2}$ are adjacent for every $1\le i_1 \le k$ and $k+1\le i_2 \le 2k$. }}
\end{center}

 There is a very simple folklore reduction from
\probclique\ to \probmcb\ (see, e.g., \cite{dell2012kernelization}) where the output parameter equals the input one,
hence the lower bound of Chen et
al.~\cite{chen2006strong} for \probclique\ can be transferred to \probmcb.

\begin{theorem}\label{th:mcb}
Assuming ETH, \probmcb\ cannot be solved in time $f(k)n^{o(k)}$ for any computable function $f$.
\end{theorem}

The reduction from \probmcb\ to \probbalanced\ uses the technique of
$k$-non-averaging sets to encode vertices
\cite{fomin2014almost,jansen2013bin}. We
say that a set $X$ of positive integers is {\em $k$-non-averaging} if for every
choice of $k$ (not necessarily distinct) integers
$x_1,\dots,x_k\in X$, their average is in $X$ only if
$x_1=x_2=\dots=x_k$. For example, $X=\{(k+1)^i\mid 1\le i \le n\}$ is
certainly a $k$-non-averaging set of size $n$. However, somewhat
surprisingly, it is possible to construct much denser $k$-non-averaging
sets where each integer is polynomially bounded in $k$ and the size $n$ of the set.

\begin{lemma}[Jansen \etal~\cite{jansen2013bin}]\label{lem:nonaveraging}
  For every $k$ and $n$ there exists a $k$-non-averaging set $X$ of
  $n$ positive integers such that the largest element of $X$ has value at most
  $32p^2n^2$. Furthermore, $X$ can be constructed in time $O(p^2n^3)$
  time.
\end{lemma}

\begin{lemma}\label{lem:balancedhard}
Assuming ETH, \probbalanced\ has no $f(k)n^{o(k)}$ time algorithm for any computable function $f$, where $k$ is the number of vertices of $D$.
\end{lemma}
\begin{proof}
  The proof is by reduction from \probmcb. Let $G$ be an undirected
  graph with a partition $V_1$, $\dots$, $V_{2k}$ of the vertices in
  $V(G)$. By padding the instance with isolated vertices, we may assume without loss of generality that each $V_i$ has
  the same number $n$ of vertices; let $v_{i,j}$ be the $j$-th vertex
  of $V_i$.  We construct an instance of \probbalanced\ on a directed graph $D$ having $2k+1$ vertices $w$, $w_1$,
  $w_2$, $\dots$, $w_{2k}$. Let us use Lemma~\ref{lem:nonaveraging} to construct a
  $k$-non-averaging set $X=\{x_1,\dots,x_n\}$ of $n$ positive integers such
  that the maximum value in $X$ is $M=O(k^2n^2)$. Let $B=2kM$.
\begin{figure}
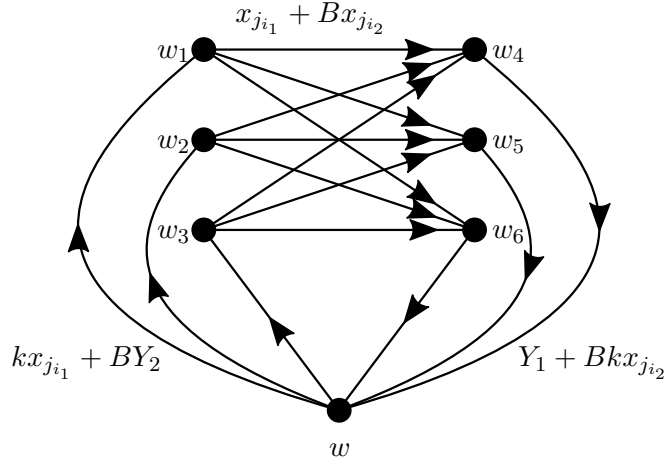

\begin{center}
{\svg{0.6\linewidth}{biclique}}
\caption{The instance of \probbalanced\ constructed in the proof of Lemma~\ref{lem:balancedhard}. The values on the edges indicate the value of $\chi$ corresponding to a solution $(v_{i,j_i})_{i=1,\dots,2k}$ of the \probmcb\ instance (we have $Y_1=\sum_{1\le i\le k}x_{j_i}$ and $Y_2=\sum_{k+1\le i\le 2k}x_{j_i}$).}\label{fig:biclique}
 \end{center}
\end{figure}
  The edges of $D$ and the sets of integers on them are constructed
  the following way (see Figure~\ref{fig:biclique}). For every $1\le i_1 \le k$ and $k+1\le i_2\le 2k$, we
  introduce the edge $(w_{i_1}, w_{i_2})$ into $D$ and we
  define the set $X_{(w_{i_1}, w_{i_2})}$ the following way: for every edge
  $v_{i_1,j_1}v_{i_2,j_2}$ between $V_{i_1}$ and $V_{i_2}$, let us
  introduce the positive integer $x_{j_1}+Bx_{j_2}$ into
  $X_{(w_{i_1}, w_{i_2})}$. Next, for every $1\le i \le k$, we introduce the
  edge $(w,w_{i})$ and let
  $X_{(w,w_i)}=\{ kx+By \mid x\in X, 1\le y \le kM\}$. Finally, for
  every $k+1 \le i \le 2k$, we introduce the edge $(w_{i},w)$ and let
  $X_{(w_i,w)}=\{ y+Bkx \mid x\in X, 1\le y \le kM\}$. This completes
  the description of the reduction.

\textbf{Biclique $\Rightarrow$ balanced assignment $\chi$.}
  Suppose that $v_{i,j_i}\in V_i$ for $1\le i \le 2k$ form a solution
  of \probmcb. We define $\chi:E(D)\to \mathbb{Z}^+$ the following way. 
For every $1\le i_1 \le k$ and $k+1\le i_2 \le 2k$, we set (see the values on the edges in Figure~\ref{fig:biclique})
\begin{itemize}
\item $\chi((w_{i_1},w_{i_2}))=x_{j_{i_1}}+Bx_{j_{i_2}}$,
\item $\chi((w,w_{i_1}))=kx_{j_{i_1}}+BY_2$, where $Y_2=\sum_{k+1\le i\le 2k}x_{j_i}$, and
\item $\chi((w_{i_1},w))=Y_1+Bkx_{j_{i_2}}$, where $Y_1=\sum_{1\le i\le k}x_{j_i}$.
\end{itemize}
Note that $\chi(e)\in X_e$ holds for every edge $e\in E(D)$.  Let us
verify that $\chi$ is balanced. For any $1\le i_1 \le k$, we have
\[ \chi(\delta_D^+(w_{i_1}))=\sum_{k+1\le i_2 \le 2k}(
x_{j_{i_1}}+Bx_{j_{i_2}})=kx_{j_{i_1}}+BY_2=\chi((w,w_{i_1}))=\chi(\delta_D^-(w_{i_1})),
\]
as required. Similarly, for 
for any $k+1\le i_2 \le 2k$, we have
\[ \chi(\delta_D^-(w_{i_2}))=\sum_{1\le i_1 \le k}(
x_{j_{i_1}}+Bx_{j_{i_2}})=Y_1+Bkx_{j_{i_2}}=\chi((w_{i_2},w))=\chi(\delta_D^+(w_{i_2})).
\]
Thus we have shown that $\chi$ is balanced at $w_1$, $\dots$, $w_{2k}$ and it follows by Proposition~\ref{prop:balanced0} that $\chi$ is balanced also at $w$.

\textbf{Balanced assignment $\chi$ $\Rightarrow$ biclique.}
For the reverse direction of the equivalence, suppose that
$\chi:E(D)\to \mathbb{Z}^+$ is a balanced assignment with $\chi(e)\in X_e$ for
every $e\in E(D)$. For every $1\le i_1 \le k$, the definition of
$X_{(w,w_{i_1})}$ implies that $\chi((w,w_{i_1}))$ is of the form
$kx_{j_{i_1}}+By_{i_1}$ where $x_{j_{i_1}}\in X$ for some $1\le j_{i_1}\le n$ and $y_{i_1}$ is a positive integer. As $kx_{j_{i_1}}\le kM < B$ follows from
$x_{j_{i_1}}\in X$, the value of $\chi((w,w_{i_1}))$ uniquely
determines $j_{i_1}$ and $y_{i_1}$. Similarly, for every
$k+1\le i_2 \le 2k$, we have that $\chi((w_{i_2},w))$ is of the form
$y_{i_2}+Bkx_{j_{i_2}}$ for uniquely determined positive integers
$1\le j_{i_2} \le n$ and $y_{i_2}$.

Having defined the values $j_1$, $\dots$, $j_{2k}$, we show that $\chi((w_{i_1},w_{i_2}))=x_{j_{i_1}}+Bx_{j_{i_2}}$ for every
$1\le i_1\le k$ and $k+1\le i_2\le 2k$. If this is true, then the
vertices $v_{i,j_{i}}\in V_i$ for $1\le i \le 2k$ form a solution of
\probmcb: the fact that $x_{j_{i_1}}+Bx_{j_{i_2}}$ was introduced into
$X_{(w_{i_1},w_{i_2})}$ implies that there is an edge between
$v_{i_1,j_{i_1}}\in V_{i_1}$ and $v_{i_2,j_{i_2}}\in V_{i_2}$.

As the balance requirement holds at vertex $w_{i_1}$, it also holds if
we count modulo $B$. We have that $\chi((w,w_{i_1}))$ modulo $B$ is
exactly $kx_{j_{i_1}}<B$ (here we use that $kM<B$). For every
$k+1 \le i_2 \le 2k$, the value $\chi((w_{i_1},w_{i_2}))$ modulo $B$
is an integer from $X$ and the value $\chi(\delta_D^+(w_{i_1}))$ modulo
$B$ is exactly the sum of these $k$ integers from $X$ (as again by $kM<B$, this
sum cannot reach $B$). Therefore, we have that the sum of $k$ integers
from $X$ is exactly $kx_{j_{i_1}}$. Since $X$ is a $k$-non-averaging
set, this is only possible if these $k$ integers are all equal to
$x_{j_{i_1}}$. Thus we have shown that
$\chi((w_{i_1},w_{i_2}))=x_{j_{i_1}}$ modulo $B$ for every
$1\le i_1\le k$ and $k+1\le i_2 \le 2k$.

Let us consider now a vertex $w_{i_2}$ for $k+1\le i_2 \le 2k$. The
balance requirement in particular implies that
$\lfloor \chi(\delta_D^+(w_{i_2}))/B\rfloor= \lfloor
\chi(\delta_D^-(w_{i_2}))/B\rfloor$.
First, we have
$\lfloor \chi(\delta_D^+(w_{i_2}))/B\rfloor=\lfloor
\chi((w_{i_2},w))/B\rfloor=kx_{j_{i_2}}$. Therefore, the fact that $\chi$ is balanced at $w_{i_2}$ implies 
\[
kx_{j_{i_2}}=\lfloor \chi(\delta_D^-(w_{i_2}))/B\rfloor=\left\lfloor\sum_{1\le i_1 \le k}\chi((w_{i_1},w_{i_2}))/B\right\rfloor=
\sum_{1\le i_1 \le k}\lfloor \chi((w_{i_1},w_{i_2}))/B\rfloor,
\]
where the third equality holds because we have $\chi((w_{i_1},w_{i_2}))/B-\lfloor \chi((w_{i_1},w_{i_2}))/B\rfloor\le  M/B<1/k$.
The definition of $X_{(w_{i_1},w_{i_2})}$ implies that
$\lfloor \chi((w_{i_1},w_{i_2}))/B\rfloor$ is an integer from
$X$. Therefore, the equation above states the the sum of $k$ integers
from $X$ is exactly $kx_{j_{i_2}}$. Since $X$ is a $k$-non-averaging
set, this is only possible if these $k$ integers are all equal to
$x_{j_{i_2}}$. Therefore, we have shown that $\chi((w_{i_1},w_{i_2}))$
modulo $B$ is exactly $x_{i_{j_1}}$ and
$\lfloor \chi((w_{i_1},w_{i_2}))/B\rfloor$ is exactly $x_{j_{i_2}}$,
proving that $\chi((w_{i_1},w_{i_2}))=x_{j_{i_1}}+Bx_{j_{i_2}}$ indeed
holds.
\end{proof}

\subsection{\probnonrep\ and \probATSP}

To make the hardness proof for \probATSP\ cleaner, we first prove hardness
for the variant of the problem, where instead of optimizing the length of
the tour, the only constraint is that certain vertices cannot be
visited more than once.

\begin{center}
\fbox{\parbox{0.9\linewidth}{
\probnonrep: Given an unweighted directed graph $G$ and set $U\subseteq V(G)$ of vertices, find a closed walk (of any length) that visits each vertex at least once and visits each vertex in $U$ exactly once.}}
\end{center}


There is a simple reduction from \probnonrep\ to \probATSP\ that preserves treewidth.

\begin{lemma}\label{lem:walktoATSP}
  An instance of \probnonrep\ on an unweighted directed graph $D$ can be reduced in
  polynomial time to an instance of \probATSP\ with polynomially
  bounded positive integer weights on an edge-weighted version $D^*$ of $D$.
\end{lemma}
\begin{proof}
It is easy to see that if we assign weight 1 to every edge $(u,v)$ with $v\in U$ and weight 0 to every other edge, 
then the \probnonrep\ instance has a solution if and only if the resulting weighted graph has closed walk of length at most $|U|$ (or, equiavelently, less than $|U|+1$) visiting every vertex.
To ensure that every weight is positive, let us replace every weight 0 with weight $\epsilon:=1/(2n^2)$. As a minimum solution of \probATSP\ contains at most $n^2$ edges, this modification increases the minimum cost by at most $1/2$. Thus it remains true that the \probnonrep\ instance has a solution if and only if there is a closed walk of length less than $|U|+1$ visiting every vertex. Finally, to ensure that every cost is integer, we multiply each of them by $2n^2$.
\end{proof}

The rest of the section is devoted to giving a lower bound for
\probnonrep.  The lower bound proof uses certain gadgets in the construction of the
instances. Formally, we define a {\em gadget} to be a graph with a set
of distinguished vertices called {\em external vertices;} every other
vertex is {\em internal.} To avoid degenerate situations, we always
require that the external vertices of a gadget are independent and
each external vertex has either indegree 0 or outdegree 0 in the
gadget; in particular, this implies that a path between two external
vertices contains no other external vertex. Also, this implies that
there is no closed walk containing an external vertex.

We say that a set $\cP$ of paths of the gadget {\em satisfies} a
gadget if (1) both endpoints of each path are external vertices and (2) every
internal vertex of the gadge tis visited by exactly one path in $\cP$. If a path
$P\in\cP$ connects two external vertices of a gadget, then we define
the {\em type} of $P$ to be the (ordered) pair of its endpoints.  If
$\cP$ satisfies the gadget, then we define the {\em type} of $\cP$ to
be the multiset of the types of the paths in $\cP$. For brevity, we
use notation such as $a\times (v_1,v_2)+ b\times (v_3,v_4)$ to denote
the type that contains $a$ times the pair $(v_1,v_2)$ and $b$ times
the pair $(v_3,v_4)$.  For a gadget $H$, we let the set $\typeset(H)$
contain every possible type of a set $\cP$ of paths satisfying $H$.

We construct gadgets where we can exactly tell the type of the
collections of paths that can satisfy the gadget, that is, the set
$\typeset(H)$ is of a certain form.  In the first gadget, we have a
simple choice between one path or a specified number of paths.

\begin{lemma}\label{lem:gadget1}
For every $s\ge 1$, we can construct in time polynomial in $n$ a gadget $H_s$ with the following properties:
\begin{enumerate}
\item $H_s$ has four external vertices $a_\textup{in}$, $a_\textup{out}$, $b_\textup{in}$, and $b_\textup{out}$.
\item $H_s$ minus its external vertices has constant pathwidth.
\item $\typeset(H_s)$ contains exactly two types: the type $(b_\textup{in},b_\textup{out})$ and the type  $s\times (a_\textup{in},a_\textup{out})$ (in other words, the gadget can be satisfied by a path from $b_\textup{in}$ to $b_\textup{out})$, can be satisfied by a collection of $s$ paths from $a_\textup{in}$ to $a_\textup{out})$, but cannot be satisfied by any other type of collection of paths).
\end{enumerate}
\end{lemma}
\begin{proof}
The gadget $H_s$ has $6s$ internal vertices $v^i_j$ ($1\le i \le 6$, $1\le j \le s$) connected as shown in Figure~\ref{fig:gadget1}(a). 
Additionally, we introduce the edges $(b_\textup{in},v^1_1)$, $(v^6_s,b_\textup{out}$, and for every $1\le j \le s$, the edges $(a_\textup{in},v^3_j)$ and $(v^4_j,a_\textup{out})$. It is clear that statement (2) holds: $H_s$ minus its vertices is a graph with constant pathwidth.
\begin{figure}
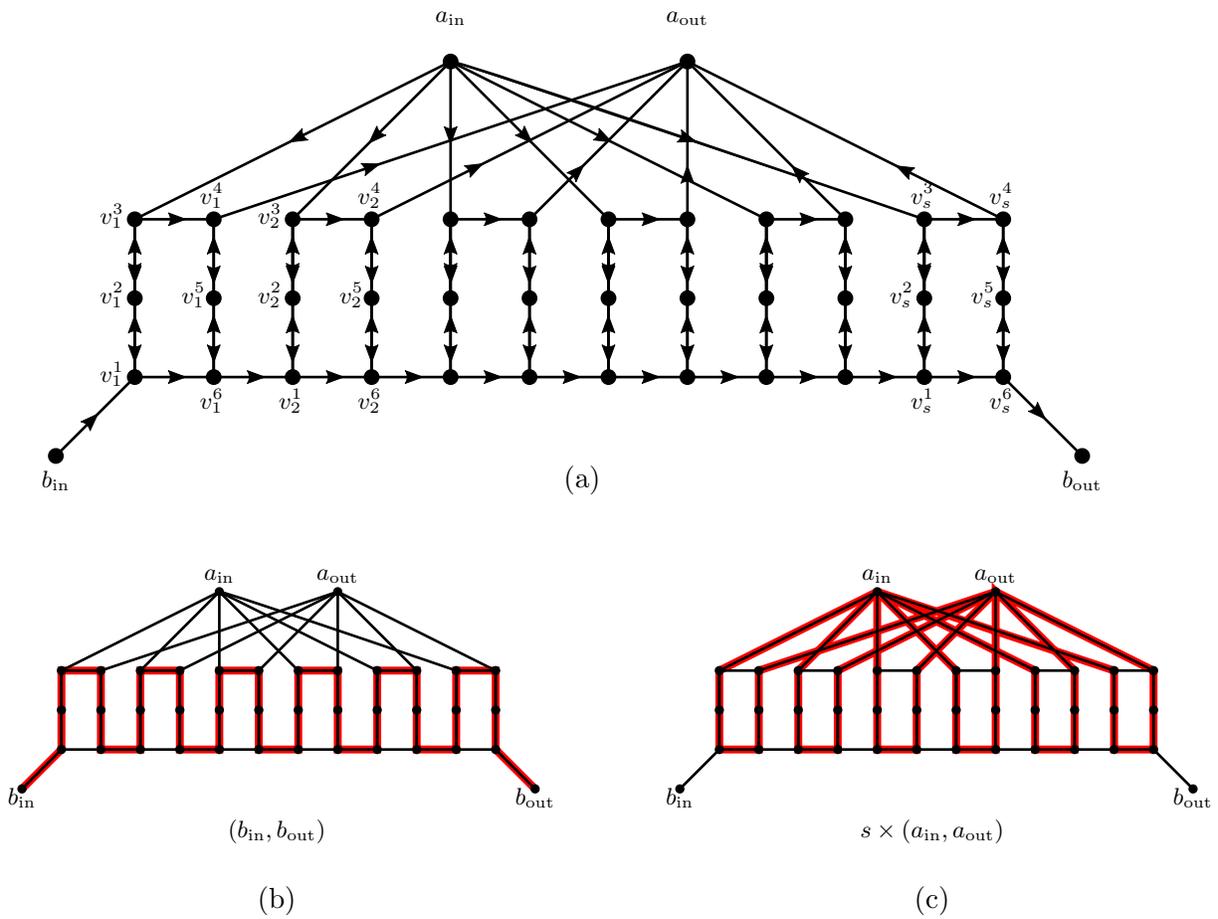

\begin{center}
{\footnotesize \svg{\linewidth}{gadget1}}
\caption{The gadget of Lemma~\ref{lem:gadget1} with two collections of paths satisfying it. }\label{fig:gadget1}
\end{center}
\end{figure}

This gadget can be satisfied by a path from $b_\textup{in}$ to
$b_\textup{out}$ (see Figure~\ref{fig:gadget1}(b)) and also by a collection of $s$ paths where the
$j$-th path is $a_\textup{in}$, $v^3_j$, $v^2_j$, $v^1_j$, $v^6_j$,
$v^5_j$, $v^4_j$, $a_\textup{out}$ (see Figure~\ref{fig:gadget1}(c)). To complete the proof of statement
(3), we need to show that if $\cP$ satisfies $H_s$, then $\cP$ is one of
these two types. The basic observation is that if a path in $\cP$
contains $v^2_j$, then it has to contain $v^1_j$ and $v^3_j$ as well
(as each internal vertex is visited exactly once), hence the three
vertices $v^1_j$, $v^2_j$, $v^3_j$ have to appear on the same
path of $\cP$. The same is true for the vertices $v^4_j$, $v^5_j$, $v^6_j$. 
 Suppose that $\cP$ contains a path $P$ starting at
$b_\textup{in}$. Then its next vertex is $v^1_1$, which should be
followed by $v^2_1$ and $v^3_1$ by the argument above. The next vertex is $v^4_1$ (the only outneighbor of $v^3_j$ not yet visited), which is followed by $v^5_1$ and $v^6_1$. Now the next vertex is $v^1_2$, the only outneighbor of $v^6_1$ not yet visited. 
With similar arguments, we
can show that $P$ is exactly of the form shown in Figure~\ref{fig:gadget1}(b), hence $\cP$
contains only this path, and $\cP$ is of type
$\{(b_\textup{in},b_\textup{out})\}$. 

Suppose now that $\cP$ does not contain a path starting at
$b_\textup{in}$. Then the only way to reach vertex $v^1_1$ is with a
path starting as $a_\textup{in}$, $v^3_1$, $v^2_1$, $v^1_1$. This has
to be followed by the unique outneighbor $v^6_1$ of $v^1_1$ that was
not yet visited. This means that the path contains also $v^5_1$ and
$v^4_1$, which has to be followed by $a_\textup{out}$. Then with
similar arguments, we can show for every $j\ge 2$ that $v^1_j$ is
visited by the path $a_\textup{in}$, $v^3_j$, $v^2_j$, $v^1_j$,
$v^6_j$, $v^5_j$, $v^4_j$, $a_\textup{out}$. This means that $|\cP|=s$
and the type of $\cP$ is $s\times (a_\textup{in},a_\textup{out})$.
\end{proof}

\begin{lemma}\label{lem:gadget-vertex}
  Let $X$ be a set of $n$ positive integers, each at most $M$ and let $S=\sum_{x\in X}x$. In time
  polynomial in $n$ and $M$, we can construct a gadget $H_X$ with the
  following properties:
\begin{enumerate}
\item $H_X$ has four external vertices $a_\textup{in}$, $a_\textup{out}$, $c_\textup{in}$, and $c_\textup{out}$.
\item $H_X$ minus its external vertices has constant pathwidth.
\item $\typeset(H_s)$ contains exactly $|X|$ types: for every
  $x\in X$, it contains the type
  $(c_\textup{in},c_\textup{out})+ (S-x)\times
  (a_\textup{in},a_\textup{out})$.
\end{enumerate}
\end{lemma}
\begin{proof}
  The gadget $H_X$ is constructed the following way (see Figure~\ref{fig:gadget2}).  Let us introduce
  an internal vertex $v$ and the edge $(c_\textup{in},v)$. For every
  $x\in X$, let us introduce a copy of $H_x$ defined by
  Lemma~\ref{lem:gadget1} where $a_\textup{in}$, $a_\textup{out}$,
  $v$, $c_\textup{out}$ of $H_X$ play the role of $a_\textup{in}$,
  $a_\textup{in}$, $b_\textup{in}$, $b_\textup{out}$, respectively. If
  we remove the four external vertices of $H_X$, then we get a graph
  with constant pathwidth: if we remove one more vertex, $v$, then we
  get the disjoint union of internal vertices of the gadgets $H_x$'s,
  which have constant pathwidth by Lemma~\ref{lem:gadget1}.
\begin{figure}
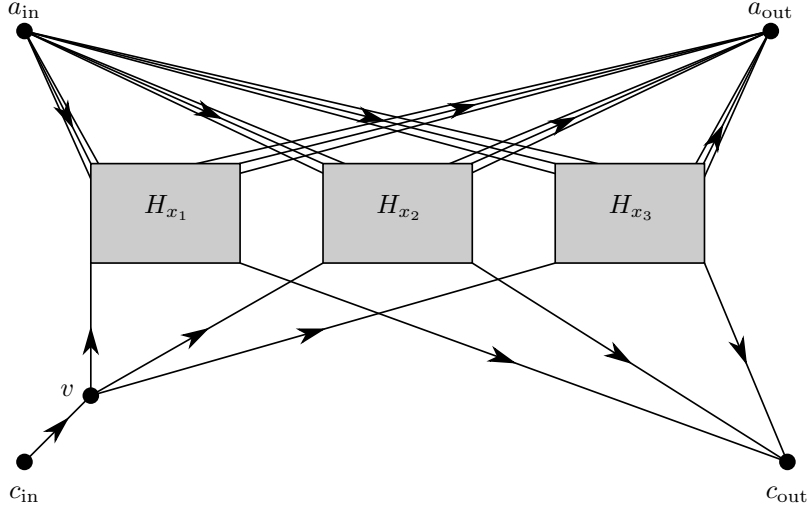

\begin{center}
{\small \svg{0.7\linewidth}{gadget2}}
\caption{The gadget $H_X$ of Lemma~\ref{lem:gadget-vertex} for a set $X=\{x_1,x_2,x_3\}$ of three integers. The gray rectangles represent the {\em internal} vertices of the three gadgets $H_{x_1}$, $H_{x_2}$, and $H_{x_3}$.}\label{fig:gadget2} 
\end{center}
\end{figure}

  For every $x\in X$, the gadget can be satisfied by the following
  collection of paths. The copy of $H_x$ in $H_X$ can be satisfied by
  a path from $v$ to $c_\textup{out}$, which can be extended with the
  edge $(c_\textup{in}, v)$ to a path from $c_\textup{in}$ to
  $c_\textup{out}$. For every $x'\in X$, $x'\neq x$, we can satisfy
  the copy of $H_{x'}$ in $H_X$ by a collection of $x'$ paths from
  $a_\textup{in}$ to $a_\textup{out}$. This way, we constructed a
  collection $\cP$ of paths satisfying $H_X$ that consists of a single
  path of type $(c_\textup{in},c_\textup{out})$ and exactly
  $\sum_{x'\in X\setminus \{x\}} x'=S-x$ paths of type
  $(a_\textup{in},a_\textup{out})$.

  To complete the proof of statement (3), consider a collection $\cP$
  of paths satisfying $H_X$. Let $P$ be the unique path of $\cP$
  visiting vertex $v$. The vertex of $P$ after $v$ is has to be an
  internal vertex of the copy of $H_x$ for some $x\in X$ (here we use
  that the external vertices of the gadget $H_x$ are independent,
  hence $v$ cannot be followed by any of $a_\textup{in}$,
  $a_\textup{out}$, and $c_\textup{out}$). As $v$ was identified with
  vertex $b_\textup{in}$ of $H_x$, Lemma~\ref{lem:gadget1} implies
  that $P$ visits every internal vertex of this copy of $H_x$ and
  leaves $H_x$ at its vertex $b_\textup{out}$, which was identified
  with $c_\textup{out}$. Consider now some $x'\in X$ with $x'\neq
  x$.
  Vertex $b_\textup{in}$ of $H_{x'}$ was identified with $v$, path $P$
  is the only path of $\cP$ visiting $v$, and $P$ does not visit any
  internal vertex of $H_{x'}$. Therefore, by Lemma~\ref{lem:gadget1},
  the internal vertices of $H_{x'}$ are visited by exactly $x'$ paths
  of type $(a_\textup{in},a_\textup{out})$. Thus $\cP$ contains one
  path of type $(c_\textup{in},c_\textup{out})$ and exactly
  $\sum_{x'\in X\setminus \{x\}} x'=S-x$ paths of type
  $(a_\textup{in},a_\textup{out})$.
\end{proof}

\begin{lemma}\label{lem:nonrep-tw}
Assuming ETH, there is no $f(p)n^{o(p)}$ time algorithm for \probnonrep\ on graphs of pathwidth at most $p$ for any computable function $f$.
\end{lemma}
\begin{proof}
  The proof is by reduction from \probbalanced\ on a directed graph
  $D$ with $k$ vertices $w_1$, $\dots$, $w_k$. We construct a
  \probnonrep\ instance on a directed graph $D^*$ the following
  way. First, let us introduce the vertices $w_1$, $\dots$, $w_k$ into
  $D^*$, as well as two auxiliary vertices $c_\textup{in}$ and
  $c_\textup{out}$. For every edge $e=(w_{i_1},w_{i_2})\in E(D)$ with
  a set $X_e$ of integers associated to it in the \probbalanced\
  instance, we construct a copy of the gadget $H_{X_e}$ defined by
  Lemma~\ref{lem:gadget-vertex} and identify external vertices
  $a_\textup{in}$, $a_\textup{out}$, $c_\textup{in}$, $c_\textup{out}$
  of the gadget $H_{X_e}$ with vertices $w_{i_1}$, $w_{i_2}$,
  $c_\textup{in}$, $c_\textup{out}$ of $D^*$, respectively. Let
  $S_e=\sum_{x\in X_e}x$ for every edge $e\in V(D)$, let
  $S^+_i=\sum_{e\in \delta_D^+(w_i)}S_s$, let
  $S^-_i=\sum_{e\in \delta_D^-(w_i)}S_s$, and let
  $S^*=\sum_{e\in E(D)}X_e=\sum_{i=1}^kS^+_i=\sum_{i=1}^kS^-_i$. We
  further extend $D^*$ the following way.
\begin{enumerate}
\item For every $1\le i \le k$, we introduce a set $\cP^+_i$ of $S^+_i$ paths of length two from $c_\textup{in}$ to $w_i$ (that is, each of these paths consists of vertex $c_\textup{in}$, vertex $w_i$, and one extra newly introduced vertex).
\item For every $1\le i \le k$, we introduce a set $\cP^-_i$ of $S^-_i$ paths of length two from $w_i$ to $c_\textup{out}$.
\item We introduce a set $\cP^*$ of $S^*+|E(D)|$ paths of length two from $c_\textup{out}$ to $c_\textup{in}$.
\end{enumerate}
Let $Z:=\{w_1, \dots, w_k, c_\textup{in}, c_\textup{out}\}$; note that
$Z$ form an independent set in $G^*$ (as the external vertices of each
gadget are independent). We define $U:=V(D^*)\setminus Z$ to be the
set of vertices that have to be visited exactly once.  This completes
the description of the reduction.

Observe that if we remove $Z$ from $D^*$, then what remains is the
disjoint union of the internal vertices of the gadgets $H_{X_e}$,
which have constant pathwidth by Lemma~\ref{lem:gadget-vertex}. As
removing a vertex can decrease pathwidth at most by one, it follows
that $D^*$ has pathwidth $|Z|+O(1)=O(k)$. Thus if we are able to show
that the constructed instance $D^*$ of \probnonrep\ is a yes-instance
if and only if $D$ is a yes-instance of \probbalanced, then this
implies that an $f(p)n^{o(p)}$ time algorithm for \probnonrep\ on
graphs of pathwidth $p$ can be used to solve \probbalanced on $k$
vertex graphs in time $(k)n^{o(k)}$, which would contradict ETH by
Lemma~\ref{lem:balancedhard}.

\textbf{Balanced assignment $\chi$ $\Rightarrow$ closed walk.}  Suppose that
balanced assignment $\chi:E(D)\to \mathbb{Z}^+$ is a solution to the
\probbalanced\ instance. For every $e=(w_{i_1},w_{i_2})\in E(D)$, the
construction of the gadget $H_{X_e}$ implies that $H_{X_e}$ can be
satisfied by a collection $\cP_e$ of paths having type
$(c_\textup{in},c_\textup{out})+(S_e-\chi(e))\times
(w_{i_1},w_{i_2})$.
Let $\cP$ be a collection of paths that is the union of the set
$\cP^*$, the sets $\cP^+_i$ and $\cP^-_i$ for $1\le i \le k$, and the
set $\cP_e$ for $e\in E(G)$. Observe that every vertex of $U$ is
contained in exactly one path in $\cP$ and the paths in $\cP$ are edge
disjoint. Let $H^*$ be the subgraph of $D^*$ formed by the union of
every path in $\cP$. It is easy to see that $H^*$ is connected: every
path in $\cP$ has endpoints in $Z$ and the paths in $\cP^*$,
$\cP^-_i$, $\cP^+_i$ ensure that every vertex of $Z$ is in the same
component of $H^*$. It is also clear that every vertex of $U$ has
indegree and outdegree exactly 1, as each vertex in $U$ is visited by
exactly one path in $\cP$. We show below that every vertex of $Z$ is
balanced in $H^*$ (its indegree equals its outdegree). If this is
true, then $H^*$ has a closed Eulerian walk, which gives a closed walk in
$G^*$ visiting every vertex at least once and every vertex in $U$
exactly once, what we had to show.

The endpoints of every path in $\cP$ are in $Z$, hence every vertex of $U$ is balanced in $H^*$ (in particular has indegree and outdegree exactly 1). Consider now a vertex $w_i$.
\begin{itemize}
\item For every $e\in \delta_D^+(w_i)$, the set $\cP_e$ contains $S_e-\chi(e)$ paths starting at $w_i$.
\item For every $e\in \delta_D^-(w_i)$, the set $\cP_e$ contains $S_e-\chi(e)$ paths ending at $w_i$. 
\item The set $\cP^+_i$ contains $S^+_i$ paths ending at $w_i$.
\item The set $\cP^-_i$ contains $S^-_i$ paths starting at $w_i$.
\end{itemize}
As these paths are edge disjoint, the difference between the outdegree and the indegree of $w_i$ in $H^*$ is
\[
(S^-_i+\sum_{e\in \delta_D^+(w_i)} (S_e-\chi(e)))-
(S^+_i+\sum_{e\in \delta_D^-(w_i)} (S_e-\chi(e)))=
(S^-_i+S^+_i-\chi(\delta_D^+(w_i)))-
(S^+_i+S^-_i-\chi(\delta_D^-(w_i)))=0,
\] 
since $\chi$ is balanced at $w_i$. Consider now vertex $c_\textup{in}$.
\begin{itemize}
\item For every $1\le i \le k$, the set $\cP^+_i$ contains $S^+_i$ paths starting at $c_\textup{in}$.
\item For every $e\in E(D)$, the set $\cP_e$ contains one path starting at $c_\textup{in}$.
\item The set $\cP^*$ contains $S^*+|E(D)|$ paths ending at $c_\textup{in}$.
\end{itemize}
It follows that $c_\textup{in}$ is balanced in $H^*$ with indegree and
outdegree exactly $S^*+|E(D)|=\sum_{i=1}^kS^+_i+|E(D)|$ and a similar
argument shows the same for $c_\textup{out}$. Thus we have shown that
the \probnonrep\ instance has a solution.

\textbf{Closed walk $\Rightarrow$ balanced assignment $\chi$.}  For
the reverse direction, suppose that the constructed \probnonrep\
instance has a solution (a closed walk $W$). The closed walk can be
split into a collection $\cP$ of walks with endpoints in $Z$ and every
internal vertex in $U$. In fact, these walks are paths: (1) as each
vertex of $U$ is visited only once, the internal vertices of each walk
are distinct, (2) the walk cannot be a cycle, since we have stated earlier
that no gadget has a cycle through an external vertex.  When defining
the sets $\cP^*$, $\cP^+_i$, $\cP^-_i$, we introduced a large number
of vertices into $D^*$ with indegree and outdegree 1. The fact that
these vertices are visited implies that $\cP$ has to contain the set
$\cP^*$ and the sets $\cP^+_i$ and $\cP^-_i$ for every $1\le i \le k$.
Moreover, every path of $\cP$ not in these sets contains an internal
vertex of some gadget $H_{X_e}$ (here we use that $Z$ is independent)
and a path of $\cP$ cannot contain the internal vertices of two
gadgets (as this would imply that it has an internal vertex in
$Z$). Therefore, the remaining paths can be partitioned into sets
$\cP_e$ for $e\in E(D)$ such that the internal vertices of $H_{X_e}$
are used only by the paths in $\cP_e$. This means that the set $\cP_e$
satisfies gadget $H_{X_e}$. If $e=(w_{i_1},w_{i_2})$, then it follows
by Lemma~\ref{lem:gadget-vertex} that $\cP_e$ has type
$(c_\textup{in},c_\textup{out})+(S_e-\chi(e))(w_{i_1},w_{i_2})$ for
some integer $\chi(e)\in X_e$. In particular, this means that $\cP_e$
contains $S_e-\chi(e)$ paths starting at $w_{i_1}$ and the same number
of paths ending at $w_{i_2}$.

 We claim that $\chi$ form a solution of the
\probbalanced\ problem.
Consider a vertex $w_i$. Taking into account the contribution of the paths in $\cP^-_i$ and $\cP_e$ for $e\in \delta_D^+(w_i)$, we have that the outdegree of $w_i$ in the walk $W$ is exactly
\[
S^-_i+\sum_{e\in \delta_D^+(w_i)}(S_e-\chi(e))=S^+_i+S^-_i-\chi(\delta_D^+(w_i)).
\]
Taking into account the contribution of the paths in $\cP^+_i$ and $\cP_e$ for $e\in \delta_D^-(w_i)$, we have that the indegree of $w_i$ in the walk $W$ is exactly
\[
S^+_i+\sum_{e\in \delta_D^-(w_i)}(S_e-\chi(e))=S^-_i+S^+_i-\chi(\delta_D^-(w_i)).
\]
As the indegree of $w_i$ in $W$ is clearly the same as its outdegree,
these two values have to be equal. This is only possible if
$\chi(\delta_D^+(w_i))=\chi(\delta_D^-(w_i))$, that is $\chi$ is
balanced at $w_i$. As this is true for every $1\le i \le k$, it
follows that the \probbalanced\ instance has a solution.
\end{proof}

\fi


\bibliography{bibfile}

\end{document}